\newif\ifShowComments

\ShowCommentsfalse

\documentclass[sigconf,nonacm]{acmart}
\makeatletter
\def\@ACM@checkaffil{%
    \if@ACM@instpresent\else
    \ClassWarningNoLine{\@classname}{No institution present for an affiliation}%
    \fi
    \if@ACM@citypresent\else
    \ClassWarningNoLine{\@classname}{No city present for an affiliation}%
    \fi
    \if@ACM@countrypresent\else
        \ClassWarningNoLine{\@classname}{No country present for an affiliation}%
    \fi
}

\makeatother
\usepackage{xcolor}
\usepackage{graphicx}
\usepackage[utf8]{inputenc}
\usepackage{amsmath}
\usepackage{amsfonts}
\usepackage{mathtools}
\usepackage{tabularx}
\usepackage{booktabs}
\usepackage{xspace}
\usepackage{boxedminipage}
\usepackage[inline]{enumitem}
\usepackage{amsthm}
\usepackage{multirow}
\usepackage{multicol}
\usepackage{pifont}
\usepackage{makecell}
\usepackage{adjustbox}
\usepackage{hyperref}
\usepackage{algorithm}
\usepackage{thmtools}
\usepackage{color}
\usepackage[noend]{algpseudocode}
\usepackage[capitalize,noabbrev]{cleveref}
\usepackage[font=small,skip=1pt, belowskip=-1pt]{caption}
\usepackage{subcaption}
\usepackage{float}
\usepackage{soul}
\usepackage{threeparttable}
\usepackage{listings}
\usepackage{pdfpages}
\lstdefinestyle{ivyspec}{
  basicstyle=\normalsize\ttfamily,
  columns=fullflexible,
  keepspaces=true,
  breaklines=true,
  breakatwhitespace=true,
  showstringspaces=false,
  aboveskip=0pt,
  belowskip=0pt
}
\renewcommand\footnotetextcopyrightpermission[1]{} 
\ifShowComments
\newcommand{\kartik}[1]{{\footnotesize{\color{orange} [Kartik: #1]}}}
\newcommand{\aniket}[1]{{\footnotesize{\color{teal} [Aniket: #1]}}}
\newcommand{\nibesh}[1]{{\footnotesize{\color{blue} [Nibesh:  #1]}}}
\newcommand{\qianyu}[1]{{\footnotesize{\color{magenta} [Qianyu: #1]}}}
\newcommand{\giuliano}[1]{{\footnotesize{\color{violet} [Giuliano: #1]}}}
\newcommand{\xuechao}[1]{{\footnotesize{\color{red} [Xuechao: #1]}}}
\else
\newcommand{\kartik}[1]{}
\newcommand{\aniket}[1]{}
\newcommand{\nibesh}[1]{}
\newcommand{\qianyu}[1]{}
\newcommand{\giuliano}[1]{}
\newcommand{\xuechao}[1]{}
\fi

\newcommand{\ignore}[1]{}

\newcommand{\nodes}{\mathcal{P}}
\newcommand{\node}[1]{\ensuremath{P_{#1}}\xspace}

\newcommand{\tuple}[1]{\langle #1 \rangle}

\newcommand{\True}{\mathsf{true}}
\newcommand{\False}{\mathsf{false}}

\algrenewcommand\textproc{}
\newcommand{\N}{\mathbb{N}}

\newcommand{\HN}{\mathcal{H}}
\newcommand{\BigS}{\mathcal{S}}

\definecolor{yescolor}{HTML}{026378}

\newcommand{\name}{Sailfish++\xspace}

\makeatletter
\renewcommand{\ALG@name}{}
\makeatother

\algrenewcommand\textproc{}
\algdef{SE}[EVENT]{Event}{EndEvent}[1]{\textbf{upon}\ #1\ \algorithmicdo}{\algorithmicend\ \textbf{event}}%
\algtext*{EndEvent}

\newcommand{\Timeout}{\texttt{timeout}}

\newcommand{\iunderline}[1]{\noindent\underline{#1}}
\newcommand{\mypara}[1]{\medskip\noindent\textbf{#1}}
\newcommand{\floor}[1]{\lfloor #1 \rfloor}
\newcommand{\ceil}[1]{\lceil #1 \rceil}

\newtheorem{claim}{Claim}
\newtheorem{lemma}{Lemma}
\newtheorem{fact}{Fact}
\newtheorem{property}{Property}
\newtheorem{corollary}{Corollary}
\newtheorem{theorem}{Theorem}
\newtheorem{rrule}{Commit rule}

\newcommand{\sig}[1]{\langle #1 \rangle}

\newcommand{\Propose}{\mathsf{propose}}
\newcommand{\Echo}{\mathsf{echo}}
\newcommand{\Vote}{\mathsf{vote}}
\newcommand{\Ready}{\mathsf{ready}}

\newenvironment{claimproof}[1][\proofname]
{\proof[\iunderline{Proof}]}
{\endproof}

\def\k{\ensuremath{\kappa}\xspace} 
\newcommand{\Enc}{\mathsf{ENC}}

\newcommand{\Retrieve}{\mathsf{retrieve}}
\newcommand{\Disperse}{\mathsf{disperse}}
\newcommand{\RETRIEVE}{\mathsf{Retrieve}}
\newcommand{\DISPERSE}{\mathsf{Disperse}}
\newcommand{\Symbol}{\mathsf{symbol}}

\newcommand{\F}{\mathbb{F}}
\newcommand{\C}{\mathbb{C}}
\newcommand{\Ack}{\mathsf{ack}}
\newcommand{\VerifyPoly}{\mathsf{verify\text{-}poly}}

\newcommand{\VerifyPoint}{\mathsf{verify\text{-}point}}
\newcommand{\VerifyShare}{\mathsf{verify\text{-}share}}
\newcommand{\Sh}{\mathsf{Sh}}
\newcommand{\Rec}{\mathsf{Rec}}
\newcommand{\negl}{\mathsf{negl}}
\newcommand{\Deal}{\mathsf{deal}}
\newcommand{\Reconstruct}{\mathsf{reconstruct\text{-}share}}

\newif\iffull
\fulltrue
\newcommand{\full}[2]{\iffull#1\else#2\fi}

\algdef{SE}[EVENT]{Upon}{EndUpon}[1]{\textbf{upon} #1 \textbf{do}}{\textbf{end upon}}
\algtext*{EndUpon}

\title{Optimistic, Signature-Free Reliable Broadcast and Its Applications}

\author{Nibesh Shrestha}
\authornote{Contributed equally and listed alphabetically. Correspondence: Xuechao Wang.\\ This version fixes a termination issue in the previous optimistic RBC design.}
\affiliation{
\institution{
Supra Research
}
\city{Rochester}
\country{USA}
}
\email{n.shrestha@supra.com}
\author{Qianyu Yu}
\authornotemark[1]
\affiliation{
\institution{
The Hong Kong University of Science and Technology (Guangzhou)
}
\city{Guangzhou}
\country{China}
}
\email{qyu100@connect.hkust-gz.edu.cn}
\author{Aniket Kate}
\affiliation{
\institution{
Purdue University/Supra Research}
\city{West Lafayette}
\country{USA}
}
\email{aniket@purdue.edu}
\author{Giuliano Losa}
\affiliation{
\institution{
Stellar Development Foundation
}
\city{San Francisco}
\country{USA}
}
\email{giuliano@stellar.org}
\author{Kartik Nayak}
\affiliation{
\institution{
Duke University
}
\city{Durham}
\country{USA}
}
\email{kartik@cs.duke.edu}
\author{Xuechao Wang}
\affiliation{
\institution{
The Hong Kong University of Science and Technology (Guangzhou)
}
\city{Guangzhou}
\country{China}
}
\email{xuechaowang@hkust-gz.edu.cn}

\begin{document}
\begin{abstract}
Reliable broadcast (RBC) is a key primitive in fault-tolerant distributed systems, and improving its efficiency can benefit a wide range of applications. This work focuses on signature-free RBC protocols, which are particularly attractive due to their computational efficiency. Existing protocols in this setting incur an optimal 3 steps to reach a decision while tolerating up to $f < n/3$ Byzantine faults, where $n$ is the number of parties. In this work, we propose an optimistic RBC protocol that maintains the $f < n/3$ fault tolerance but achieves termination in just 2 steps under certain optimistic conditions—when at least $\lceil{\frac{n+2f-2}{2}}\rceil$ non-broadcaster parties behave honestly. We also prove a matching lower bound on the number of honest parties required for 2-step termination.

We show that our latency-reduction technique generalizes beyond RBC and applies to other primitives such as asynchronous verifiable secret sharing (AVSS) and asynchronous verifiable information dispersal (AVID), enabling them to complete in 2 steps under similar optimistic conditions.

To highlight the practical impact of our RBC protocol, we integrate it into a new signature-free, post-quantum secure DAG-based Byzantine fault-tolerant (BFT) consensus protocol. Under optimistic conditions, this protocol achieves a commit latency of 3 steps—matching the performance of the best signature-based protocols. Our experimental evaluation shows that our protocol significantly outperforms existing post-quantum secure and signature-based protocols, even on machines with limited CPU resources. In contrast, signature-based protocols require high CPU capacity to achieve comparable performance. 
\end{abstract}



\maketitle
\section{Introduction}
Reliable broadcast (RBC) is a fundamental primitive used in fault-tolerant distributed computing. It ensures that among $n$ parties with up to $f$ Byzantine faults, either all honest parties deliver the same message or none do. If the broadcaster is honest, all honest parties eventually deliver its message. RBC serves as a crucial building block for several distributed computing primitives, including state-machine replication (SMR)~\cite{miller2016honey,bracha1987asynchronous,spiegelman2022bullshark,shrestha2024sailfish}, asynchronous verifiable information dispersal (AVID)~\cite{cachin2005asynchronous,alhaddad2022asynchronous,Alhaddad_2022,yang2022dispersedledger}, and asynchronous verifiable secret sharing (AVSS)~\cite{das2021asynchronous,cachin2002asynchronous,das2023verifiable}. As such, the performance and security of these protocols are closely tied to that of RBC.

In this work, we focus on the signature-free setting. Signature-free protocols are particularly appealing for their computational efficiency, scalability, and flexible security.
Existing protocols often rely on digital signatures to ensure message authenticity, which introduces significant computational overhead, particularly in large-scale or resource-constrained environments~\cite{li2023performance}. Signature-free RBC protocols can improve scalability by eliminating the computation bottlenecks associated with signature verification~\cite{cheng2024jumbo}. Moreover, these protocols provide a simpler and more efficient trust model, which is especially valuable for trustless cross-chain synchronization protocols like TrustBoost~\cite{sheng2023trustboost} and the Interchain timestamping protocol~\cite{tas2023interchain}. In such settings, public-key cryptography is often prohibitively expensive and difficult to manage, making signature-free protocols a more practical solution. Motivated by this, we focus on signature-free protocols, which have been shown to be viable for a wide range of classic BFT problems~\cite{yu2024tetrabft,das2024asynchronous,castro1999practical,bracha1987asynchronous,alhaddad2022asynchronous,Alhaddad_2022}.

In the signature-free setting, Bracha's RBC~\cite{bracha1987asynchronous} is the best-known protocol tolerating optimal $f < n/3$ Byzantine faults, with a good-case latency\footnote{Informally, good-case latency refers to the number of steps required to reach a decision when the broadcaster and at least $n-f-1$ non-broadcaster parties behave honestly.} of 3 steps and a bad-case latency\footnote{Informally, bad-case latency refers to the worst-case number of steps required to terminate, given that at least one honest party terminates.} of 4 steps. A lower bound result~\cite{abraham2022good} confirms that the good-case latency of Bracha's RBC is indeed optimal when tolerating $f<n/3$ Byzantine faults. A natural follow-up question is whether faster termination is possible under some optimistic conditions when more than $n-f$ parties behave honestly. This question is particularly relevant for applications like blockchain, which often operate in environments where most nodes are honest and failures are rare.

\mypara{Optimistic RBC to improve latency in the signature-free setting.} To answer the above question, we introduce \emph{optimistic reliable broadcast} that tolerates $f<n/3$ Byzantine faults and terminates in 2 steps when the broadcaster and at least $\ceil{\frac{n+2f-2}{2}}$ non-broadcaster parties behave honestly. We define the latency required for RBC under such conditions as \emph{optimistic good-case latency}.
Under $f$ Byzantine faults, our protocol achieves a good-case latency of
$4$ steps and a bad-case latency of $5$ steps, incurring one additional step
compared to Bracha’s RBC~\cite{bracha1987asynchronous}. This protocol is also information-theoretically secure. In particular, we establish the following result:

\begin{theorem}
    There exists a reliable broadcast protocol that tolerates $f$ Byzantine faults, under $n \ge 3f+1$, and such that, when the broadcaster is honest, the protocol achieves an optimistic good-case latency of $2$ steps if $\ceil{\frac{n+2f-2}{2}}$ non-broadcaster parties behave honestly, a good-case latency of $4$ steps, and a bad case latency of $5$ steps.
\end{theorem}

\mypara{A matching lower bound on the resilience for optimistic good-case latency.}
It is worth noting that Bracha's RBC can also terminate in two steps, assuming all nodes are honest. However, this assumption holds little value in a fault-tolerant distributed system. Therefore, we also establish a tight lower bound on the resilience for achieving optimistic good-case latency. Specifically, we prove that if more than $\floor{\frac{n-2f}{2}}$ non-broadcaster parties behave maliciously, no protocol can achieve a latency of $2$ steps when tolerating $f$ faults under $n \le 4f-1$\footnote{A good-case latency of $2$ steps is achievable when $n \ge 4f$~\cite{abraham2022good}.}. Our lower bound result holds even in a synchronous network model, further strengthening our findings. In particular, we show the following result:

\begin{theorem}
There does not exist a reliable broadcast protocol tolerating $f$ Byzantine faults under $n \leq 4f-1$ that can achieve an optimistic good-case latency of $2$ steps when more than $\floor{\frac{n-2f}{2}}$ non-broadcaster parties behave maliciously, even under synchrony.
\end{theorem}

\begin{table}[ht]
    \footnotesize
\centering
    \caption{\textbf{Comparison of the protocols proposed in this paper with the best-known existing protocols.}}
    \def\arraystretch{1}
    \setlength\tabcolsep{1mm}
    \begin{center}
    \begin{tabularx}{\columnwidth}{c c c c c c c}
    \toprule
     \textbf{Scheme} &
     \multirow{2}{*}{\makecell{\textbf{Good-case } \\ \textbf{Latency}}} &
     \multirow{2}{*}{\makecell{\textbf{Communication } \\ \textbf{Complexity}}} &
     \multirow{2}{*}{\makecell{\textbf{Cryptographic } \\ \textbf{Assumption}}} & 
     \textbf{Reference} &
     \\
     \\
    \midrule
    \multirow{2}{*}{RBC} & 3 & $O(n^2)$ & None &  \cite{bracha1987asynchronous}
    \\
    & (2, 4) & $O(n^2)$ & None  & \textbf{this work}
    \\
    \midrule
    \multirow{3}{1.3cm}{\centering Balanced \\RBC  for long \\ messages} & 3 & $O(nL +\k n^2 \log{n})$ & Hash & \cite{cachin2005asynchronous}
    \\
    & 4 & $O(nL + \k n^2)$ & Hash &\cite{alhaddad2022balanced}
    \\
    & (2, 4) & $O(nL +\k n^2 \log{n})$ & Hash & \textbf{this work}
    \\
    \midrule
    \multirow{2}{*}{ AVSS} & 3 & $O(\k n^2)$ & DL+Hash & \cite{das2021asynchronous}
    \\
    & (2,4) & $O(\k n^2 \log{n})$ & DL+Hash & \textbf{this work}
    \\
    \midrule
    \multirow{2}{*}{\centering ACSS} & 3 & $O(\k n^3)$ & DL+Hash & \cite{cachin2002asynchronous}
    \\
    & (2, 4) & $O(\k n^2 \log{n})$ & DL+Hash & \textbf{this work}
    \\
    \midrule
    \multirow{3}{*}{\centering AVID} & 3 & $O(L + \k n^2 \log{n})$ & Hash & \cite{yang2022dispersedledger}
    \\
    & 4 & $O(L + \k n^2)$ & Hash & \cite{alhaddad2022asynchronous,Alhaddad_2022}
    \\
    & (2, 4) & $O(L + \k n^2 \log{n})$ & Hash & \textbf{this work}
    \\
    \midrule
    \multirow{3}{1.5cm}{\centering PQ-safe \\DAG-based \\BFT SMR} & 6 & $O(\k n^3 \log{n})$ & Hash & \cite{spiegelman2022bullshark}
    \\
    & 4 & $O(\k n^3 \log{n})$ & Hash & \cite{arunShoalHighThroughput2025a}
    \\
    & (3, 5) & $O(\k n^3 \log{n})$ & Hash & \textbf{this work}
    \\
    \bottomrule
    \end{tabularx}
    \end{center}
    $(a, b)$ denotes a good-case latency of $a$ steps under optimistic conditions, a latency of $b$ steps under normal conditions.
    For AVID, the communication complexity column reports the total server-side dispersal cost for an input of $L$ bits.
    For the signature-free DAG-based BFT protocols of~\cite{spiegelman2022bullshark} and~\cite{arunShoalHighThroughput2025a}, the reported latency and communication complexity are based on their use of the RBC protocol by Cachin et al.~\cite{cachin2005asynchronous}. 
    \label{tbl:related_work}
    \vspace{-5mm}
\end{table}

\mypara{Communication-efficient balanced optimistic RBC for long messages.} We also present a communication-efficient balanced\footnote{In balanced RBC, each party, including the broadcaster, incurs the same asymptotic communication cost.} optimistic RBC that offers the same latency guarantees as the optimistic RBC with $O(nL + \k n^2 \log{n})$ communication to propagate an $L$-bit input. This protocol relies on hash functions and is post-quantum secure. %

Our optimistic RBC paves the way for latency improvements in signature-free implementations of various distributed computing primitives. Notably, the underlying latency reduction technique generalizes beyond RBC to other primitives with similar communication patterns. We demonstrate its applicability by enabling two key primitives—AVSS and AVID—to achieve 2-step optimistic termination. Additionally, we leverage our optimistic RBC to construct latency-efficient, signature-free DAG-based BFT protocols.

\mypara{Extension: Optimistic asynchronous verifiable secret sharing.}
Many existing AVSS and asynchronous complete secret sharing (ACSS) schemes~\cite{alhaddad2021high,das2021asynchronous,das2022practical,das2023verifiable} use RBC as a black-box component. For such protocols, our optimistic RBC can be directly integrated to enable 2-step optimistic termination. As concrete examples, we consider the recent AVSS and ACSS protocols of Das et al.~\cite{das2021asynchronous}, replacing their RBC component with our communication-efficient optimistic RBC. Their AVSS protocol relies on discrete logarithm (DL) assumptions and hash functions, while their ACSS protocol additionally requires a PKI setup. The resulting constructions preserve the original cryptographic and setup assumptions while supporting 2-step optimistic termination with $O(\kappa n^2 \log n)$ communication.

Our latency reduction technique can also be applied to existing ACSS protocols in a white-box manner. As a concrete example, we build on the setup-free ACSS protocol of Cachin et al.~\cite{cachin2002asynchronous}, which relies only on DL and hash function assumptions. To the best of our knowledge, their construction remains the state-of-the-art among setup-free ACSS protocols, achieving $O(\kappa n^3)$ communication and 3-step good-case latency.  We extend this construction to achieve 2-step optimistic termination and reduce the communication complexity to $O(\kappa n^2 \log n)$ albeit at the cost of increasing the good-case and bad-case latency by an additional step. %

\mypara{Extension: Optimistic asynchronous verifiable information dispersal.}
Our latency reduction technique can also be applied to AVID, enabling 2-step optimistic termination while preserving the dispersal, storage, and retrieval costs of existing AVID protocols. Particularly, the protocol achieves client dispersal cost of $O(L + \k n \log{n})$, total server-side dispersal cost of $O(L + \k n^2 \log{n})$, and total storage cost of $O(L + \k n \log{n})$ to disperse an $L$-bit input. In AVID, the presence of a distinct client outside the set of $n$ parties slightly increases the quorum size required for optimistic termination.  When the client is honest, the protocol can achieve optimistic good-case latency of $2$ steps if $\ceil{\frac{n+2f+1}{2}}$ parties behave honestly, a 4-step good-case latency and a 5-step bad-case latency. %

\mypara{Application: Latency-efficient signature-free DAG-based BFT protocols.} Recently, DAG-based BFT SMR protocols~\cite{baird2016swirlds,danezis2018blockmania,gkagol2019aleph,keidar2021all,keidar2023cordial,spiegelman2022bullsharkpartially,shrestha2024sailfish} have gained traction for improving throughput by maximizing bandwidth utilization. In the partially synchronous model~\cite{dwork1988consensus}, recent works~\cite{shrestha2024sailfish,arunShoalHighThroughput2025a,babelMysticetiReachingLatency2025} achieve 3-step good-case latency for the leader vertex (proposed by the leader) in the authenticated setting. However, most of these protocols (except~\cite{spiegelman2022bullshark,spiegelman2023shoal,arunShoalHighThroughput2025a}) are inherently tied to signatures and thus incompatible with the signature-free setting. In contrast, Bullshark~\cite{spiegelman2022bullshark}, Shoal~\cite{spiegelman2023shoal}, and Shoal++~\cite{arunShoalHighThroughput2025a} treat RBC as a black-box component, enabling signature-free variants when instantiated with a signature-free RBC protocol. Using our optimistic RBC in these protocols can naturally reduce their latency under optimistic conditions.

However, existing signature-free DAG-based consensus protocols nevertheless continue to incur high latency, primarily because they do not guarantee the presence of a leader in every DAG round~\cite{shrestha2024sailfish}.
Shoal~\cite{spiegelman2023shoal} and its successor Shoal++~\cite{arunShoalHighThroughput2025a} attempt to address this using a pseudo-pipelining strategy, chaining multiple Bullshark~\cite{spiegelman2022bullshark} instances to emulate leader availability in every round.
However, these instances are sequentially dependent—each new instance relies on the successful commit of the previous one.
If a commit fails, this dependency disrupts leader availability in the following rounds, thereby increasing overall latency.

To address this fundamental limitation, we propose \name, a latency-efficient, signature-free DAG-based protocol that guarantees a leader in every round and that can use our optimistic RBC. We build on the core ideas of Sailfish~\cite{shrestha2024sailfish}, which also supports a leader in every round, but requires the use of signatures during failures. \name achieves comparable latency guarantees while operating entirely in the signature-free setting—matching the latency of signature-based protocols~\cite{castro1999practical,doidge2024moonshot} under optimistic conditions. In particular, \name commits the leader vertex in 1 RBC + 1 step and requires only 1 additional RBC to commit the non-leader vertices. This contribution is of independent interest and improves the latency guarantees of signature-free DAG-based BFT protocols.

\mypara{Evaluation.} We implement and evaluate \name against Bullshark~\cite{spiegelman2022bullshark} and Sailfish~\cite{shrestha2024sailfish}. Even on machines with just 4 vCPUs, \name outperforms these protocols, which typically require 48–64 vCPUs~\cite{babelMysticetiReachingLatency2025,arunShoalHighThroughput2025a}. In a 50-node deployment, \name achieves ~67\% lower latency than Sailfish at 100K TPS and delivers higher peak throughput. The optimistic-RBC variant further reduces latency compared to the Bracha-based version, confirming our theoretical insights. These results demonstrate that signature-free protocols can be both cost-efficient and post-quantum secure.

\mypara{Formal specifications, model-checking, and mechanically-checked proofs.} We provide formal specifications of the optimistic RBC (\full{\Cref{sec:bcast-spec}}{full version~\cite{shrestha2025optimistic}}) and of \name (\full{\Cref{sec:sailfish-spec}}{{full version~\cite{shrestha2025optimistic}}}); the specifications are written in PlusCal/TLA+~\cite{lamportPlusCalAlgorithmLanguage2009,lamportSpecifyingSystemsTLA2002}, and the TLC model-checker~\cite{yu_model_1999} checks representative finite configurations of the executable models.
We additionally specify the optimistic RBC protocol using Ivy~\cite{padonIvySafetyVerification2016,mcmillanIvyMultimodalVerification2020} and we mechanically verify using a combination of the Ivy prover and the Isabelle/HOL proof assistant~\cite{nipkow2002isabelle} that the safety and liveness properties of the optimistic RBC hold regardless of the number of parties and in all executions regardless of their length.
We describe the proof in~\full{\Cref{sec:ivy-proof} and include the full Ivy specification in~\Cref{sec:ivy-specification}}{the full version~\cite{shrestha2025optimistic}}.
Artifacts available on Zenodo~\cite{zenodo} allow reproducing all the verification results.

\mypara{Remark on post-quantum security.} All of our protocols—except the AVSS and ACSS constructions—are also post-quantum secure.  Our optimistic RBC can be applied in a black-box manner to the post-quantum secure AVSS protocol of Shoup et al.~\cite{shoup2024lightweight}, effectively reducing the number of steps under optimistic conditions. Our efforts in this work is particularly timely given the recent landmark move by NIST to deprecate and disallow legacy encryption and signature algorithms in favor of post-quantum cryptography~\cite{moody2024transition}. In this light, our protocols offer a compelling direction for building latency-efficient post-quantum secure distributed systems, combining strong theoretical guarantees with future-proof cryptographic resilience. %

\mypara{Organization.}~\Cref{sec:preliminaries} introduces the system model and preliminaries. In~\Cref{sec:opt_rbc}, we present our optimistic RBC protocol, followed by a matching lower bound in~\Cref{sec:lowerbound}.~\Cref{sec:comm-eff-rbc} describes a communication-efficient variant for long messages. We then present optimistic AVSS and AVID protocols in~\Cref{sec:opt-avss} and \full{~\Cref{sec:opt-avid}}{full version~\cite{shrestha2025optimistic}}, respectively. \Cref{sec:dag_bft} details a latency-efficient, signature-free \name. Evaluation results are presented in~\Cref{sec:evaluation}, and we conclude with a comprehensive discussion of related work in~\Cref{sec:related_work}.

\section{Preliminaries}
\label{sec:preliminaries}
We consider a system $\nodes := \node{1},\ldots,\node{n}$ consisting of $n$ parties out of which up to $f$ parties can be Byzantine, meaning they can behave arbitrarily. A party that is not faulty throughout the execution is considered to be \emph{honest} and executes the protocol as specified.

We assume an asynchronous network where the adversary has the ability to arbitrarily delay or reorder messages exchanged between honest parties. However, the adversary is constrained to eventually deliver all messages. We define $\delta$ 
to represent the actual (variable and unknown) transmission latencies of messages during a time period under consideration.

For the latency-efficient DAG-based BFT protocol, we consider the partial synchrony model of Dwork et al.~\cite{dwork1988consensus}. In this model, the network initially operates in an asynchronous state, allowing the adversary to arbitrarily delay messages sent by honest parties. However, after an unknown point in time, known as the \emph{Global Stabilization Time} (GST), the adversary must ensure that all messages from honest parties are delivered to their intended recipients within $\Delta$ time of being sent (where $\Delta$ is known to the parties). Note that $\delta \leq \Delta$ after GST. Additionally, we assume the local clocks of the parties have \emph{no clock drift} and \emph{arbitrary clock skew}.

We represent the hash of an input $x$ by $H(x)$, where $H$ denotes the hash function.
We use the same parameter $\k$ to denote the hash size, the security parameter, and the size of secret shares, depending on context.

We say that a protocol is \emph{signature-free} when it does not rely on a PKI. \emph{Post-quantum secure} protocols are secure against adversaries that use polynomially-bounded quantum computers (for example, protocols whose security relies only on hash-preimage resistance are post-quantum secure). Meanwhile, a protocol is considered \emph{information-theoretically secure} when it is secure against adversaries with unbounded computational power, such as those capable of guessing pre-images of hash functions. Related but different terms appear in the literature: \emph{unauthenticated} protocols do not use message authentication, and \emph{error-free} protocols are secure against adversaries with unbounded computational power, and always achieve guarantees in every possible execution.

\subsection{Definitions}

\subsubsection{Reliable Broadcast}

\begin{definition}[Byzantine reliable broadcast~\cite{bracha1987asynchronous}]
In a Byzantine reliable broadcast (RBC), a designated sender $\node{k}$ may invoke \Call{r\_bcast$_k$}{$m$} to propagate an input $m$. Each party $\node{i}$ may then output (we also say commit) the message $m$ via \Call{r\_deliver$_i$}{$m, \node{k}$} where $\node{k}$ is the designated sender. The reliable broadcast primitive satisfies the following properties:
\begin{itemize}[noitemsep,leftmargin=*]
    \item[-] \textbf{Validity.} If the designated sender $\node{k}$ is honest and calls \Call{r\_bcast$_k$}{$m$} then every honest party  eventually outputs \Call{r\_deliver}{$m,\node{k}$}.
    \item[-] \textbf{Agreement.} If an honest party $\node{i}$ outputs \Call{r\_deliver$_i$}{$m, \node{k}$}, then every other honest party $\node{j}$ eventually outputs \Call{r\_deliver$_j$}{$m, \node{k}$}.
    \item[-] \textbf{Integrity.} Each honest party $\node{i}$ outputs \Call{r\_deliver$_i$}{} at most once regardless of $m$.
\end{itemize}
\end{definition}

Next, we consider the Byzantine reliable broadcast (RBC) protocol that tolerates up to~$f$ Byzantine faults and define its associated latency parameters. In asynchronous setting, where message delays are unbounded, latency is measured using the notion of asynchronous communication steps~\cite{canetti1993fast}. Specifically, a protocol is said to run in $R$ asynchronous steps if its execution completes in time at most $R$ times the maximum message delay between honest parties during the run.

\begin{definition}[Good-case Latency]
The protocol has a good-case latency of $R$ communication steps, if all honest parties deliver within $R$ asynchronous steps, given the designated broadcaster is honest.
\end{definition}

\begin{definition}[Optimistic Good-case Latency]
The protocol has a good-case latency of $R$ communication steps, if all honest parties deliver within time $R$ asynchronous steps, given the designated broadcaster is honest and at least $n_o$ non-broadcaster parties behave honestly, for some parameter $n_o \ge n-f$.
\end{definition}

In this paper, we always consider $n_o = \ceil{\frac{n+2f-2}{2}}$.

\begin{definition}[Bad-case Latency]
The protocol has a bad-case latency of $R$ when, given that an honest party has terminated in the (optimistic) good-case latency, all honest parties terminate within $R$ asynchronous steps.
\end{definition}

Finally, we say that the protocol is an $(i,j,k)$-optimistic reliable broadcast protocol when it has optimistic good-case latency $i$, good-case latency $j$, and bad-case latency $k$.
Note that our latency definitions apply both to the synchronous and asynchronous settings.

\subsubsection{Byzantine Atomic Broadcast}
We rely on the following definition to establish the security of \name.
\begin{definition}[Byzantine atomic broadcast~\cite{keidar2021all,spiegelman2022bullshark}]
\label{dfn:bab}
Each honest party $\node{i} \in \nodes$ can call \Call{a\_bcast$_i$}{$m,r$} to propagate its input $m$ in some round $r \in \N$. Each party $\node{i}$ then outputs \Call{a\_deliver$_i$}{$m,r,\node{k}$}, where $\node{k}\in \nodes$ represents the sender of the message. A Byzantine atomic broadcast protocol satisfies the following properties:

\begin{itemize}[noitemsep,leftmargin=*]
    \item[-] \textbf{Agreement.} If an honest party $\node{i}$ outputs \Call{a\_deliver$_i$}{$m, r, \node{k}$}, then every other honest party $\node{j}$ eventually outputs \Call{a\_deliver$_j$}{$m,$ $ r, \node{k}$}.
    
    \item[-] \textbf{Integrity.} For every round $r \in \N$ and party $\node{k} \in \nodes$, an honest party $\node{i}$ outputs \Call{a\_deliver$_i$}{} at most once regardless of $m$.
    
    \item[-] \textbf{Validity.} If an honest party $\node{k}$ calls \Call{a\_bcast$_k$}{$m,r$} then every honest party  eventually outputs \Call{a\_deliver}{$m,r,\node{k}$}.
    
    \item[-] \textbf{Total order.} If an honest party $\node{i}$ outputs \Call{a\_deliver$_i$}{$m, r, \node{k}$} before \Call{a\_deliver$_i$} {$m', r', \node{\ell}$}, then no honest party $\node{j}$ outputs \Call{a\_deliver$_j$}{$m', r', \node{\ell}$} before \Call{a\_deliver$_j$}{$m, r, \node{k}$}.
\end{itemize}
\end{definition}

\section{A (2, 4, 5)-Optimistic RBC}
\label{sec:opt_rbc}

\begin{figure}[tbp]
    \footnotesize
        \begin{boxedminipage}[t]{\columnwidth}
        \begin{enumerate}[leftmargin=*]
            \item \textbf{Propose.} The broadcaster $\node{k}$ sends $\sig{\Propose, m}$ to all parties.
            
            \item \textbf{Echo.} Upon receiving the first $\sig{\Propose, m}$ message from the broadcaster, party $\node{i}$ sends $\sig{\Echo, m}$ to all parties.

            {\color{blue}
            \item \textbf{Vote.} Upon receiving $\sig{\Echo, m}$ from $\ceil{\frac{n}{2}}$ non-broadcaster parties, send $\sig{\Vote, m}$ to all parties if not already sent.
            }

            {\color{blue}
            \item \textbf{Ack.} Party $\node{i}$ sends a $\sig{\Ack, m}$ message to  all parties (if not already sent) under either of the following conditions:
            \begin{itemize}
                \item[-] Upon receiving $\sig{\Vote, m}$ from $\ceil{\frac{n+f-1}{2}}$ non-broadcaster parties,
                \item[-] Upon receiving $\sig{\Echo, m}$ from $\ceil{\frac{n+f-1}{2}}$ non-broadcaster parties
            \end{itemize}
            }

            \item \textbf{Ready.} Party $\node{i}$ sends a $\sig{\Ready, m}$ message to  all parties (if not already sent) under either of the following conditions:
            \begin{itemize}
                \item[-] Upon receiving $\sig{\Ack, m}$ from $\ceil{\frac{n+f-1}{2}}$ non-broadcaster parties,
                \item[-] Upon receiving $f+1$ $\sig{\Ready, m}$
            \end{itemize}
            
             {\color{blue}
            \item \textbf{Opt Commit.} Upon receiving $\sig{\Echo, m}$ from $\ceil{\frac{n+2f-2}{2}}$ non-broadcaster parties, party $\node{i}$ commits $m$.
            }

            \item \textbf{\{4,5\}-step Commit.} Upon receiving $2f+1$ $\sig{\Ready, m}$ messages, party $\node{i}$ commits $m$ and terminates.%
        \end{enumerate}
    \end{boxedminipage}
\caption{(2,4,5) Optimistic reliable broadcast under $n \ge 3f+1$}
\label{fig:opt_rbc}
\end{figure}

In this section, we present a signature-free RBC protocol under $n \ge 3f+1$. The protocol achieves an optimistic good-case latency of $2$ steps when the broadcaster is honest and at least $n_o=\ceil{\frac{n+2f-2}{2}}$ non-broadcaster parties behave honestly. When more than $\floor{\frac{n-2f}{2}}$ non-broadcaster parties behave maliciously, the protocol achieves a good-case latency of 4 steps and a bad-case latency of 5 steps, incurring one additional step compared to Bracha's RBC~\cite{bracha1987asynchronous}.  In this regard, our protocol is a $(2, 4, 5)$-optimistic reliable broadcast. This protocol is also information-theoretically secure.

\mypara{Protocol details.}  We extend Bracha’s RBC~\cite{bracha1987asynchronous} to enable optimistic delivery in just 2 steps. The complete protocol is presented in~\Cref{fig:opt_rbc}, with our modifications highlighted in \textcolor{blue}{blue}. 

Recall that in Bracha’s protocol, parties first send an $\Echo$ message upon
receiving the broadcaster’s value, then send a $\Ready$ message after receiving
a quorum of $\Echo$ messages, and finally terminate upon receiving a quorum of
$\Ready$ messages. To ensure termination, a party also sends a $\Ready$ message
upon receiving $f+1$ $\Ready$ messages.

The key idea in our protocol is to augment Bracha’s algorithm with an
optimistic delivery rule, allowing a party to deliver after receiving
sufficiently many $\Echo$ messages. However, if an honest party delivers
optimistically (possibly due to Byzantine messages observed only by that
party) we must guarantee that all honest parties eventually deliver the same
value.

To achieve this, we introduce a new type of $\Vote$ message, which is
guaranteed to propagate once an honest party delivers optimistically. When $\Vote$ messages reach the threshold of $\ceil{\frac{n+f-1}{2}}$ non-broadcaster parties, they trigger a new type of $\Ack$ message that helps ensure unique delivery. Similarly, an $\Ack$ message is also triggered upon receiving $\ceil{\frac{n+f-1}{2}}$ $\Echo$ messages from non-broadcaster parties for the same value.

Finally, when $\Ack$ messages reach the same threshold of
$\ceil{\frac{n+f-1}{2}}$ non-broadcaster parties, they trigger Bracha’s
original $\Ready$ message.

In more detail, for an input value $m$, the broadcaster sends $\sig{\Propose, m}$ to all parties, and each party responds by sending $\sig{\Echo, m}$ upon receiving the first $\sig{\Propose, m}$ from the broadcaster. An honest party $\node{i}$ can optimistically commit $m$ in $2$ communication steps as soon as it receives at least $\ceil{\frac{n+2f-2}{2}}$ $\sig{\Echo, m}$ messages from non-broadcaster parties.

To ensure agreement, each party $\node{i}$ sends a $\sig{\Vote, m}$ message upon receiving $\sig{\Echo, m}$ from at least $\ceil{\frac{n}{2}}$ non-broadcaster parties. This condition is guaranteed to hold for all honest parties whenever some honest party commits to $m$ optimistically, since at least $\ceil{\frac{n+2f-2}{2}} - (f-1) = \ceil{\frac{n}{2}}$ honest non-broadcaster parties must have already sent $\sig{\Echo, m}$ to all parties.

Moreover, this voting condition cannot hold for any conflicting value $m' \neq m$. Once an honest party commits to $m$, at least $\ceil{\frac{n}{2}}$ honest parties must have already sent $\sig{\Echo, m}$, leaving at most $n - f - \ceil{\frac{n}{2}} = \ceil{\frac{n-2f}{2}}$ honest non-broadcaster parties that could send $\sig{\Echo, m'}$. Even with the help of all $f-1$ Byzantine non-broadcaster parties, it is impossible to collect $\ceil{\frac{n}{2}}$ $\sig{\Echo, m'}$ messages.

Upon receiving $\sig{\Vote, m}$ from at least $\ceil{\frac{n+f-1}{2}}$ non-broadcaster parties, an honest party sends $\sig{\Ack, m}$. An honest party also sends $\sig{\Ack, m}$ (if it has not already done so) upon receiving $\sig{\Echo, m}$ from at least $\ceil{\frac{n+f-1}{2}}$ non-broadcaster parties. Similarly, upon receiving $\sig{\Ack, m}$ from at least $\ceil{\frac{n+f-1}{2}}$ non-broadcaster parties, an honest party sends $\sig{\Ready, m}$ and decides $m$ upon receiving at least $2f+1$ $\sig{\Ready, m}$ messages. These thresholds are eventually satisfied since there are at least $n-f$
honest parties.

If more than $n - \ceil{\frac{n+2f-2}{2}} - 1 = \floor{\frac{n-2f}{2}}$ non-broadcaster parties behave maliciously, they may prevent optimistic delivery. In this case, the protocol falls back to the standard path and terminates in $4$ communication steps in the good case and $5$ communication steps in the bad case.

\full{We present detailed security analysis in~\Cref{sec:opt-rbc-proof} and its formal specification in~\Cref{sec:bcast-spec}}{We present detailed security analysis and its formal specification in full version~\cite{shrestha2025optimistic}}.

\section{A Lower Bound on the Resilience for Optimistic Good-case Latency}
\label{sec:lowerbound}
This section establishes a tight lower bound on the resilience for optimistic good-case latency under the condition $n \le 4f-1$. As previously noted, our result holds for signature-free RBC protocols, even in the synchronous setting. The proof of the following lower bound is illustrated in~\Cref{fig:lb}.

\begin{figure}[tbp]
\centering
\begin{subfigure}[t]{\columnwidth}
\includegraphics[scale=0.38]{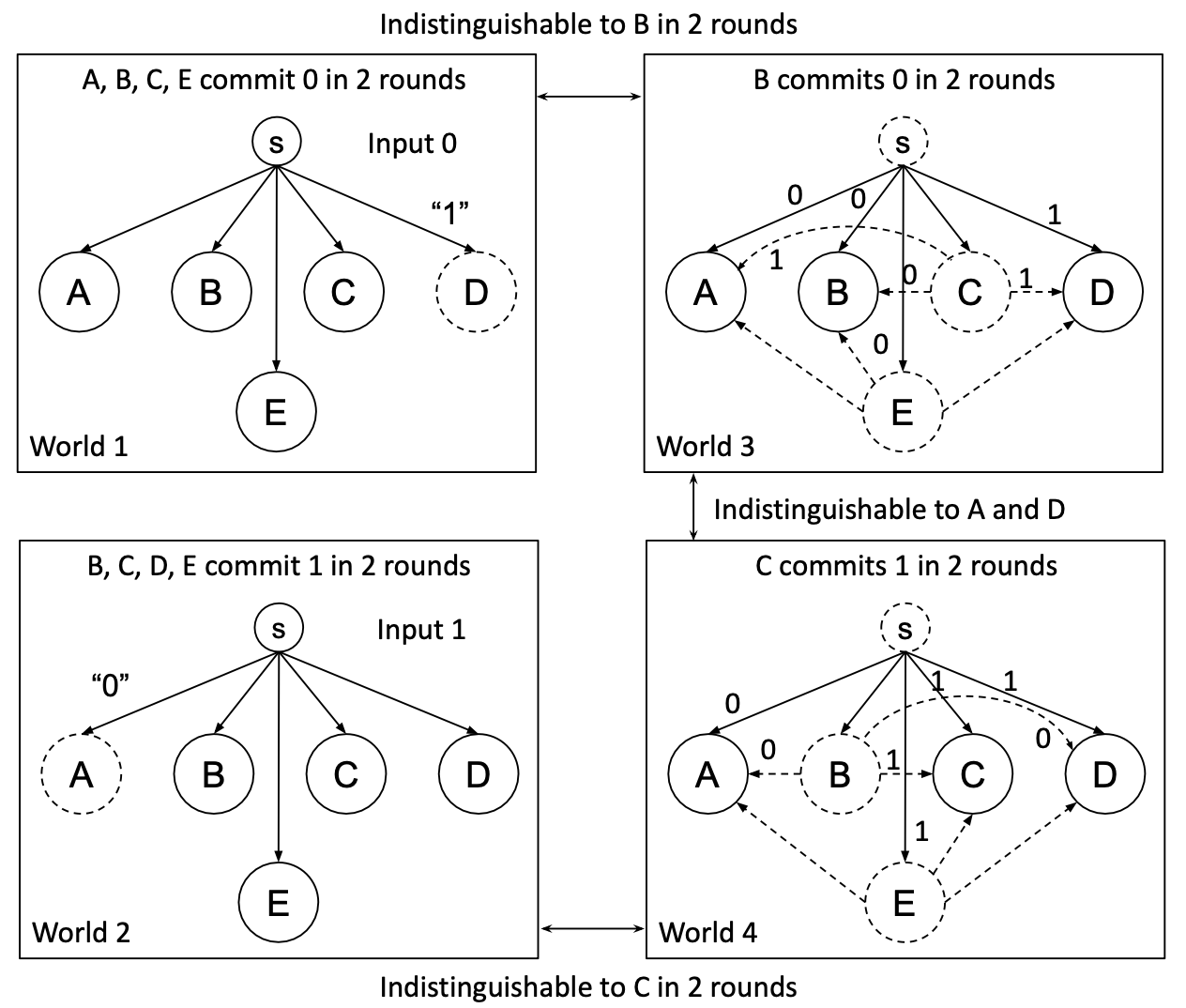}
\end{subfigure}
\caption{Resilience lower bound for optimistic good-case latency under $n \le 4f-1$. Dotted circles represent Byzantine parties.}
\label{fig:lb}
\end{figure}

\begin{theorem}
There does not exist a reliable broadcast protocol tolerating $f$ Byzantine faults under $3f+1 \leq n \leq 4f-1$ that can achieve an optimistic good-case latency of $2$ steps when more than $\floor{\frac{n-2f}{2}}$ non-broadcaster parties behave maliciously, even under synchrony.
\end{theorem}
\begin{proof}
We assume all parties start their protocol at the same time, which strengthens the lower bound result. Under synchrony, messages exchanged between honest parties are guaranteed to be delivered within $\Delta$ time. Consequently, each step of the protocol lasts for at most $\Delta$ time. Suppose, for the sake of contradiction, that a protocol exists that can achieve an optimistic good-case latency of $2$ steps even when $\floor{\frac{n-2f}{2}}+1$ (i.e., $\floor{\frac{n-2f+2}{2}}$) non-broadcaster parties behave maliciously. We partition the $n$ parties into the designated broadcaster $s$ and five distinct sets: $A$, $B$, $C$, $D$ and $E$, where $|A| = |D| = \floor{\frac{n-2f+2}{2}}$. When $n$ is even, we set $|B| = |C| = f-2$ and $|E| = 1$, whereas $n$ is odd, we set  $|B| = |C| = f-1$ and $|E| = 0$.

Notice that when $n$ is even, the sum $|A| + |B| + |C| + |D| + |E| + 1 = n$ holds.  Similarly, when $n$ is odd, we have $|A| + |B| + |C| + |D| + 1 = n$, implying $|E| = 0$. For brevity, we use $A$ ($B$, $C$, $D$, $E$) to refer to all parties in set $A$ ($B$, $C$, $D$, $E$), respectively.

\mypara{World~$1$ (W$1$).}
The broadcaster $s$ is honest and sends $0$ to all parties. The parties in $D$ are Byzantine but follow the protocol honestly, except that they pretend to have received input $1$ from $s$. Since the broadcaster is honest and there are at most $\floor{\frac{n-2f+2}{2}}$ Byzantine non-broadcaster parties, by validity and the optimistic good-case latency guarantee, all parties in $A \cup B \cup C \cup E$ commit $0$ in 2 steps.

\mypara{World~$2$ (W$2$).}
The broadcaster $s$ is honest and sends $1$ to all parties. The parties in $A$ are Byzantine but follow the protocol honestly, except that they pretend to have received input $0$ from $s$. Since the broadcaster is honest and there are at most $\floor{\frac{n-2f+2}{2}}$ Byzantine non-broadcaster parties, by validity and the optimistic good-case latency guarantee, all parties in $B \cup C \cup D \cup E$ commit $1$ in $2$ steps.

\mypara{World~$3$ (W$3$).}
The broadcaster $s$ and the parties in $C \cup E$ are Byzantine. Observe that when $n$ is even, we have $|C| = f-2$ and $|E| = 1$, whereas $n$ is odd, $|C| = f-1$ and $|E| = 0$. Thus, the total number of Byzantine parties is at most $f$. 

In \textbf{W$3$}, the broadcaster $s$ behaves identically to \textbf{W$1$} when communicating with $A \cup B$ and identically to \textbf{W$2$} when communicating with $D$. Similarly, the Byzantine parties in $C$ behave identically to \textbf{W$1$} when interacting with $B$ and identically to \textbf{W$2$} when interacting with $A \cup D$. Additionally, when $|E| \not = 0$, the Byzantine parties in $E$ behave identically to \textbf{W$1$} when interacting with $B$. 

\begin{claim}\label{clm:w1-w3}
The parties in $B$ cannot distinguish between \textbf{W$1$} and \textbf{W$3$} in the first two steps.
\end{claim}
\begin{claimproof}
The broadcaster $s$ behaves identical to \textbf{W$1$} when communicating with $A \cup B$ and identical to \textbf{W$2$} when communicating with $D$. Moreover, the Byzantine parties in $C \cup E$ behave identical to \textbf{W$1$} when communicating with $B$. Thus, the first two step of messages to parties in $B$ are identical in both  \textbf{W$1$} and \textbf{W$3$}. Thus, parties in $B$ cannot distinguish between \textbf{W$1$} and \textbf{W$3$}.
\end{claimproof}

\mypara{World~$4$ (W$4$).}
The broadcaster $s$ and the parties in $B \cup E$ are Byzantine. When $n$ is even, we have $|B| = f-2$ and $|E| = 1$, whereas $n$ is odd, $|B| = f-1$ and $|E| = 0$. Thus, the total number of Byzantine parties is at most $f$. 

In \textbf{W$4$}, the broadcaster $s$ behaves identically to \textbf{W$2$} when communicating with $C \cup D$ and identically to \textbf{W$1$} when communicating with $A$. Similarly, the Byzantine parties in $B$ behave identically to \textbf{W$2$} when interacting with $C$ and identically to \textbf{W$1$} when interacting with $A \cup D$. Additionally, when $|E| \not = 0$, the Byzantine parties in $E$ behave identically to \textbf{W$2$} when interacting with $C$ and behaves identically to \textbf{W$3$} when interacting with $A$ and $D$.

\begin{claim}\label{clm:w2-w4}
The parties in $C$ cannot distinguish between \textbf{W$2$} and \textbf{W$4$} in the first two steps. Similarly, the parties in $A$ cannot differentiate between \textbf{W$2$} and \textbf{W$4$}.
\end{claim}
\begin{claimproof}
The broadcaster $s$ behaves identical to \textbf{W$2$} when communicating with $C \cup D$ and identical to \textbf{W$1$} when communicating with $A$. Moreover, the Byzantine parties in $B \cup E$ behave identical to \textbf{W$2$} when communicating with $C$. Thus, the first two step of messages to parties in $C$ are identical in both  \textbf{W$2$} and \textbf{W$4$}. Thus, parties in $C$ cannot distinguish between \textbf{W$2$} and \textbf{W$4$}.
\end{claimproof}

\begin{claim}\label{clm:w3-w4}
The parties in $A$ and $D$ cannot distinguish between \textbf{W$3$} and \textbf{W$4$}.
\end{claim}
\begin{claimproof}
Observe that the broadcaster $s$ behaves identically in \textbf{W$3$} and \textbf{W$4$} when communicating with parties in $A$ and $D$. Likewise, the parties in $B \cup C \cup E$ behave identically in \textbf{W$3$} and \textbf{W$4$} when interacting with parties in $A$ and  $D$. Consequently, the parties in $A$ and $D$ cannot distinguish between \textbf{W$3$} and \textbf{W$4$}.
\end{claimproof}

By~\Cref{clm:w1-w3}, parties in $B$ cannot distinguish between \textbf{W$1$} and \textbf{W$3$} in the first two steps. Therefore, the parties in $B$ will commit $0$ in $2$ steps in \textbf{W$3$}. Similarly, by~\Cref{clm:w2-w4}, the parties in $C$ cannot distinguish between \textbf{W$2$} and \textbf{W$4$} in the first two steps. Therefore, the parties in $C$ will commit $1$ in $2$ steps in \textbf{W$4$}. 

By~\Cref{clm:w3-w4}, parties in $A \cup D$ cannot distinguish between \textbf{W$3$} and \textbf{W$4$}. This leads to a contradiction as agreement will be violated. Therefore, no such protocol can exist.
\end{proof}

\section{Communication-Efficient Balanced Optimistic RBC}
\label{sec:comm-eff-rbc}
This section introduces a communication-efficient balanced $(2, 4, 5)$ optimistic reliable broadcast protocol. It leverages an erasure coding scheme~\cite{reed1960polynomial} (detailed in\full{~\Cref{sec:extended-preliminaries}}{ full version~\cite{shrestha2025optimistic}}) and cryptographic accumulators~\cite{merkle1987digital} to enhance communication efficiency. By employing Merkle proofs~\cite{merkle1987digital} as accumulators, the protocol achieves a communication complexity of $O(nL + \k n^2 \log{n})$ for propagating an $L$-bit input, where each party sends $O(L+\kappa n\log{n})$ bits (communication is balanced over all parties).

\mypara{Protocol details.} We build upon the communication-efficient balanced RBC protocol by Cachin et al.~\cite{cachin2005asynchronous, miller2016honey} to enable optimistic delivery. The complete protocol appears in~\Cref{fig:balanced_opt_rbc},  with our modifications highlighted in \textcolor{blue}{blue}.

The broadcaster encodes the input message $m$ using $(n, \ceil{\frac{n-f+1}{2}})$-RS codes to generate the code words $M:= [s_1, \ldots, s_n]$. It then sends $\sig{\Propose, s_j, w_j, h}$ to party $\node{j}$ $\forall j \in [n]$, where $w_j$ is the Merkle proof of $j^{th}$ code word, and $h$ is the Merkle root over $M$. Upon receiving the first $\sig{\Propose, s_i, w_i, h}$ from the broadcaster, each party $\node{i}$ verifies whether $w_i$ is a valid Merkle proof for the root $h$. If the verification succeeds, it sends $\sig{\Echo, s_i, w_i, h}$ to all parties.

Upon receiving valid $\sig{\Echo, s_*, w_*, h}$ messages from at least $\ceil{\frac{n}{2}}$ non-broadcaster parties, party $\node{i}$ attempts to interpolate the codeword vector $M' := [s'_1, \ldots, s'_n]$ using any $\ceil{\frac{n-f+1}{2}}$ of the received codewords. It then verifies whether the Merkle root computed from $M'$ matches $h$. This check ensures that the broadcaster correctly constructed $h$ from a valid encoding of the message $m$, and that honest parties will be able to decode $m$ once they receive a sufficient number of codewords. This functionality is encapsulated in the verify\_interpolation procedure. If the verification succeeds, party $\node{i}$ sends $\sig{\Vote, s_i, w_i, h}$ to all parties.

An honest party optimistically commits to $m$ and terminates upon receiving valid $\sig{\Echo, s_*, w_*, h}$ messages from at least $\ceil{\frac{n+2f-2}{2}}$ non-broadcaster parties. Observe that if some honest party optimistically commits to $m$, then all honest parties will eventually receive at least $\ceil{\frac{n}{2}}$ such $\sig{\Echo, s_*, w_*, h}$ messages. Consequently, each honest party sends $\sig{\Vote, s_j, w_j, h}$ to every party $\node{j}$, for all $j \in [n]$.

To ensure unique delivery, an honest party $\node{i}$ sends $\sig{\Ack, h}$ to all  parties either (i) upon receiving $\sig{\Vote, s_*, w_*, h}$ messages from at least $\ceil{\frac{n+f-1}{2}}$ non-broadcaster parties, or (ii) upon receiving $\sig{\Echo, s_*, w_*, h}$ messages from at least $\ceil{\frac{n+f-1}{2}}$ non-broadcaster parties. 

An honest party sends $\sig{\Ready, h}$ to all parties upon receiving at least
$\ceil{\frac{n+f-1}{2}}$ $\sig{\Ack, h}$ messages. To facilitate agreement, an
honest party also sends $\sig{\Ready, h}$ upon receiving $f+1$
$\sig{\Ready, h}$ messages from other parties. Finally, an honest party
decides $m$ and terminates upon receiving at least $2f+1$
$\sig{\Ready, h}$ messages.

\begin{figure}[ht]
    \footnotesize
    \begin{boxedminipage}[t]{\columnwidth}
         \begin{algorithmic}
            \Procedure{verify\_interpolation}{$\{s_{*}\}, h$}
            \State interpolate $M':= [s'_1, \ldots, s'_n]$ from any $\ceil{\frac{n-f+1}{2}}$ codewords in $\{s_{*}\}$
            \If{$h$ is equal to the Merkle root over $M'$}
            \State return $M'$
            \EndIf
            \State return $\bot$
            \EndProcedure
        \end{algorithmic}
        \begin{enumerate}[leftmargin=*]
            \item \textbf{Propose.} Let $M := [s_1, \ldots, s_n] := \Enc(m, n, \ceil{\frac{n-f+1}{2}})$ and $h$ be the Merkle root computed over $M$. Additionally, let $w_j$ be the Merkle proof of $j^{th}$ code word $s_j$. The broadcaster $\node{k}$ sends $\sig{\Propose, s_j, w_j, h}$ to $\node{j}$  $\forall j \in [n]$.
            
            \item \textbf{Echo.} Upon receiving the first $\sig{\Propose, s_i, w_i, h}$, if the Merkle proof is valid then party $\node{i}$ sends $\sig{\Echo, s_i, w_i, h}$  to all parties.

            {\color{blue}
            \item \textbf{Vote.} Upon receiving $\sig{\Echo, s_*, w_*, h}$ from $\ceil{\frac{n}{2}}$ non-broadcaster parties with valid Merkle proofs, let $M =$ \Call{verify\_interpolation}{$\{s_*\}$, h}. If $M \not = \bot$, send $\sig{\Vote, s_i, w_i, h}$ to all parties.
            }

            {\color{blue}
            \item \textbf{Ack.} Party $\node{i}$ sends a $\sig{\Ack, h}$  message to  all parties (if not already sent) under either of the following conditions:
            \begin{itemize}
                \item[-] Upon receiving $\sig{\Echo, s_*, w_*, h}$ from $\ceil{\frac{n+f-1}{2}}$ non-broadcaster parties, each accompanied by a valid Merkle proof,  such that $M =$ \Call{verify\_interpolation}{$\{s_*\}$, h} and  $M \not = \bot$,
                \item[-] Upon receiving $\sig{\Vote, s_*, w_*, h}$ from $\ceil{\frac{n+f-1}{2}}$ non-broadcaster parties, each accompanied by a valid Merkle proof,  such that $M =$ \Call{verify\_interpolation}{$\{s_*\}$, h} and  $M \not = \bot$
            \end{itemize}
            }
            
            \item \textbf{Ready.} Party $\node{i}$ sends a $\Ready$ message to all parties (if not already sent) under either of the following conditions:

            \begin{itemize}[leftmargin=*]
                \item[-] Upon receiving $\sig{\Ack, h}$ from $\ceil{\frac{n+f-1}{2}}$ non-broadcaster parties, send $\sig{\Ready, h}$

                \item[-] Upon receiving $f+1$ $\sig{\Ready, h}$, send $\sig{\Ready, h}$
            \end{itemize}

            {\color{blue}
            \item \textbf{Opt Commit.} Upon receiving $\sig{\Echo, s_*, w_*, h}$ from $\ceil{\frac{n+2f-2}{2}}$ distinct non-broadcaster parties with valid Merkle proofs, let $M =$ \Call{verify\_interpolation}{$\{s_*\}$, h}. If $M\neq\bot$, party $\node{i}$ decodes $M$ to obtain $m$, then commits $m$.
            }

            \item \textbf{\{3,4\}-step Commit.} Upon receiving $2f+1$ $\sig{\Ready, h}$ messages, wait for $\ceil{\frac{n-f+1}{2}}$ distinct $(\Echo, s_*, w_*, h)$ messages with valid Merkle proofs. Then, decode $m$, commit it, and terminate.

        \end{enumerate}
    \end{boxedminipage}
\caption{Balanced (2,4,5) optimistic reliable broadcast under $n \ge 3f+1$ with $O(nL + \k n^2 \log{n})$ communication.}
\label{fig:balanced_opt_rbc}
\end{figure}

\mypara{Unbalanced optimistic RBC with $O(nL + \k n^2 \log{n})$ communication.} 
The protocol in~\Cref{fig:balanced_opt_rbc} is balanced, meaning all parties asymptotically send the same number of bits.
However, in some applications (e.g., AVSS protocols in~\Cref{sec:opt-avss}), parties must be able to verify the complete input message before sending $\Echo$ or $\Ready$. A solution is to have the broadcaster transmit the entire input message to all parties, while the non-broadcaster parties send $\Echo$ and $\Vote$ messages containing only codewords and Merkle proofs (to ensure agreement even if the broadcaster is faulty). In this modified protocol, the broadcaster incurs a communication cost of $O(nL + \k n)$, while the total communication complexity remains $O(nL + \k n^2 \log{n})$.

We present detailed security analysis in\full{~\Cref{sec:bal-rbc-proof}}{ full version~\cite{shrestha2025optimistic}}. Due to the similarities between balanced RBC and AVID, our latency reduction technique can be naturally extended to AVID, enabling 2-step optimistic termination while preserving the dispersal, storage, and retrieval costs of existing AVID protocols. The key difference is that in AVID, the presence of a distinct client outside the set of $n$ parties slightly increases the quorum size required for optimistic termination. Detailed protocol and analysis can be found in\full{~\Cref{sec:opt-avid}}{ full version~\cite{shrestha2025optimistic}}.

\section{Extensions to Verifiable Secret Sharing}
\label{sec:opt-avss}
In this section, we introduce optimistic AVSS protocols. We observe that optimistic termination can be incorporated into existing AVSS protocols by replacing their RBC component with our optimistic RBC in a ``black-box'' manner. Most of these AVSS protocols rely on verifiable encryption~\cite{das2023verifiable} which requires setup such as PKI. We also propose a setup-free optimistic AVSS protocol by enhancing an existing AVSS protocol in a ``white-box'' manner. We begin by formally defining several variants of the AVSS problem, following the definitions from~\cite{das2021asynchronous}. 

\subsection{Definitions}
For all these primitives, we assume a finite field $\F$ of order $q$. Let $\negl(\k)$ be a negligible function in $\k$. %

\begin{definition}[Asynchronous Verifiable Secret Sharing~\cite{das2021asynchronous}]
    Let $(\Sh, \Rec)$ be a pair of protocols in which a dealer $\node{d}$ shares a secret $s$ using $\Sh$ and other parties use $\Rec$ to recover the shared secret. We say that $(\Sh, \Rec)$ is a $f$-resilient AVSS scheme if the following properties hold with probability $1 - \negl(\k)$: 
\begin{itemize}[leftmargin=*]
    \item[-] \textbf{Validity.} If $\node{d}$ is honest, then all honest parties will complete the $\Sh$ protocol. Moreover, if all honest parties start the $\Rec$ protocol, each honest party $\node{i}$ reconstructs the same secret $s$.
    
    \item[-] \textbf{Termination.} If an honest party terminates the $\Sh$ protocol, then all honest parties will eventually terminate the $\Sh$ protocol.

    \item[-] \textbf{Correctness.} If $\node{d}$ is honest, then each honest party upon terminating Rec, outputs the shared secret $s$. Moreover, if  some honest party terminates the $\Sh$ protocol, then there exists a fixed secret $s' \in \F$, such that each honest party upon terminating $\Rec$, will output $s'$.
    
    \item[-] \textbf{Secrecy.} If the dealer \node{d} is honest and no honest party has started $\Rec$, then an adversary that corrupts up to $f$ parties has no information about $s$.
\end{itemize}
\end{definition}

In the above definition of AVSS, honest parties may complete the sharing phase without actually receiving their share from the dealer. 
To address this issue, a stronger primitive called asynchronous complete secret sharing (ACSS) is introduced. ACSS ensures that by the end of the sharing phase, every honest party receives its share of a consistent secret, even if the dealer is malicious. Such robust guarantees are crucial in several applications, including distributed key generation~\cite{das2022practical, das2023practical}.

\begin{definition}[Asynchronous Complete Secret Sharing] An ACSS protocol is an AVSS protocol that additionally satisfies the following completeness property.
\begin{itemize}[leftmargin=*]
    \item[-] \textbf{Completeness.} If some honest party terminates $\Sh$, then there exists a degree-$f$ polynomial p(·) over $\F$ such that $p(0) = s'$ and each honest party $\node{i}$ will eventually hold a share $s_i = p(i)$. Moreover, when $\node{d}$ is honest, $s' = s$.
\end{itemize}
\end{definition}

\subsection{Optimistic AVSS using Optimistic RBC}
Many existing AVSS protocols in the literature can be enhanced to support optimistic termination by simply replacing their RBC protocol with our optimistic RBC in a ``black-box'' manner. Below, we list a few AVSS schemes that can benefit from this enhancement:

\mypara{AVSS protocol.} We consider the AVSS protocol of Das et al.~\cite{das2021asynchronous}. In their AVSS protocol, the dealer $\node{d}$ randomly selects a degree-$f$ polynomial $p(.)$ such that $p(0) = s$, where $s$ is the secret, and commits to the polynomial using either Feldman’s commitment~\cite{feldman1987practical} or Pedersen’s commitment~\cite{pedersen1991non}. The dealer then disseminates the $O(\k n)$-sized commitment $\C$ using RBC and privately sends each party $\node{j}$ its corresponding secret share $p(j)$ for all $j \in [n]$. Within the RBC, parties send an $\Echo$ message only if the received secret share is consistent with the commitment $\C$ from the first message of the RBC. Consequently, if the RBC terminates, at least $\ceil{\frac{n-f+1}{2}}$ honest parties must have received their secret shares consistent with $\C$, and the secret can be reconstructed. Their protocol is based on the DL assumption and leverages a communication-efficient RBC protocol~\cite{das2021asynchronous}, which achieves $O(nL + \kappa n^2)$ communication to propagate an $L$-bit input and has a 4-step good-case latency. As a result, the corresponding AVSS protocol inherits these properties, achieving $O(\kappa n^2)$ total communication and a 4-step good-case latency.

We replace their RBC with the unbalanced optimistic RBC (discussed in \S\ref{sec:comm-eff-rbc}) to propagate the $O(\k n)$-sized commitment $\C$, while keeping the rest of the protocol unchanged. The resulting protocol retains the DL assumption and achieves an optimistic good-case latency of $2$ steps, albeit with $O(\k n^2 \log{n})$ communication.

\mypara{Dual-threshold ACSS protocol.} We consider the dual-threshold\footnote{In the dual-threshold ACSS protocol, the reconstruction threshold $\ell$ can exceed the Byzantine threshold $f$, i.e., $f \le \ell \le n - f$. Such schemes ensure the secrecy of the secret even when the adversary corrupts more than $f$ and up to $\ell$ parties in the system.} ACSS protocol of Das et al.~\cite{das2021asynchronous}. Their dual-threshold ACSS scheme relies on verifiable encryption~\cite{das2023verifiable}, where the secret shares of all parties are encrypted, and the $O(\k n)$-sized encrypted vector is propagated using communication-efficient RBC~\cite{das2021asynchronous}. This allows all parties to verify the correctness of secret shares without learning the individual shares of others. However, a public-key infrastructure (PKI) is required to enable verifiable encryption. By using our unbalanced optimistic RBC, we can obtain a dual-threshold ACSS protocol with an optimistic good-case latency of $2$ steps and $O(\k n^2 \log{n})$ communication. The resulting protocol requires a setup to enable PKI for verifiable encryption; however, the RBC protocol can still operate without signatures. Such schemes can still be beneficial in avoiding costly signature verification.

\subsection{Setup-free Optimistic ACSS} 
\label{sec:setup-free-acss}
\begin{figure}[ht]
	\footnotesize
	\begin{boxedminipage}[t]{\columnwidth}
    // Sharing phase ($\Sh$)
    \begin{enumerate}[leftmargin=*]
        \item The dealer $\node{d}$ chooses a random symmetric bivariate polynomial $u$ of degree $f$ with $u(0, 0) = s$ i.e.,
        
        $$u(x, y) = \sum_{j, l = 0}^{f} u_{jl} x^j y^l$$
        
        Let $\C = \{C_{j,l}\}$ where $C_{j,l} = g^{u_{j,l}}$ for all $j,l \in [0,f]$. Additionally, let $a_j(y) = u(j, y)$ and $b_j(x) = u(x, j)$. The dealer sends $\sig{\Deal, \C, a_j(y), b_j(x)}$ to party $\node{j}$ $\forall j \in [n]$.

        \item Upon receiving $\sig{\Deal, \C, a_i(y), b_i(x)}$ from $\node{d}$ for the first time, if $\VerifyPoly(\C, i, a, b)$, then send $\sig{\Echo, \C, a_i(j), b_i(j)}$ to $\node{j}$ $\forall j \in [n]$.
        
        \item  Upon receiving $\sig{\Echo, \C, \alpha, \beta}$ from non-dealer party $\node{m}$ for the first time,
        \begin{algorithmic}
        \If{$\VerifyPoint(\C, i, m, \alpha, \beta)$}
        \State $A_{\C} \gets A_{\C} \cup \{(m, \alpha)\}; B_{\C} \gets B_{\C} \cup \{(m, \beta)\}$
        \State $e_{\C} \gets e_{\C} + 1$
        {\color{blue}
        \If{$e_{\C} = \ceil{\frac{n}{2}}$}
        \State interpolate $\bar{a}, \bar{b}$ from $B_{\C}, A_{\C}$ respectively.
        \State send $\sig{\Vote, \C, \bar{a}(j), \bar{b}(j)}$ to $\node{j}$ $\forall j \in [n]$
        \EndIf
        }
        \If{$e_{\C} = \ceil{\frac{n+f-1}{2}}$}
        \State send $\sig{\Ack, \C, \bar{a}(j), \bar{b}(j)}$ to $\node{j}$ $\forall j \in [n]$
        \EndIf
        {\color{blue}
        \If{$e_{\C} = \ceil{\frac{n+2f-2}{2}}$}
        \State $s_i \gets \bar{a}(0)$;  \textbf{terminate} (optimistically)
        \EndIf
        }
        \EndIf
        \end{algorithmic}

        {\color{blue}
        \item Upon receiving $\sig{\Vote, \C, \alpha, \beta}$ from non-dealer party $\node{m}$ for the first time, 
        \begin{algorithmic}
            \If{$\VerifyPoint(\C, i, m, \alpha, \beta)$}
                \State $A_{\C} \gets A_{\C} \cup \{(m, \alpha)\}; B_{\C} \gets B_{\C} \cup {(m, \beta)}$
                \State $v_{\C} \gets v_{\C} + 1$
                \If{$v_{\C} = \ceil{\frac{n+f-1}{2}}$}
                    \State interpolate $\bar{a}, \bar{b}$ from $B_{\C}, A_{\C}$ respectively
            \State send $\sig{\Ack, \C, \bar{a}(j), \bar{b}(j)}$ to $\node{j}$ $\forall j \in [n]$ (if not already sent)
                \EndIf
            \EndIf
        \end{algorithmic}
        }
{\color{blue}
        \item Upon receiving $\sig{\Ack, \C, \alpha, \beta}$ from non-dealer party $\node{m}$ for the first time,
        \begin{algorithmic}
            \If{$\VerifyPoint(\C, i, m, \alpha, \beta)$}
                \State $A_{\C} \gets A_{\C} \cup \{(m, \alpha)\}; B_{\C} \gets B_{\C} \cup {(m, \beta)}$
                \State $a_{\C} \gets a_{\C} + 1$
                \If{$a_{\C} = \ceil{\frac{n+f-1}{2}}$}
                    \State interpolate $\bar{a}, \bar{b}$ from $B_{\C}, A_{\C}$ respectively
            \State send $\sig{\Ready, \C, \bar{a}(j), \bar{b}(j)}$ to $\node{j}$ $\forall j \in [n]$ (if not already sent)
                \EndIf
            \EndIf
        \end{algorithmic}
}
        \item Upon receiving $\sig{\Ready, \C, \alpha, \beta}$ from $\node{m}$ for the first time,
        \begin{algorithmic}    
            \If{$\VerifyPoint(\C, i, m, \alpha, \beta)$}
            \State $A_{\C} \gets A_{\C} \cup \{(m, \alpha)\}; B_{\C} \gets B_{\C} \cup {(m, \beta)}$
            \State $r_{\C} \gets r_{\C} + 1$
            \If{$r_{\C} = f+1$}
            \State interpolate $\bar{a}, \bar{b}$ from $B_{\C}, A_{\C}$ respectively
            \State send $\sig{\Ready, \C, \bar{a}(j), \bar{b}(j)}$ to $\node{j}$ $\forall j \in [n]$ (if not already sent)
            \EndIf
            \If{$r_{\C} = 2f+1$}
            \State $s_i \gets \bar{a}(0)$;  \textbf{terminate}
            \EndIf
            \EndIf
        \end{algorithmic}
    \end{enumerate}
    \vspace{0.4em}
    // Reconstruction phase ($\Rec$)
    \begin{enumerate}[leftmargin=*]
        \setcounter{enumi}{5}
        \item Each party $\node{i}$ sends $\sig{\Reconstruct, s_i}$ to all parties.
        \item Upon receiving $\sig{\Reconstruct, \sigma}$ from $\node{m}$ for the first time, 
        \begin{algorithmic}
            \If{$\VerifyShare(\C, m, \sigma)$}
            \State $\BigS \gets \BigS \cup \{(m, \sigma)\}$
            \State $c \gets c + 1$
            \If{$c = f+1$}
            \State interpolate $a_0$ from $\BigS$ and output $a_0(0)$
            \EndIf
            \EndIf
        \end{algorithmic}
    \end{enumerate}
    \end{boxedminipage}
    \caption{(2,4,5) Setup-free Optimistic ACSS}
    \label{fig:acss}
\end{figure}

We also present a setup-free optimistic ACSS protocol which requires ``white-box'' extension. To achieve this, we extend the setup-free ACSS protocol of Cachin et al.~\cite{cachin2002asynchronous}, which has the DL assumption and a good-case latency of 3 steps. We note that their base protocol has a communication complexity of $O(\k n^4)$, which can be optimized to $O(\k n^3)$ using the modifications they suggested. To the best of our knowledge, the protocol by Cachin et al.~\cite{cachin2002asynchronous} remains the state-of-the-art setup-free ACSS protocol. For the sake of simplicity, we extend their base protocol and then discuss ways to further improve the communication complexity. The details of the protocol are provided in~\Cref{fig:acss}. We also present the helper verification functions (e.g., $\VerifyPoly$, $\VerifyPoint$, $\VerifyShare$) used in the protocol in\full{~\Cref{sec:acss-functions}}{ full version~\cite{shrestha2025optimistic}}.

In our protocol, the dealer $\node{d}$ randomly samples a symmetric bivariate polynomial $u(x, y)$ of degree $f$ such that $u(0, 0) = s$ and commits to it using a two-dimensional commitment $\C$. The dealer then computes the  row and column polynomials $a_j(y)$ and $b_j(x)$  respectively and sends $\sig{\Deal, \C, a_j(y), b_j(x)}$ to each party $\node{j}$ for all $j \in [n]$. Upon receiving $\sig{\Deal, \C, a_i(y), b_i(x)}$, each party $\node{i}$ verifies the polynomial $a_i(y)$ and $b_i(x)$ using the $\VerifyPoly$ function. If the verification succeeds, $\node{i}$ sends an $\Echo$ message containing $\C$ and the evaluations $a_i(j)$ and $b_i(j)$ to each party $\node{j}$ for all $j \in [n]$.

Upon receiving $\sig{\Echo, \C, a_m(i), b_m(i)}$ from party $\node{m}$, party $\node{i}$ verifies the evaluations $a_m(i)$ and $b_m(i)$ against the commitment $\C$ using the $\VerifyPoint$ procedure. An honest party optimistically terminates upon receiving $\ceil{\frac{n+2f-2}{2}}$ valid $\Echo$ messages for $\C$.

To ensure termination when some honest party terminates optimistically, an honest party $\node{i}$ interpolates the polynomials $\bar{a}$ and $\bar{b}$ upon receiving $\ceil{\frac{n}{2}}$ valid $\Echo$ messages for $\C$ from non-dealer parties, and sends $\sig{\Vote, \C, \bar{a}(j), \bar{b}(j)}$ to every party $\node{j}$, for all $j \in [n]$. Furthermore, upon receiving valid $\Echo$ messages for $\C$ from at least $\ceil{\frac{n+f-1}{2}}$ non-dealer parties, it sends $\sig{\Ack, \C, \bar{a}(j), \bar{b}(j)}$ to every party $\node{j}$, for all $j \in [n]$.

Upon receiving $\ceil{\frac{n+f-1}{2}}$ valid $\sig{\Vote}$ messages for
$\C$, party $\node{i}$ interpolates the polynomials
$\bar{a}$ and $\bar{b}$ from the received evaluations, and sends
$\sig{\Ack, \C, \bar{a}(j), \bar{b}(j)}$ to every party $\node{j}$,
for all $j \in [n]$.

Upon receiving $\ceil{\frac{n+f-1}{2}}$ valid $\Ack$ messages for
$\C$, party $\node{i}$ sends $\sig{\Ready, \C, *, *}$ to all parties.
To ensure termination, an honest party also sends
$\sig{\Ready, \C, *, *}$ upon receiving $f+1$ valid
$\sig{\Ready}$ messages, provided it has not already done so.

Finally, party $\node{i}$ terminates upon receiving at least
$2f+1$ valid $\Ready$ messages for $\C$. It then interpolates
$\bar{a}$ from the received evaluations and sets $\bar{a}(0)$ as its share.

\paragraph{Remark on the communication complexity.} In the above protocol, $\C$ has a size of $O(\k n^2)$. Since each message consists of $\C$, the communication complexity of the protocol is $O(\k n^4)$. However, the size of the commitment can be reduced to $O(\k n)$ by relying on hash functions, as presented in Cachin et al.~\cite{cachin2002asynchronous}. We refer the reader to \S3.4 in~\cite{cachin2002asynchronous} for more details on this optimization.

When the $O(\k n)$-sized commitment is included in each message, the resulting communication complexity is $O(\k n^3)$. However, instead of including the full commitment in every message, only the dealer can send the complete $O(\k n)$-sized commitment in the $\Deal$ message, allowing parties to verify the received polynomial before sending their $\Echo$ messages. Subsequent $\Echo$, $\Vote$, $\Ack$ and $\Ready$ messages sent by other parties can then include only an erasure-coded version of the commitment along with Merkle proofs. With this optimization, the optimistic ACSS protocol achieves a communication complexity of $O(\k n^2 \log{n})$.

We present a security analysis in\full{~\Cref{sec:acss-security-proof}}{ full version~\cite{shrestha2025optimistic}}.

\section{Application: Latency-efficient Signature-free DAG-based BFT}
\label{sec:dag_bft}
In this section, we introduce \name, a signature-free DAG-based BFT protocol that guarantees a designated leader in every round and achieves improved latency, committing the leader vertex with a latency of just one RBC plus $1\delta$. Our design builds on the core ideas of Sailfish~\cite{shrestha2024sailfish}, which also designates a leader in every round, but requires use of signatures during failures.

\mypara{Dissecting Sailfish~\cite{shrestha2024sailfish}.}
In DAG-based protocols, a designated ``leader vertex'' is committed first, and the remaining non-leader vertices are ordered when the leader vertex is committed. Thus, the frequency with which leaders are designated and how fast the leader vertices are committed play a crucial role in determining commit latency. Sailfish ensures the presence of a leader in each round and achieves leader vertex commitment with a latency of one RBC plus $1\delta$. Consequently, the non-leader vertices are ordered with a latency of $2$ RBC plus $1\delta$.

A key technique used in Sailfish to support a leader in each round is to enforce one of two conditions on the round $r$ leader vertex: either it must reference the round $r-1$ leader vertex, or it must provide a ``no-vote'' proof demonstrating that a sufficient number of honest parties did not vote for the round $r-1$ leader vertex, thereby proving that it could not have been committed. This constraint is crucial for ensuring the protocol's safety while maintaining leader support in each round.

To commit the round $r$ leader vertex $v_k$ with a latency of one RBC plus $1\delta$, Sailfish ``peeks'' into the RBC for round $r+1$ vertices and commits $v_k$ upon detecting $2f+1$ first messages (of the RBC) referencing $v_k$. This optimization reduces latency by proactively identifying commitment conditions early in the RBC process.

\mypara{Towards a signature-fee DAG-based BFT with a per-round leader and a latency of $1$ RBC plus $1\delta$.} In DAG-based BFT protocols, each party proposes a vertex via RBC, referencing at least $2f+1$ vertices from the previous round. In common-case executions with no failures, no signatures are required when a signature-free RBC is used for vertex proposals. In this respect, Sailfish is already signature-free in the common case. However, in the presence of failures, Sailfish relies on signed ``no-vote'' messages to construct a ``no-vote'' proof, necessitating the use of signatures in such scenarios.

To achieve signature-free DAG-based BFT that supports a leader vertex in each round, we embed the no-vote messages directly within the DAG vertices. Specifically, when a round $r$ vertex $v$ does not reference the round $r-1$ leader vertex, it inherently acts as a no-vote message against that leader vertex. Following Sailfish, we impose a requirement on the round $r$ leader vertex: it must either include a reference to the round $r-1$ leader vertex or provide proof that a quorum of round $r$ vertices did not vote for it. This proof is constructed by referencing at least $2f$ round $r$ non-leader vertices that omit a reference to the round $r-1$ leader vertex. Together with the round $r$ leader vertex, this implies the existence of $2f+1$ round $r$ vertices that do not reference the round $r-1$ leader vertex. Notably, when such a quorum exists, the round $r-1$ leader vertex cannot be committed. Consequently, these references serve as a ``no-vote'' proof, signifying that the round $r-1$ leader vertex could not have been committed.

We inherit the commit rule from Sailfish, which involves peeking into the RBC and observes the first messages of the RBC. This technique enables committing the leader vertex with a latency of 1 RBC plus $1\delta$.

\subsection{Protocol Details}
Our protocol advances through a sequence of consecutively numbered \emph{rounds}, starting from $1$. In each round $r$, a designated leader, denoted as $L_r$, is selected using a deterministic method based on the round number.

\begin{figure*}[!ht]
	\footnotesize
	\begin{boxedminipage}[t]{\textwidth}
		\textbf{Local variables:}
		\setlist{nolistsep}
		\begin{itemize}[noitemsep]
			\item[] struct vertex $v$: \Comment{The struct of a vertex in the DAG}
			      \begin{itemize}[noitemsep]
			      	\item[] $v.round$ - the round of $v$ in the DAG
			      	\item[] $v.source$ - the party that broadcast $v$
			      	\item[] $v.block$ - a block of transactions
			      	\item[] $v.strongEdges$ - a set of vertices in $v.round-1$ that represent strong edges
					\item[] $v.weakEdges$ - a set of vertices in rounds $<$ $v.round-1$ that represent weak edges
                    
                    \item[] {\color{blue} $v.nvEdges$ - a set of at least $2f$ vertices in $v.round$ without paths to the leader vertex of $v.round - 1$, representing no-vote edges}
                    
			      \end{itemize}
			\item[] $DAG_i[] - $ An array of sets of vertices (indexed by rounds)
			\item[] $blocksToPropose$ - A queue, initially empty, $\node{i}$ enqueues valid blocks of transactions from clients
			\item[] $leaderStack \gets $ initialize empty stack
		\end{itemize}

		\begin{algorithmic}[1]
			\Procedure{path}{$v,u$} \Comment{Check if exists a path consisting of strong, weak and no-vote edges in the DAG}
			\State \Return \parbox[t]{\dimexpr\textwidth-\leftmargin-\labelsep-\labelwidth}{exists a sequence of $k \in \N$, vertices $v_1,\ldots, v_k$ s.t. \\
$v_1 = v$, $v_k = u$, and $\forall j \in [2,..,k]: v_j \in \bigcup_{r\ge 1} DAG_i[r] \land (v_j \in v_{j-1}.weakEdges \cup v_{j-1}.strongEdges  \cup  v_{j-1}.nvEdges)$\strut}  
			\EndProcedure
			
			\Procedure{strong\_path}{$v,u$} \Comment{Check if exists a path consisting of only strong edges from $v$ to $u$ in the DAG}
			\State \Return \parbox[t]{313pt}{exists a sequence of $k \in \N$, vertices $v_1,\ldots, v_k$ s.t.\ \\
				$v_1 = v$, $v_k = u$, and $\forall j \in [2,..,k]: v_j \in \bigcup_{r\ge 1} DAG_i[r] \land v_j \in  v_{j-1}.strongEdges$\strut}  
			\EndProcedure
            \algstore{bkbreak}
        \end{algorithmic}
    
    \vspace{0.4em}
    
   \begin{minipage}{0.48\textwidth}
        \begin{algorithmic}
            \algrestore{bkbreak}
            \Procedure{set\_weak\_edges}{$v, r$} \Comment{Add edges to orphan vertices}
            \State $v.weakEdges \gets \{\}$
            \For{$r' = r-2 \text{ down to } 1$}
            \For{\textbf{every} $u \in DAG_i[r']$ s.t.  $\neg$\Call{path}{$v,u$}}
            \State $v.weakEdges \gets v.weakEdges \cup \{u\}$
            \EndFor
            \EndFor
            \EndProcedure

            \vspace{0.4em}
   
			\Procedure{get\_vertex}{$p, r$}
			\If{$\exists v \in DAG_i[r]$ s.t. $v.source = p$}
			\State \Return $v$
			\EndIf
			\State \Return $\bot$
			\EndProcedure
            
            \vspace{0.4em}
			
            \Procedure{get\_leader\_vertex}{$r$}
			\State \Return \Call{get\_vertex}{$L_r, r$} 
			\EndProcedure
   
            \vspace{0.4em}
   			
            \Procedure{a\_bcast$_i$}{$b,r$}
			\State $blocksToPropose.$enqueue($b$)
			\EndProcedure
            \algstore{bkbreak}
        \end{algorithmic}
    \end{minipage}
    \hfill
   \begin{minipage}{0.48\textwidth}
        \begin{algorithmic}
        \algrestore{bkbreak}
            \Procedure{broadcast\_vertex}{$r$}
			\State $v \gets$ \Call{create\_new\_vertex}{$r$}
			\State \Call{try\_add\_to\_dag}{$v$}
			\State \Call{r\_bcast$_i$}{$v, r$}
			\EndProcedure
            
             \vspace{0.4em}

            \Procedure{order\_vertices}{$ $}\label{func:order-vertices}
			\While{$\neg leaderStack.$isEmpty()}
			\State $v \gets leaderStack.$pop()
			\State $verticesToDeliver \gets \{v' \in \bigcup_{r>0} DAG_i[r] \,| \, path(v, v')  \land v' \not \in deliveredVertices\}$
			\For{\textbf{every} $v' \in verticesToDeliver$ in some deterministic order}
			\State \textbf{output} \Call{a\_deliver$_i$}{$v'.block, v'.round, v'.source$}
			\State $deliveredVertices \gets deliveredVertices \cup \{v'\}$
            \EndFor
            \EndWhile
			\EndProcedure
			\algstore{bkbreak}
		\end{algorithmic}
        \end{minipage}
  
	\end{boxedminipage}
	\caption{Basic data structures for \name. The utility functions are adapted from~\cite{spiegelman2022bullshark,shrestha2024sailfish}.}
	\label{fig:data_structures}
\end{figure*}

\paragraph{Basic data structures.} In each round $r$, every party $P_i$ proposes a single vertex $v$, which contains a (potentially empty) block of transactions and references at least $2f+1$ vertices from a previous round. These references form the edges of the DAG. The proposed vertices are disseminated using RBC to prevent equivocation and ensure that all honest parties eventually receive them.

The data structures and utilities used for DAG construction are illustrated in~\Cref{fig:data_structures}. Each party maintains a local copy of the DAG, and while different honest parties may initially have distinct views of it, the reliable broadcast of vertices ensures eventual convergence to a consistent DAG view. The local DAG view for party $P_i$ is represented as $DAG_i$. Each vertex is uniquely identified by a round number and a sender (source). When $\node{i}$ delivers a vertex from round $r$, it is added to $DAG_i[r]$, where $DAG_i[r]$ can contain up to $n$ vertices.

Each vertex $v$ contains two types of outgoing edges: strong edges and weak edges. The strong edges of a vertex $v$ in round $r$ connect to at least $2f+1$ vertices from round $r - 1$, while the weak edges link to at most $f$ vertices from earlier rounds ($< r - 1$) such that no path exists from $v$ to these vertices. A path that follows only strong edges from vertex $v_k$ to $v_\ell$ is referred to as a strong path. Additionally, we introduce a unique field in our protocol, denoted as $v.nvEdges$, which references at least $2f$ vertices from round $r$ that satisfy specific constraints. This field is utilized exclusively by certain special vertices, as explained later.

\begin{figure*}[!ht]
	\footnotesize
	\begin{boxedminipage}[t]{\textwidth}
		\textbf{Local variables:}
		\setlist{nolistsep}
		\begin{itemize}[noitemsep]
			\item[] $round \gets 1;$ $\textit{buffer} \gets \{\}$
		\end{itemize}
		
		\begin{algorithmic}[1]
			\algrestore{bkbreak}
			        \Upon {\Call{r\_deliver$_i$}{$v, r, p$}} \label{line: deliver vertex} \label{line: accept}
        \If {$v.source = p \land v.round = r \land |v.strongEdges| \geq 2f+1$}
        \If {$r>1 \land v.source = L_r \land \not\exists v' \in v.strongEdges$ s.t. $v'.source = L_{r-1}$}

        \State {\color{blue} \textbf{wait until} r\_deliver vertices in $v.nvEdges$  s.t. $\forall v' \in v.nvEdges: \not \exists v'' \in v'.strongEdges$  s.t. $v''.source = L_{r-1}$ }
        
        \EndIf
        
            \If {$\neg$\Call{try\_add\_to\_dag}{$v$}} \label{line: is_valid}
                \State $\textit{buffer} \gets \textit{buffer} \cup \{v\}$
            \Else
                \For {$v' \in \textit{buffer} : v'.round \geq r$}
                    \State \Call{try\_add\_to\_dag}{$v'$}
                \EndFor
            \EndIf
        \EndIf
        \EndUpon
    
    \vspace{0.4em}
    
        \Upon {$|DAG_i[r]| \geq 2f+1 \land (\exists v' \in DAG_i[r] : v'.source = L_r\ \lor$ $2f+1$ $\langle \texttt{timeout}, r \rangle$ are received) for $r \geq round$} \label{line: advance round}

            \If {$P_i = L_{r+1}$}
                \State 
                {\color{blue}
                \textbf{wait until} {$\exists v'\in DAG_i[r]:v'.source=L_{r}$ $\lor$
                (r\_deliver a set $\mathcal{Q}$ of $\ge 2f$ round $r+1$ vertices s.t. $\forall v' \in \mathcal{Q}: \not \exists  v'' \in v'.strongEdges$ s.t. $ v''.source = L_{r}$)}
                }
               
            \EndIf
            \State advance\_round$(r+1)$
        \EndUpon
        \algstore{bkbreak}
        \end{algorithmic}
    
    \vspace{0.4em}
    
   \begin{minipage}{0.48\textwidth}
        \begin{algorithmic}
            \algrestore{bkbreak}

\vspace{0.4em}
        
        \Upon {$timeout$} \label{line: timeout}

        \If {$\not \exists v' \in DAG_i[round] : v'.source = L_{round}$}
        \State multicast $\langle\Timeout, round\rangle$
        \EndIf
        \EndUpon
        \vspace{0.4em}
        \Upon {receiving $f + 1\ \tuple{\Timeout, r}$} \label{line: f+1 timeout}
        \State multicast $\tuple{\Timeout, r}$ if not already sent
        \EndUpon

        \vspace{0.4em}

        \Procedure{create\_new\_vertex}{$r$}
        \State $v.round \gets r$
        \State $v.source \gets \node{i}$
        \State $v.block \gets blocksToPropose$.dequeue()
        \If {$r>1$}
        \State $v.strongEdges \gets DAG_i[r-1]$
        \EndIf
        \If{$\node{i} = L_{r} \land \not \exists v' \in v.strongEdges \text{ s.t. } v'.source =L_{r-1}$}
        \State $v.nvEdges \gets DAG_i[r]$
        \EndIf
        \State \Call{set\_weak\_edges}{$v, r$}
        \State \Return $v$
        \EndProcedure

        \algstore{bkbreak}
        \end{algorithmic}
    \end{minipage}
    \hfill
   \begin{minipage}{0.48\textwidth}
        \begin{algorithmic}
        \algrestore{bkbreak}

        \vspace{0.4em}
        \Procedure {try\_add\_to\_dag}{$v$}
        \If {$\forall v' \in v.strongEdges \cup v.weakEdges : v' \in \bigcup_{k \geq 1} DAG_i[k]$}
            \State $DAG_i[v.round] \gets DAG_i[v.round] \cup \{v\}$
            \If {$|DAG_i[v.round]| \ge f+1$} \label{line: commit2}
                \State \Call{try\_commit}{$v.round-1, DAG_i[v.round], f+1$}
            \EndIf
            \State $\textit{buffer}\gets \textit{buffer}\ \backslash\ \{v\}$
            \State \Return $\True$
        \EndIf
        \State \Return $\False$
        \EndProcedure
        \vspace{0.4em}
    \Procedure {advance\_round}{$r$}
            \State $round\gets r$; $start\ timer$
            \State broadcast\_vertex($round$)
        \EndProcedure
        
         \vspace{0.4em}
         
        \Upon {init}
        \State advance\_round($1$)
        \EndUpon
        \algstore{bkbreak}
		\end{algorithmic}
        \end{minipage}
	\end{boxedminipage}
	\caption{\name: DAG construction protocol for party $\node{i}$}
	\label{fig:dag_construction}
    \vspace{-4mm}
\end{figure*}

\paragraph{DAG construction protocol.} The DAG construction protocol is outlined in~\Cref{fig:dag_construction}. A vertex proposed by $L_r$ is termed the round $r$ leader vertex, while all other round $r$ vertices are classified as non-leader vertices. To propose a vertex in round $r$, each party $\node{i}$ waits to receive at least $2f+1$ round $r-1$ vertices, including the round $r-1$ leader vertex, until a timeout occurs in round $r-1$. If $\node{i}$ successfully receives $2f+1$ round $r-1$ vertices along with the round $r-1$ leader vertex, it can immediately advance to round $r$ and propose a round $r$ vertex (see Line~\ref{line: advance round}). Notably, referencing the round $r-1$ leader vertex functions as a ``vote'' for it. These votes are subsequently used to commit the leader vertex. Therefore, waiting for the leader vertex until a timeout ensures that honest parties contribute votes toward it, helping to commit the leader vertex with minimal latency when the leader is honest (after GST).

If an honest party $\node{i}$ does not receive the round $r-1$ leader vertex when its timer expires, it multicasts $\tuple{\Timeout, r-1}$ to all parties (see Line~\ref{line: timeout}). Additionally, $\node{i}$ sends $\tuple{\Timeout, r-1}$ upon receiving $f+1$ $\tuple{\Timeout, r-1}$ messages, provided it has not already done so (see Line~\ref{line: f+1 timeout}). This process follows the Bracha-style amplification technique, ensuring that when one honest party receives $2f+1$ $\tuple{\Timeout, r-1}$ messages, all honest parties eventually receive them. A non-leader party $\node{i}$ advances to round $r$ once it has received $2f+1$ round $r-1$ vertices and $2f+1$ $\tuple{\Timeout, r-1}$ messages.

When the leader $L_r$ does not receive the round $r-1$ leader vertex, it must gather at least $2f$ round $r$ non-leader vertices that have no path to the round $r-1$ leader vertex before advancing to round $r$. These $2f$ round $r$ vertices, which do not reference the round $r-1$ leader vertex, are recorded in the $v.nvEdges$ field. Along with the round $r$ leader vertex (which also lacks a path to the round $r-1$ leader vertex), they serve as evidence that a sufficient number of honest parties did not ``vote'' for the round $r-1$ leader vertex. Consequently, the round $r-1$ leader vertex cannot be committed. Thus, it is safe for the round $r$ leader vertex to not reference the round $r-1$ leader vertex. This imposes a constraint on the round $r$ leader vertex—either it must have a strong path to the round $r-1$ leader vertex or it must reference a set of $2f$ round $r$ non-leader vertices that did not vote for the round $r-1$ leader vertex.

Upon delivering a round $r$ vertex $v$ (see Line~\ref{line: deliver vertex}), party $\node{i}$ verifies its validity. Specifically, $v$ must include at least $2f+1$ strong edges, and if $v$ is a round $r$ leader vertex, it must either have a strong path to the round $r-1$ leader vertex or reference at least $2f$ round $r$ non-leader vertices that lack a strong path to the round $r-1$ leader vertex. Consequently, $\node{i}$ waits until all vertices in $v.nvEdges$ are delivered and satisfy this condition. If $v$ is deemed valid, it is added to $DAG_i[r]$ via try\_add\_to\_dag($v$), which succeeds once $\node{i}$ has delivered all vertices that establish a path from $v$ in $DAG_i$. If try\_add\_to\_dag($v$) fails, the vertex is placed in a \textit{buffer} for later reattempts. Additionally, when try\_add\_to\_dag($v$) succeeds, any buffered vertices are retried for inclusion in $DAG_i$.

\paragraph{Jumping rounds.} In addition to advancing rounds sequentially, our protocol allows honest parties in round $r' < r$ to ``jump'' directly to round $r$. This occurs when they either observe at least $2f+1$ round $r-1$ vertices, including the round $r-1$ leader vertex, or receive $2f+1$ $\tuple{\Timeout, r}$ messages. If $L_r$ is lagging, it must also wait to receive either the round $r-1$ leader vertex or $2f$ round $r$ vertices that do not reference the round $r-1$ leader vertex before proposing the round $r$ leader vertex. When jumping from round $r'$ to round $r$, parties do not propose vertices for the skipped rounds.

\begin{figure}[t]
 \footnotesize
	\begin{boxedminipage}[t]{\columnwidth}
		\textbf{Local variables:}
		\setlist{nolistsep}
		\begin{itemize}[noitemsep]
			\item[] $committedRound \gets 0$
		\end{itemize}
		\begin{algorithmic}
            \algrestore{bkbreak}

        \Event{receiving a set $\BigS$ of $\ge 2f+1$ first messages for round $r+1$ vertices} \label{line: commit1}
        \State \Call{try\_commit}{$r, \BigS, 2f+1$}
        \EndEvent
    
        \vspace{0.4em}
        \Procedure{try\_commit}{$r, \BigS, threshold$}
            \State $v \gets$ \Call{get\_leader\_vertex}{$r$}
            \State $votes \gets \{v' \in \BigS \mid$ \Call{strong\_path}{$v', v$} $\}$
            \If{$|votes| \ge threshold \land committedRound < r$}
                \State \Call{commit\_leader}{$v$}
            \EndIf
        \EndProcedure

			\vspace{0.4em}
			
			\Procedure{commit\_leader}{$v$}
			\State $leaderStack.$\Call{push}{$v$}
			\State $r \gets v.round - 1$
			\State $v' \gets v$
			\While{$r > committedRound$}
			\State $v_s \gets$ \Call{get\_leader\_vertex}{$r$}
			\If{\Call{strong\_path}{$v', v_s$}}
			\State $leaderStack.$\Call{push}{$v_s$} \label{step:indirect-commit}
			\State $v' \gets v_s$
			\EndIf
			\State $r \gets r -1$
			\EndWhile
			\State $committedRound \gets v.round$
			\State \Call{order\_vertices}{$ $}
			\EndProcedure
		\end{algorithmic}

	\end{boxedminipage}
	\caption{\name: The commit rule for party $\node{i}$}
	\label{fig:commit-rule}
\end{figure}

\paragraph{Committing and ordering the DAG.}
In our protocol, only leader vertices are committed, while non-leader vertices are ordered in a deterministic manner as part of the causal history of a leader vertex. This ordering occurs when the leader vertex is (directly or indirectly) committed, as defined in the order\_vertices function.

Our protocol consists of two distinct commit rules for directly committing the leader vertex (as presented in~\Cref{fig:commit-rule}).

\begin{rrule} \label{commit rule: messages}
    An honest party commits a round $r$ leader vertex when it observes $2f+1$ first messages of RBC for the round $r+1$ vertices that have strong paths to the round $r$ leader vertex. 
\end{rrule}
\begin{rrule} \label{commit rule: vertices}
    An honest party commits a round $r$ leader vertex when it delivers $f+1$ round $r+1$ vertices that have strong paths to the round $r$ leader vertex. (see Line~\ref{line: commit2}).
\end{rrule}

Under~\Cref{commit rule: messages}, $\node{i}$ does not need to wait for the delivery of round $r+1$ vertices. This is because, when the sender of the RBC is honest, the first observed value (i.e., the first message of the RBC) is the same value that will ultimately be delivered. Among the $2f+1$ round $r+1$ vertices, at least $f+1$ originate from honest parties and will eventually be delivered, ensuring that the delivered value matches the first received value (from the first RBC message). Furthermore, since at least $f+1$ honest parties send round $r+1$ vertices with a strong path to the round $r$ leader vertex, it is impossible to form a set of $2f+1$ round $r+1$ vertices that entirely lack a strong path to the round $r$ leader vertex. Consequently, the round $r' > r$ leader vertex (if it exists) will necessarily have a strong path to the round $r$ leader vertex, ensuring the safety of the commit. This rule aligns with the approach used in Sailfish~\cite{shrestha2024sailfish}.

~\Cref{commit rule: vertices} is simplification of~\Cref{commit rule: messages} where $\node{i}$ waits to deliver $f+1$ round $r+1$ vertices that have strong paths to the round $r$ leader vertex. Note that having $f+1$ round $r+1$ vertices that have strong paths to the round $r$ leader vertex implies there cannot exists a set of $2f+1$ round $r+1$ vertices that lack a strong path to the round $r$ leader vertex. Consequently, all leader vertex for round $r' > r$ will have a strong path to the round $r$ leader vertex, ensuring the safety of the commit. %

When directly committing a round $r$ vertex $v_k$ in round $r$, $\node{i}$ traces back through earlier rounds to \textit{indirectly} commit leader vertices $v_m$ if a strong path exists from $v_k$ to $v_m$. This process continues until $P_i$ encounters a round $r' < r$ where it has already directly committed a leader vertex. Our protocol ensures that once any honest party directly commits a round $r$ leader vertex $v_k$, all leader vertices of subsequent rounds $r' > r$ will have a strong path to $v_k$. This guarantees that $v_k$ will eventually be (directly or indirectly) committed by all honest parties.

\paragraph{Commit latencies.} The commit latency of a leader vertex under~\Cref{commit rule: messages} includes the time to propagate round $r$ vertices via RBC, followed by one additional step to receive the first messages for $2f+1$ round $r+1$ vertices—i.e., one RBC plus $1\delta$. Similarly, committing via the second rule incurs a latency of two RBCs. The non-leader vertices require an additional RBC to be committed. Although Commit Rule 2 involves more steps, it requires less messages. In practice, the effectiveness and speed of each rule depend on the actual network configuration.

\paragraph{Remark on timeout parameter $\tau$.}
The timeout parameter $\tau$ must be set sufficiently long to guarantee that, upon entering round $r$, an honest party has enough time to deliver both the round $r$ leader vertex—if broadcast by an honest leader—and at least $2f + 1$ round $r$ vertices before its timer expires. The precise value of $\tau$ depends on the underlying RBC primitive used to disseminate vertices, as well as the specific conditions under which an honest party $\node{i}$ transitions to round $r$. When using Bracha’s RBC~\cite{bracha1987asynchronous}, our protocol requires $\tau = 5\Delta$ if $\node{i}$ enters round $r$ after delivering the leader vertex from round $r - 1$, and $\tau = 8\Delta$ otherwise. We prove in\full{~\Cref{claim: GST strong path}}{ full version~\cite{shrestha2025optimistic}} that these timeout bounds are sufficient.

\paragraph{Handling crashed/offline parties.} Our protocol relies on reliable message delivery to complete RBC instances and build the DAG. While the theoretical model assumes that parties are either Byzantine or fully honest and that honest parties remain perpetually online, in practice, honest nodes may crash or go offline temporarily and later rejoin. During their downtime, messages sent to these nodes may be lost, preventing them from building their local DAG upon rejoining. To address this, a recovery mechanism is necessary to help such nodes catch up and synchronize with the latest system state.

As a simple strategy, rejoining nodes can request missing vertices from other parties when they RBC-deliver vertices in new rounds. To ensure correctness, they only include a vertex in their DAG if it is received identically from at least $f+1$ parties. Since there are at least $2f+1$ honest and online parties that maintain the full DAG state, retrieving missing vertices is straightforward. To further reduce communication overhead, honest parties can respond by sending erasure-coded chunks of the requested vertex. The requesting node can then reconstruct the missing vertex by collecting $f+1$ valid chunks that match a common Merkle root.

We present detailed security analysis in\full{~\Cref{sec:dag_bft_proofs}}{ full version~\cite{shrestha2025optimistic}}.

\section{Evaluation}
\label{sec:evaluation}

We evaluate the performance of \name when instantiated with optimistic RBC, comparing it against Bullshark~\cite{spiegelman2022bullshark} (instantiated with Bracha's RBC) and signature-based Sailfish (instantiated with the 2-step RBC protocol of Abraham et al.~\cite{abraham2021good}). We could not compare with Shoal~\cite{spiegelman2023shoal} and Shoal++~\cite{arunShoalHighThroughput2025a} as their source code was unavailable.

\begin{figure*}[ht]%
\begin{subfigure}[]{0.32\textwidth}
\includegraphics[scale=0.22]{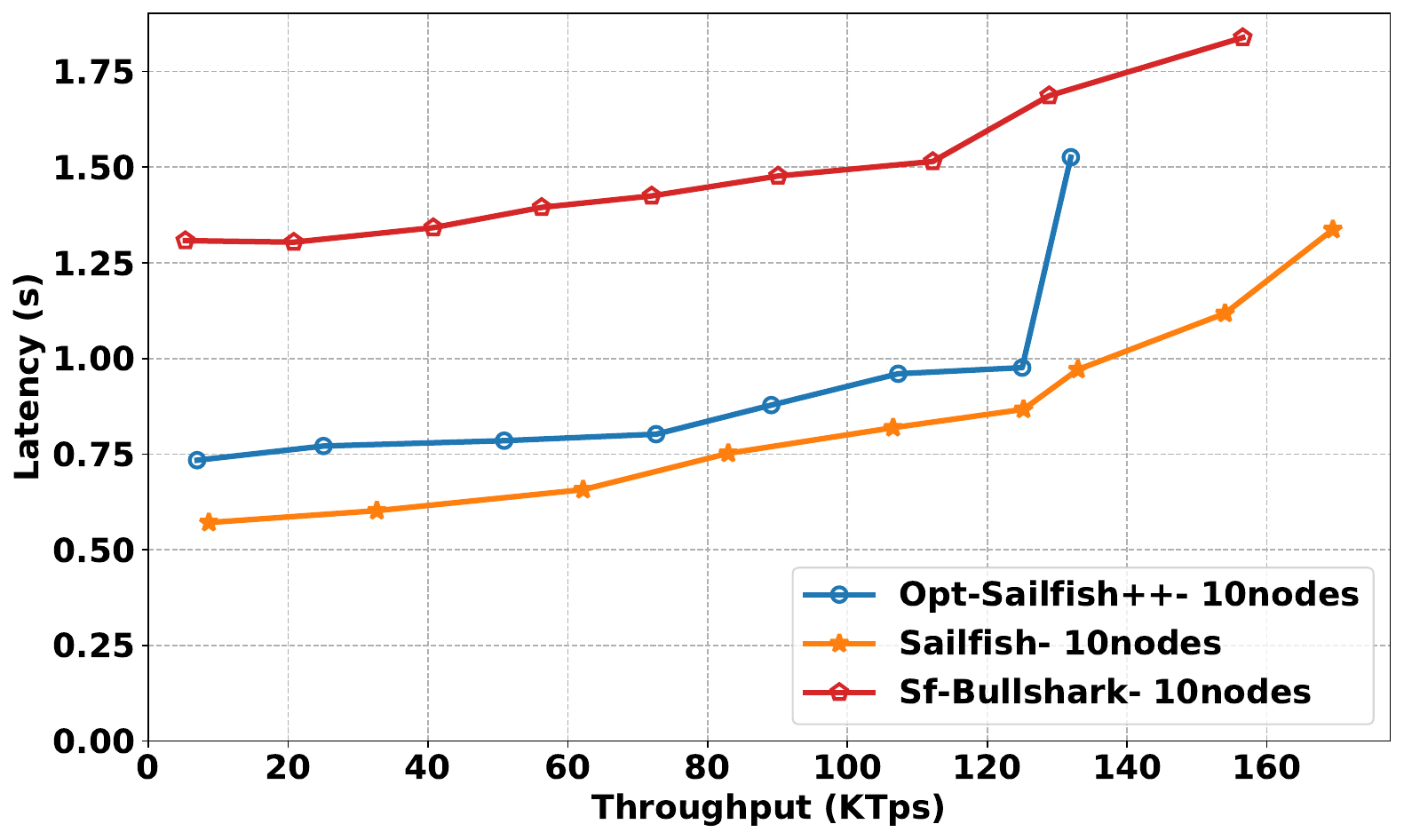}

\end{subfigure}
\hfill
\begin{subfigure}[]{0.32\textwidth}
\includegraphics[scale=0.22]{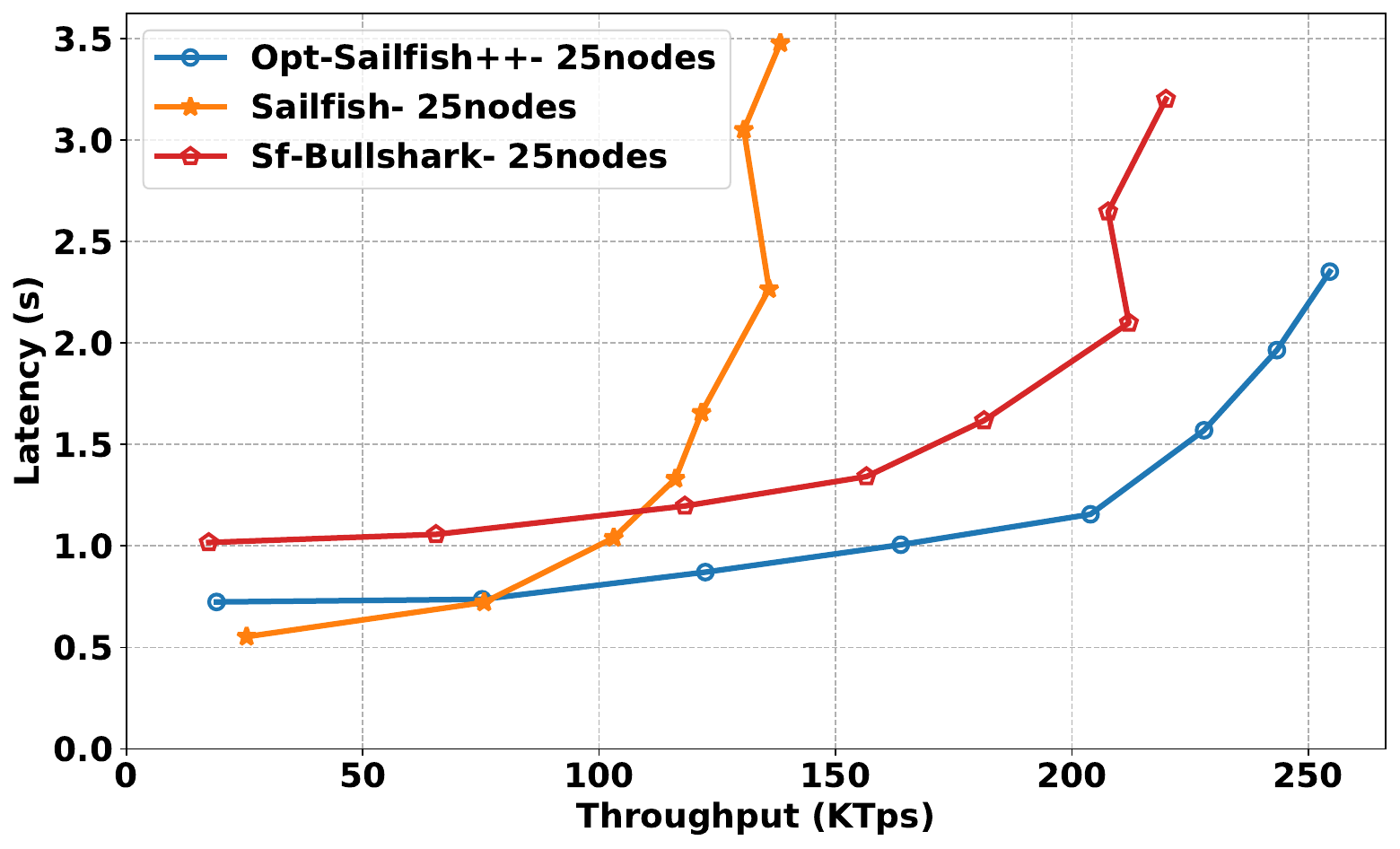}

\end{subfigure}
\hfill
\begin{subfigure}[]{0.32\textwidth}
\includegraphics[scale=0.22]{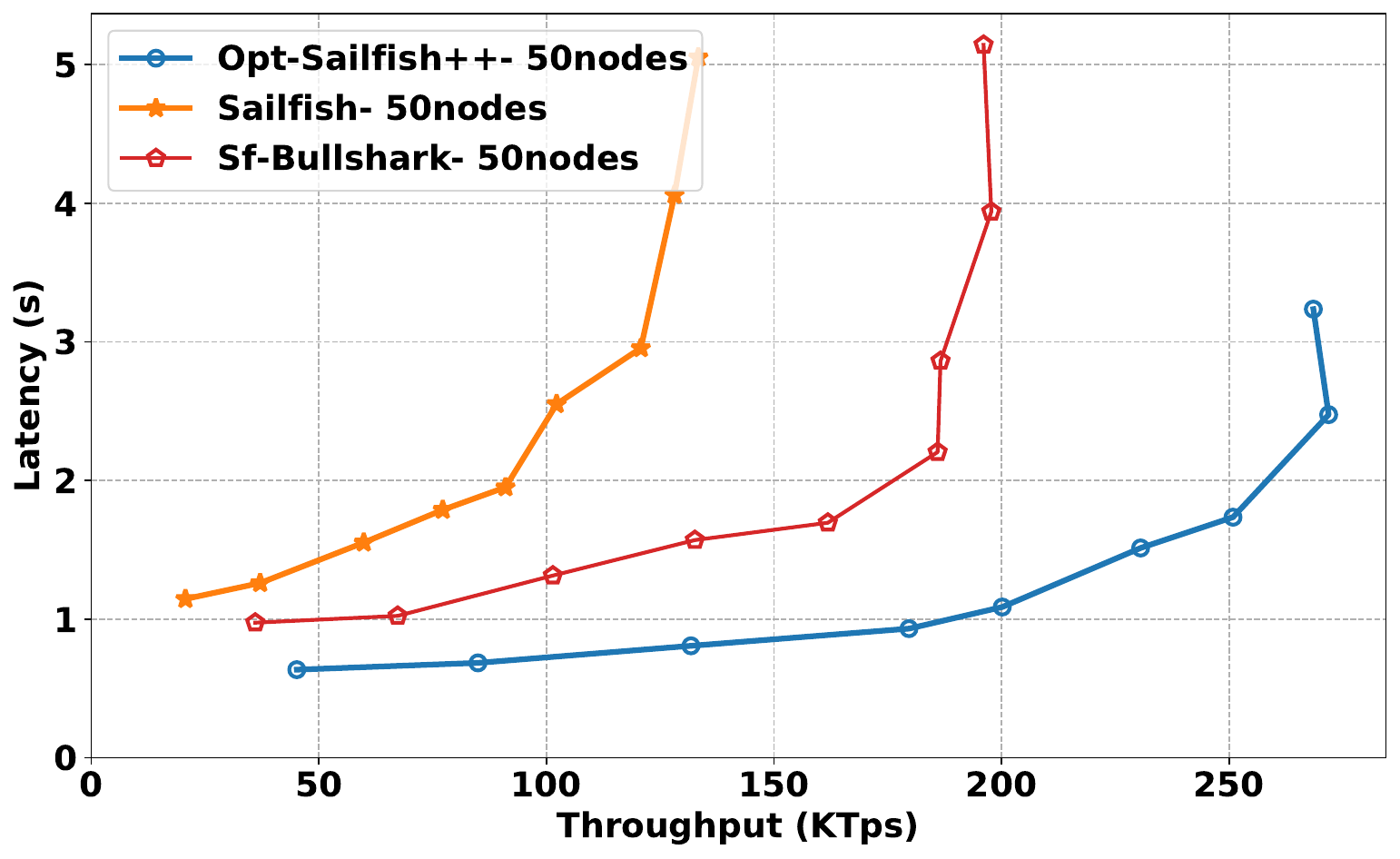}
\end{subfigure}
\caption{Throughput vs. latency at various system sizes and varying input}
\label{fig:thr-lat}
\end{figure*}

\mypara{Implementation details.} Our implementation modifies the RBC primitive from the open-source Sailfish~\cite{sailfish-impl} to create \name. We use Opt-Sailfish++~\cite{opt_sailfish-impl} to represent Sailfish++ using optimistic RBC. Additionally, we modify the RBC primitive from the open-source Bullshark~\cite{bullshark-impl} to Bracha's RBC, resulting in its signature-free variant~\cite{Sf-Bullshark-impl} that we refer to as Sf-Bullshark in our evaluation. For vanilla Sailfish, BLS multi-signatures are used for signing messages. %

\begin{table}[t]
\centering
\caption{Ping latencies (in ms) between GCP regions}
\label{tab:ping-latencies}
\begin{threeparttable}
\begin{tabular}{l|c r r r r}
\toprule
\multicolumn{1}{c|}{} & \multicolumn{5}{c}{\textbf{Destination}*} \\
\midrule
\textbf{Source} & us-e1 & us-w1 & eu-w1 & eu-n1 & as-n1 \\
\midrule
us-east1-b              & 0.62 & 66.30 & 91.38 & 114.63 & 164.05 \\
us-west1-a              & 66.42 & 0.57 & 135.27 & 158.52 & 90.84 \\
europe-west1-b          & 91.45 & 135.27 & 0.63 & 34.45 & 225.50 \\
europe-north1-b         & 114.97 & 158.55 & 34.50 & 0.74 & 246.24 \\
asia-northeast1-a       & 164.00 & 90.80 & 226.04 & 246.36 & 0.55\\
\bottomrule
\end{tabular}
\begin{tablenotes}
\small
\item[*] Region names are abbreviated versions of the source regions.
\end{tablenotes}
\end{threeparttable}
\end{table}

\mypara{Experimental setup.} We conducted our evaluations on the Google Cloud Platform (GCP), deploying nodes evenly across five distinct regions: us-east1-b (South California), us-west1-a (Oregon), europe-west1-b (Belgium), europe-north1-b (Finland) and asia-northeast1-a (Japan). We employed e2-standard-4 instances~\cite{gcp}, each featuring 4vCPUs, 16GB of memory, and up to 8Gbps network bandwidth. All nodes ran on Ubuntu 20.04.

In our evaluations, each party generates a configurable number of transactions ($512$ random bytes each) for inclusion in its vertex, with a vertex containing up to $10,000$ transactions (i.e., $5$ MB). Each experiment runs for 180 seconds, and the data shown in the figures represents the average of three independent runs. The results are observed to be stable across these runs. Latency is measured as the average time between the creation of a transaction and its commit by all non-faulty nodes. Throughput is measured by the number of committed transactions per second. We summarize round-trip latencies between GCP regions in~\Cref{tab:ping-latencies}.

\mypara{On the choice of low-CPU machines and operating costs.}
Existing signature-based DAG-based BFT protocols require high-end machines, typically 48–64 vCPUs per node, even for moderate-scale deployments of 50–100 nodes~\cite{babelMysticetiReachingLatency2025,arunShoalHighThroughput2025a}. This makes production deployments cost-prohibitive. For instance, running 100 nodes with 64 vCPUs each costs approximately \$2.4M/year based on current GCP pricing~\cite{gcp-cost}. In contrast, our protocol is designed to run efficiently on low-CPU machines, performing well even with just 4 vCPUs per node and peaking at only 270\% CPU usage (i.e., 2.7 vCPUs) at $n=50$. The annual cost of running 100 such nodes is just \$117K, making our protocol significantly more cost-effective for production. Accordingly, we focus our evaluation on realistic, low-cost configurations.

\mypara{Methodology.} In our evaluations, we gradually increased the number of input transactions per vertex. As shown in~\Cref{fig:thr-lat}, throughput increases with the load, accompanied by a slight increase in latency, up to a certain point before reaching saturation. Beyond this point, latency begins to rise while throughput either stabilizes or increases marginally.

\mypara{Performance comparison under fault-free case.}
We compare the performance of four different protocols: Opt-Sailfish++, Sailfish, and Sf-Bullshark all under fault-free scenarios across system sizes of 10, 25 and 50 nodes. \Cref{fig:thr-lat} presents the latency achieved by these protocols as throughput varies.

Our implementation of Sailfish adopts an optimistic signature verification strategy: nodes aggregate signatures without verifying each one individually and verify only the final aggregated signature. We retain this optimization, as it aligns with our broader goal of designing optimistic protocols. Without it, Sailfish incurs around 11 seconds of latency at 3,000 TPS with 50 nodes, due to limited per-node CPU capacity. In small deployments (e.g., $n=10$), Sailfish outperforms Sailfish++ and Sf-Bullshark as it uses a 2-round RBC and the signature verification cost is negligible. In contrast, Opt-Sailfish++ waits for more messages, and Sf-Bullshark suffers from high latency for both leaders and non-leaders.

Even in a moderate-sized system with $n=25$, Sailfish begins to underperform relative to both Sailfish++ and Sf-Bullshark as the load increases. At $n=50$, the cost of signature verification dominates the cost of message delays, allowing Sailfish++ to outperform both Sailfish and Sf-Bullshark. Sf-Bullshark exhibits higher latency than both Opt-Sailfish++ in all configurations, due to its higher commit latency for both leaders and non-leaders.

\section{Related Work}
\label{sec:related_work}
A substantial body of work has aimed to improve various aspects of reliable broadcast (RBC). In the signature-free setting, Bracha’s classical RBC protocol~\cite{bracha1987asynchronous} tolerates an optimal $f<n/3$ Byzantine faults and achieves a good-case latency of 3 steps which is optimal~\cite{abraham2022good}. Subsequent works~\cite{abraham2022good,imbs2016trading} reduced the good-case latency to 2 steps by sacrificing resilience, requiring at least $n \ge 4f$. In contrast, our optimistic reliable broadcast protocols tolerate $f < n/3$ Byzantine faults and commit in 2 steps under optimistic conditions, and in 4 steps otherwise.

We note a trivial optimistic RBC variant that commits in 2 steps when all $n$ parties send matching $\Echo$ messages. However, this approach requires all $n$ parties to be responsive, making it fragile i.e., one crash prevents optimistic commit. In contrast, our optimistic RBC tolerates up to $\floor{\frac{n - 2f}{2}}$ faults and requires responses from only $\ceil{\frac{n + 2f - 2}{2}}$ non-broadcaster parties. Moreover, the trivial approach depends on the slowest party, which incurs high latency in geo-distributed environments with variable delays. Our protocol avoids this bottleneck by relying on the fastest subset of parties. While our approach introduces an additional $\Vote$ message, this overhead remains minimal: the payload (proposed value or codewords) can be sent in either the $\Echo$ or $\Vote$ message, with the other carrying only the hash.

\mypara{Communication-efficient balanced RBC for long messages.}
The balanced RBC protocol of Cachin et al.~\cite{cachin2005asynchronous} achieves a good-case latency of 3 steps with $O(nL + \kappa n^2 \log n)$ communication for $L$-bit messages. The state-of-the-art in terms of communication is the work of Alhaddad et al.~\cite{alhaddad2022balanced}, which improves the communication to $O(nL + \kappa n^2)$ but increases the good-case latency to 5 steps. We refer readers to~\cite{alhaddad2022balanced} for a detailed comparison of balanced RBC protocols. Compared to these, our protocol matches the asymptotic complexity of Cachin et al.~\cite{cachin2005asynchronous}, while enabling 2-step optimistic termination.

\mypara{Signature-free DAG-based BFT.} 
To the best of our knowledge, Bullshark's partially synchronous version~\cite{spiegelman2022bullsharkpartially} is the first DAG-based BFT protocol in the partially synchronous and signature-free setting when instantiated with signature-free RBC protocols~\cite{bracha1987asynchronous, cachin2005asynchronous}. Bullshark~\cite{spiegelman2022bullshark} designates leaders every two rounds and commits a leader vertex with a latency of 2 RBCs. However, non-leader vertices that share the same round as the previous leader require an additional 2 RBCs to commit, totaling a latency of 4 RBCs.

Shoal~\cite{spiegelman2023shoal} introduced a ``pseudo-pipelining'' technique to reduce the commit latency of non-leader vertices by running multiple instances of Bullshark sequentially, ensuring a leader vertex in every round. However, its design depends on Bullshark committing a vertex before launching a new instance with a leader in the next round. If Bullshark fails to commit, Shoal requires an additional two RBCs to initiate a new instance, compromising its ability to guarantee a leader vertex in every round. Moreover, Shoal retains the $2$ RBC latency for committing the leader vertex.

Shoal++~\cite{arunShoalHighThroughput2025a} extends Shoal to commit the leader vertex with $1$ RBC + $1\delta$. However, Shoal++ inherits a key limitation from Shoal—it does not guarantee the presence of a leader vertex in every round.

In contrast to these protocols, the \name protocol supports a leader vertex in every round and commits the leader vertex with a latency of one RBC plus $1\delta$, while non-leader vertices require one additional RBC. When paired with our optimistic RBC, \name achieves a 3-step commit latency for leader vertices under optimistic conditions, matching the performance of signature-based DAG-based BFT protocols such as Sailfish~\cite{shrestha2024sailfish}.

We also note several recent DAG-based protocols~\cite{keidar2023cordial,malkhi2023bbca,babelMysticetiReachingLatency2025,polyanskii2025starfish} that rely on best-effort broadcast (BEB) instead of RBC. These protocols require parties to send signed messages at each protocol step, making them inherently reliant on digital signatures and thus unsuitable for signature-free deployment. 

\mypara{Optimistic fast-path in other consensus protocols.}
We also compare several consensus protocols that support an optimistic fast path. In the synchronous network setting, a few BFT SMR protocols—such as Thunderella~\cite{pass2018thunderella} and OptSync~\cite{shrestha2020optimality}—can commit optimistically in 2 steps when at least $\floor{\frac{3n}{4}} + 1$ parties behave honestly. In the asynchronous setting, Kursawe~\cite{kursaweOptimisticByzantineAgreement2002} proposed an optimistic Byzantine agreement protocol tolerating $f < n/3$ Byzantine faults, which can terminate in 2 steps under the optimistic condition that all $n$ parties behave honestly i.e., its fast path does not tolerate any failures. Similarly, in the partially synchronous setting, a few BFT SMR protocols—such as Zyzzyva~\cite{kotlaZyzzyvaSpeculativeByzantine2007}, UpRight~\cite{clementUprightClusterServices2009}, and The Next 700 BFT Protocols~\cite{guerraouiNext700BFT2010a}—support 2-step optimistic commits when all $n$ parties behave honestly.

\begin{acks}
We thank  Chenxu Wang for helpful discussions and for identifying a subtle termination issue in an earlier version of our optimistic RBC protocol.

Xuechao Wang is supported by the Guangzhou-HKUST(GZ) Joint Funding Program (No. 2024A03J0630 and No. 2025A03J3882), the Guangzhou Municipal Science and Technology Project (No. 2025A04J4168), and a gift from Stellar Development Foundation. 
\end{acks}

\bibliographystyle{plain}
\bibliography{main}

\appendix
\section{Extended Preliminaries}
\label{sec:extended-preliminaries}

\subsection{Primitives}
\mypara{Linear erasure codes.} We use standard ($n, b$) Reed-Solomon (RS) codes~\cite{reed1960polynomial}. This code encodes  $b$ data symbols  into code words of $n$ symbols using {\sf ENC} function and can decode the $b$ elements of code words to recover the original data using {\sf DEC} function as explained below.

\begin{itemize}[leftmargin=*]
\item {\sf ENC}.
Given inputs $m_1, \dots, m_{b}$, an encoding function {\sf ENC} computes $(s_1, \dots, s_n)  = {\sf ENC}(m_1, \dots, m_{b})$, where ($s_1, \dots, s_n$) are code words of length $n$. A combination of any $b$ elements of $n$ code words uniquely determines the input message and the remaining of the code word.

\item {\sf DEC}.
The function {\sf DEC} computes $(m_1, \dots, m_{b}) = {\sf DEC}(s_1, ..., s_n)$, and is capable of tolerating up to $c$ errors and $d$ erasures in code words $(s_1, \dots, s_n)$, if and only if $f \ge 2c + d$.
\end{itemize}

\subsection{Verification functions used in ACSS protocol in~\Cref{fig:acss}}
\label{sec:acss-functions}
We adopt the following verification functions from the ACSS protocol of Cachin et al.~\cite{cachin2002asynchronous} in our optimistic ACSS protocol.

\begin{itemize}[leftmargin=*]
\item $\VerifyPoly(\C, i , a, b)$, where $a$ and $b$ are polynomials of degree $f$, i.e., 
$$a(y) = \sum_{l=0}^{f}a_ly^l, \,\,\,\,\,\, b(x) = \sum_{j=0}^{f}b_jy^j$$
The function is satisfied if and only if $g^{a_{l}} = \prod_{j=0}^{f} (C_{jl})^{i^j}$ for all $l \in [0, f]$ and $g^{b_{l}} = \prod_{l=0}^{f} (C_{jl})^{i^l}$ for all $j \in [0, f]$

\item $\VerifyPoint(\C, i, m, \alpha, \beta)$, verifies that the given values $\alpha$, $\beta$ correspond to points $u(m, i)$, $u(i, m)$, respectively committed to $\C$, which $\node{i}$ receives from $\node{m}$ ; it is true if and only if $g^{\alpha} = \prod_{j, l =0}^{f} (C_{j,l})^{m^{j}i^l}$ and $g^{\beta} = \prod_{j,l=0}^{f} (C_{j,l })^{i^jm^l}$

\item $\VerifyShare(\C, m, \sigma)$ verifies $\sigma$ is the valid share of $\node{m}$ with respect to $\C$; it is true if and only if $g^{\sigma} = \prod_{j=0}^t (C_{j0})^{m^l}$.
\end{itemize}

\section{Security Analysis of Optimistic RBC}
\label{sec:opt-rbc-proof}

\begin{claim}\label{clm:no-conflict-vote}
If an honest party optimistically commits $m$, then no honest party will send $\Vote$ or $\Ack$ for any value $m' \not = m$.
\end{claim}
\begin{proof}
For an honest party $\node{j}$ to send $\Vote$ on some another value $m'$, it must receive $\sig{\Echo, m'}$ from at least $\ceil{\frac{n}{2}}$ non-broadcaster parties which, since $\ceil{\frac{n}{2}}>f$, is impossible when the broadcaster is honest. Thus, we consider the case where the broadcaster is Byzantine. In this case, there are at most $(f-1)$ Byzantine non-broadcaster parties.

Assume that an honest party $\node{i}$ optimistically commits to the value $m$. This indicates that $\node{i}$ must have received $\sig{\Echo, m}$ from at least $\ceil{\frac{n+2f-2}{2}}$ distinct non-broadcaster parties. Therefore, at least $\ceil{\frac{n+2f-2}{2}} - (f-1) = \ceil{\frac{n}{2}}$ honest parties must have sent $\sig{\Echo, m}$. Now, for the sake of contradiction, suppose an honest party $\node{j}$ sends a $\Vote$ for another value $m' \neq m$. This would mean that $\node{j}$ must have received $\sig{\Echo, m'}$ from $\ceil{\frac{n}{2}}$ distinct non-broadcaster parties, of which at least $\ceil{\frac{n}{2}} - (f-1) = \ceil{\frac{n-2f+2}{2}}$ must have come from honest parties. Observe that $\ceil{\frac{n}{2}} + \ceil{\frac{n-2f+2}{2}} - (n-f) \ge 1$, implying that at least one honest party must have sent $\Echo$ for both $m$ and $m'$. This is a contradiction. Hence, no honest party can send a $\Vote$ for a conflicting value.

Similarly, by a straightforward quorum intersection argument, it is impossible for $\ceil{\frac{n+f-1}{2}}$ $\Echo$ messages for any value $m' \neq m$ to exist. Consequently, no honest party will send an $\Ack$ message for $m'$ upon receiving 
$\ceil{\frac{n+f-1}{2}}$ $\Echo$ messages for $m'$.

Furthermore, since no honest party sends a $\Vote$ for $m'$, no honest party can  receive $\ceil{\frac{n+f-1}{2}}$ $\Vote$ messages for $m'$. Therefore, no honest party will send an $\Ack$ for $m'$.
\end{proof}

\begin{claim}\label{clm:opt-all-commit}
    If an honest party optimistically commit value $m$, all other parties will commit value $m$.
\end{claim}
\begin{proof}
Assume that an honest party $\node{i}$ optimistically commits to the value $m$. This implies that $\node{i}$ must have received $\sig{\Echo, m}$ from at least $\ceil{\frac{n+2f-2}{2}}$ non-broadcaster parties. Therefore, at least $\ceil{\frac{n}{2}}$ honest non-broadcaster parties must have sent $\sig{\Echo, m}$ to all parties, and thus every honest party eventually receives at least $\ceil{\frac{n}{2}}$ $\sig{\Echo, m}$ messages.

By~\Cref{clm:no-conflict-vote}, no honest party sends a $\Vote$ for any $m' \neq m$. Consequently, all honest parties send $\sig{\Vote, m}$ to all parties and eventually receive at least $\ceil{\frac{n+f-1}{2}}$ distinct $\sig{\Vote, m}$ messages. Again, by~\Cref{clm:no-conflict-vote}, no honest party sends an $\Ack$ for $m' \neq m$. Hence, all honest parties send $\sig{\Ack, m}$ and eventually receive at least $\ceil{\frac{n+f-1}{2}}$ $\sig{\Ack, m}$ messages.

Finally, for an honest party to send $\Ready$ for some value $m' \neq m$,
at least one honest party must receive at least
$\ceil{\frac{n+f-1}{2}}$ $\sig{\Ack, m'}$ messages. However, since no honest
party sends an $\Ack$ for $m'$, no honest party sends $\Ready$ for $m'$.
Moreover, since all honest parties eventually receive at least
$\ceil{\frac{n+f-1}{2}}$ $\sig{\Ack, m}$ messages, they will send $\Ready$
for $m$ and eventually receive at least $2f+1$ $\sig{\Ready, m}$ messages,
ensuring that they commit to the value $m$.
\end{proof}

\begin{claim}\label{clm:no-conflict-ready}
If an honest party sends $\Ready$ for $m$, then no honest party will send $\Ready$ for any value $m' \neq m$.
\end{claim}
\begin{proof}
For an honest party to send $\Ready$ for some value $m'$, at least one honest party must have received $\sig{\Ack, m'}$ from at least $\ceil{\frac{n+f-1}{2}}$ non-broadcaster parties. Similarly, for an honest party to send $\sig{\Ack, m'}$, it must have either (i) received $\sig{\Vote, m'}$ from $\ceil{\frac{n+f-1}{2}}$ non-broadcaster parties, or (ii) received $\sig{\Echo, m'}$ from $\ceil{\frac{n+f-1}{2}}$ non-broadcaster parties.

In the latter case, $\ceil{\frac{n+f-1}{2}} > f$, which implies that at least one honest party must have sent $\sig{\Echo, m'}$ which is impossible when the broadcaster is honest. In the former case, at least one honest party must have sent $\sig{\Vote, m'}$, which occurs only after receiving $\ceil{\frac{n}{2}} > f$ $\sig{\Echo, m'}$ messages which is also impossible when the broadcaster is honest. Therefore, it suffices to consider the case where the broadcaster is Byzantine, in which case there are at most
$f-1$ Byzantine non-broadcaster parties.

For an honest party $\node{i}$ to send $\Ready$ for $m$, at least one honest party must have received $\sig{\Ack, m}$ from at least $\ceil{\frac{n+f-1}{2}}$ non-broadcaster parties. Suppose, for the sake of contradiction, that an honest party $\node{j}$ sends $\Ready$ for some value $m' \neq m$. Then at least one honest party must have received $\sig{\Ack, m'}$ from $\ceil{\frac{n+f-1}{2}}$ non-broadcaster parties.

Therefore, there must exist two sets of at least $\ceil{\frac{n+f-1}{2}}$ non-broadcaster parties sending $\Ack$ for $m$ and $m'$ respectively. Observe that $\ceil{\frac{n+f-1}{2}} + \ceil{\frac{n+f-1}{2}} - (n-1) \ge (f-1)+1$. Therefore, at least one honest party must have sent $\Ack$ for both $m$ and $m'$. A contradiction.
\end{proof}

\begin{theorem}\label{thm:opt-rbc-proof}
The protocol in~\Cref{fig:opt_rbc} solves Byzantine reliable broadcast under asynchrony with optimal resilience $f<n/3$ and optimistic good-case latency of $2$ steps, good-case latency of $4$ steps and has bad-case latency of $5$ steps.
\end{theorem}
\begin{proof}
\mypara{Validity and good-case latency.}
When the broadcaster is honest, it proposes the same value $m$ to all parties.
Consequently, all honest parties send $\Echo$ for $m$, ensuring that every honest
party receives at least $\ceil{\frac{n+f-1}{2}}$ $\Echo$ messages for $m$.
This triggers them to send $\Ack$ for $m$. Subsequently, all honest parties
receive $\sig{\Ack, m}$ from at least $\ceil{\frac{n+f-1}{2}}$ non-broadcaster
parties and send $\Ready$ for $m$. Once every honest party receives at least
$2f+1$ $\Ready$ messages for $m$, they commit to $m$. This protocol therefore
achieves a good-case latency of 4 steps.

Furthermore, if at least $\ceil{\frac{f-1}{2}}$ other parties send $\Echo$ for $m$ to at least one honest party, that party can commit $m$ within $2$ steps, resulting in an optimistic good-case latency of 2 steps.

\mypara{Agreement and bad-case latency.} We first argue that no two honest parties commit to different values. If two honest parties optimistically commit to different values, then there must be at least $\ceil{\frac{n}{2}} + \ceil{\frac{n}{2}} - (n-f) \ge f $ honest parties who sent $\Echo$ for both of these values. This is impossible because honest parties only send $\Echo$ for a single value. Similarly, no two honest parties can commit different value using $4$ or $5$ step commit rule.

Now we show that if an honest party $\node{i}$ optimistically commits $m$ and another honest party $\node{j}$ commits $m'$ using $4$ step commit rule, then it must be that $m = m'$. By~\Cref{clm:no-conflict-vote}, no honest party sends $\Vote$ or $\Ack$ for value $m'$. Thus, no honest party sends $\Ready$ for $m'$ and no honest party commits value $m'$ using $4$ step commit rule. Thus, $m=m'$.

Next, we argue that if an honest party commits value $m$, all honest parties will commit $m$. If the broadcaster is honest, by validity, all honest parties will commit the same value. Now consider the case when the broadcaster is Byzantine. When an honest party $\node{i}$ optimistically commits to $m$. By~\Cref{clm:opt-all-commit}, all honest parties will commit $m$.

Now, consider the case where an honest party commits to $m$ upon receiving $2f+1$ $\Ready$ messages for $m$. This implies that at least $f+1$ honest parties must have sent $\Ready$ for $m$, and these messages will eventually be received by all honest parties. Furthermore, by ~\Cref{clm:no-conflict-ready}, no honest party sends $\Ready$ for any value $m' \neq m$. Consequently, any honest parties that have not yet sent a $\Ready$ message will send $\Ready$ for $m$. As a result, all honest parties will eventually receive $2f+1$ $\Ready$ messages for $m$ and commit $m$. This processs achieves a bad-case latency of $5$ steps.

\mypara{Integrity.} Integrity follows immediately as each honest party commits and outputs a single value.
\end{proof}

\section{Security Analysis of Balanced Optimistic RBC}
\label{sec:bal-rbc-proof}

The proof of ~\Cref{clm:bal-rbc-no-conflict-vote} remains identical to~\Cref{clm:no-conflict-vote} and the proof of~\Cref{clm:bal-rbc-no-conflict-ready} remains identical to~\Cref{clm:no-conflict-ready}.

\begin{claim}\label{clm:bal-rbc-no-conflict-vote}
    If an honest party optimistically commits value $m$, no honest party sends $\Vote$ or $\Ack$ for $h' \not= h$ where $h$ is the Merkle root corresponding to $m$.
\end{claim}

\begin{claim}\label{clm:bal-rbc-no-conflict-ready}
If an honest party sends $\Ready$ for $h$, then no honest party will send $\Ready$ for any value $h' \neq h$.
\end{claim}

\begin{claim}\label{clm:opt-all-commit-v3}
If an honest party optimistically commit value $m$, all other parties will commit value $m$.
\end{claim}
\begin{proof}
When an honest party $\node{i}$ optimistically commits to the value $m$, it must have received valid $\sig{\Echo, s_*, w_*, h}$ messages from at least  $\ceil{\frac{n+2f-2}{2}}$ non-broadcaster parties. Therefore, at least $\ceil{\frac{n}{2}}$ honest non-broadcaster parties must have already sent $\sig{\Echo, s_*, w_*, h}$ to all parties, implying that every honest party  will eventually receive at least $\ceil{\frac{n}{2}}$ such $\sig{\Echo, s_*, w_*, h}$ messages. Note that this is sufficient to decode $m$, since only $\ceil{\frac{n-f+1}{2}}$ valid codewords are required for decoding.

By~\Cref{clm:bal-rbc-no-conflict-vote}, no honest party sends a $\Vote$ for any $h' \neq h$. Consequently, all honest parties send  $\sig{\Vote, s_*, w_*, h}$ to all parties and eventually receive such $\sig{\Vote, s_*, w_*, h}$ messages from at least $\ceil{\frac{n+f-1}{2}}$ non-broadcaster parties. Hence, all honest parties send $\sig{\Ack, h}$ and subsequently receive at least
$\ceil{\frac{n+f-1}{2}}$ $\sig{\Ack, h}$ messages.

Finally, for an honest party to send $\Ready$ for some value $h' \neq h$, at least one honest party must receive at least $\ceil{\frac{n+f-1}{2}}$ $\sig{\Ack, h'}$ messages. However, by~\Cref{clm:bal-rbc-no-conflict-vote}, no honest party sends $\Ack$ for $h'$. Therefore, no honest party sends $\Ready$ for $h'$ and all honest parties sends $\sig{\Ready,h}$ upon receiving $\sig{\Ack, h}$ messages from $\ceil{\frac{n+f-1}{2}}$ non-broadcaster parties and eventually receive at least $2f+1$ $\sig{\Ready, h}$ messages, decode $m$, and commit to it.
\end{proof}

\begin{theorem}
    The protocol in~\Cref{fig:balanced_opt_rbc} solves Byzantine reliable broadcast under asynchrony with optimal resilience $f<n/3$ and optimistic good-case latency of $2$ steps, good-case latency of $4$ steps and has bad-case latency of $5$ steps.
\end{theorem}
\begin{proof}
\mypara{Validity and Good-case Latency.}
When the broadcaster is honest, it sends $\sig{\Propose, s_j, w_j, h}$ to every party $\node{j}$, for all $j \in [n]$. Subsequently, each honest party $\node{i}$ sends $\sig{\Echo, s_i, w_i, h}$, ensuring that every honest party eventually receives valid $\sig{\Echo, s_*, w_*, h}$ messages from at least $\ceil{\frac{n+f-1}{2}}$ non-broadcaster parties, which is sufficient to decode $m$.

This causes all honest parties to send $\sig{\Ack, h}$ and subsequently receive $\sig{\Ack, h}$ from at least $\ceil{\frac{n+f-1}{2}}$ non-broadcaster parties. Consequently, all honest parties send $\sig{\Ready, h}$ and eventually receive at least $2f+1$ $\sig{\Ready, h}$ messages, allowing them to decode $m$ and commit to it.

This process achieves a good-case latency of $4$ communication steps.

Furthermore, if at least $\ceil{\frac{f-1}{2}}$ other parties send $\sig{\Echo, s_*, w_*, h}$ to at least one honest party, that party can decode $m$ and commit it within $2$ steps, resulting in an optimistic good-case latency of 2 steps.

\mypara{Agreement and bad-case latency.} The proof that no two honest parties commit to different values follows identical to~\Cref{thm:opt-rbc-proof}.

Next, we argue that if an honest party commits value $m$, all honest parties will commit $m$. If the broadcaster is honest, by validity, all honest parties will commit the same value. Now consider the case when the broadcaster is Byzantine. In this case, there are at most $(f-1)$ Byzantine non-broadcaster parties. When an honest party $\node{i}$ optimistically commits to $m$. By~\Cref{clm:opt-all-commit-v3}, all honest parties will commit $m$.

Now consider the scenario in which an honest party commits to $m$ upon
receiving at least $2f+1$ $\sig{\Ready, h}$ messages. This implies that at
least $f+1$ honest parties must have already sent $\sig{\Ready, h}$, ensuring
that every honest party will eventually receive at least $f+1$
$\sig{\Ready, h}$ messages. Moreover, by~\Cref{clm:bal-rbc-no-conflict-ready},
no honest party sends $\sig{\Ready}$ for any conflicting value
$h' \neq h$. Therefore, any honest party that has not yet sent a
$\sig{\Ready, h}$ message will eventually do so. Consequently, all honest
parties eventually receive at least $2f+1$ $\sig{\Ready, h}$ messages.

Furthermore, at least one honest party must have sent $\sig{\Ready, h}$ upon
receiving $\sig{\Ack, h}$ from at least $\ceil{\frac{n+f-1}{2}}$ non-broadcaster parties. This, in turn, implies that at least one honest party must have sent $\sig{\Ack, h}$ upon receiving either $\sig{\Echo, s_*, w_*, h}$ or $\sig{\Vote, s_*, w_*, h}$ messages from at least $\ceil{\frac{n+f-1}{2}}$ non-broadcaster parties.

Since there can be at most $f-1$ Byzantine non-broadcaster parties when the
broadcaster is Byzantine, at least $\ceil{\frac{n-f+1}{2}}$ honest parties must have sent $\sig{\Echo, s_*, w_*, h}$ or $\sig{\Vote, s_*, w_*, h}$. Therefore, every honest party eventually receives at least
$\ceil{\frac{n-f+1}{2}}$ valid codewords $(s_*, w_*, h)$, which are sufficient
to decode $m$. Hence, all honest parties eventually commit to $m$ and
terminate.

This process incurs a bad-case latency of $5$ communication steps.

\mypara{Integrity.} Integrity follows immediately as each honest party commits and outputs a single value.
\end{proof}

\begin{lemma}[Communication Complexity]
Let $L$ be the size of the input, $\k$ be the size of hash. The communication complexity of the protocol in~\Cref{fig:balanced_opt_rbc} has a communication complexity of $O(nL+ \k n^2 \log{n})$.
\end{lemma}
\begin{proof}
In the propose step, the broadcaster sends a code word of size $O(L/n)$, a Merkle proof of size $O(\k \log{n})$ and $O(\k)$-sized Merkle root to each party. This results in a communication complexity of $O(L+ \k n \log{n})$. 

During the echo and vote steps, each party sends a code word of size $O(L/n)$, a Merkle proof of size $O(\k \log{n})$ and $O(\k)$-sized Merkle root to all parties. Consequently, these steps incur a communication complexity of $O(nL+ \k n^2 \log{n})$. The $\Ack$ and $\Ready$ steps contributes an additional communication complexity of $O(\k n^2)$. Thus, the total communication complexity is $O(nL+ \k n^2 \log{n})$.
\end{proof}

\section{Security Analysis of Setup-free ACSS}
\label{sec:acss-security-proof}
Under normal conditions, the protocol operates identically to the AVSS protocol of Cachin et al.~\cite{cachin2002asynchronous}. Consequently, its security guarantees are the same as those of the AVSS protocol by Cachin et al.~\cite{cachin2002asynchronous}. Therefore, we only need to argue the termination property of the protocol when an honest party optimistically completes the $\Sh$ phase.

The proof of the following claim is identical to~\Cref{clm:no-conflict-vote}, as honest parties send $\Vote$ and $\Ready$ messages upon receiving the same number of $\Echo$ and $\Vote$ messages.
\begin{claim}\label{clm:acss-no-conflict-vote}
If an honest party optimistically terminates the $\Sh$ protocol for some commitment $\C$, then no honest party sends $\Vote$ or $\Ack$ for a conflicting commitment $\C'$.
\end{claim}

\begin{claim}\label{clm:acss-all-terminate}
If an honest party optimistically terminates the $\Sh$ protocol for some commitment $\C$, then all honest parties will terminate the $\Sh$ protocol for $\C$.
\end{claim}
\begin{proof}
Suppose an honest party $\node{i}$ optimistically terminates the $\Sh$ protocol for some commitment $\C$. This implies that $\node{i}$ must have received valid $\sig{\Echo, \C, *, *}$ from at least $\ceil{\frac{n+2f-2}{2}}$ non-dealer parties. Therefore, at least $\ceil{\frac{n}{2}}$ honest non-dealer parties must have sent $\sig{\Echo, \C, *, *}$ to all parties and all honest parties will receive at least $\ceil{\frac{n}{2}}$ $\sig{\Echo, \C ,* , *}$ messages. By~\Cref{clm:acss-no-conflict-vote}, no honest party sends a $\Vote$ for $\C' \not = \C$. Therefore, all honest parties will send $\Vote$ for $\C$ to all parties and subsequently receive at least $\ceil{\frac{n+f-1}{2}}$ distinct $\sig{\Vote, \C, *, *}$ messages which is sufficient to interpolate $\bar{a}$ and $\bar{b}$. Hence, all honest parties will send $\sig{\Ack, \C, *, *}$ and eventually receive at least $\ceil{\frac{n+f-1}{2}}$ distinct $\sig{\Ack, \C, *, *}$. Subsequently, all honest parties send $\sig{\Ready, \C, *, *}$ and receive $2f+1$ $\sig{\Ready, \C, *, *}$, ensuring that they terminate the $\Sh$ protocol for $\C$.
\end{proof}

\Cref{clm:acss-all-terminate} establishes the termination property when an honest party optimistically completes the sharing phase of the protocol. The proof of following theorem follows similar to~\Cref{thm:opt-rbc-proof}

\begin{theorem}
The protocol in~\Cref{fig:acss} is an ACSS protocol tolerating $f$ Byzantine faults under $n \ge 3f+1$ with an optimistic good-case latency of $2$ steps, a good-case latency of $4$ steps and a bad-case latency of $5$ steps.
\end{theorem}

\section{Extension: Optimistic Asynchronous Verifiable Information Dispersal}
\label{sec:opt-avid}
This section presents $(2, 4, 5)$ optimistic AVID protocol. We begin by formally defining the AVID problem.

\begin{definition}[Asynchronous Verifiable Information Dispersal~\cite{cachin2005asynchronous}]
\label{def:avid}
An AVID scheme for a message $m$ consists of a pair of protocols $\DISPERSE$ and $\RETRIEVE$ which satisfy the following requirements under asynchrony:
\begin{itemize}[noitemsep,leftmargin=*]
    \item[-] \textbf{Termination.} If an honest client invokes $\DISPERSE(m)$, then every honest party eventually finishes the dispersal phase for $m$.
    \item[-] \textbf{Agreement.} If any honest party finishes the dispersal phase on message $m$, all honest parties eventually finish the dispersal phase for message $m$.
    \item[-] \textbf{Availability.} If an honest party has finished the dispersal phase, and some honest client initiates $\RETRIEVE$, then the client eventually reconstructs some message $m'$
    \item[-] \textbf{Correctness.} If an honest party has finished the dispersal phase, then honest clients always reconstruct the same message $m'$ when invoking $\RETRIEVE$. Furthermore, if an honest client invoked $\DISPERSE(m)$, then $m' = m$.
\end{itemize}
\end{definition}

We define latency in AVID by the number of steps required to complete the dispersal phase, distinguishing between optimistic good-case, good-case, and bad-case latency.

In AVID, the client is a separate entity distinct from the $n$ parties in the system. As a result, unlike in RBC, there is no differentiation between broadcaster and non-broadcaster parties within the $n$ parties, and there are always $f$ Byzantine parties. This contrasts with RBC, where the broadcaster is included among the $n$ parties, and if the broadcaster is Byzantine, only $f-1$ non-broadcaster parties are Byzantine. Consequently, the quorum size required in AVID is slightly larger than in RBC.

\begin{figure}[ht]
	\footnotesize
	\begin{boxedminipage}[t]{\columnwidth}
        Refer to~\Cref{fig:balanced_opt_rbc} for verify\_interpolation function.

        // Dispersal phase $(\DISPERSE)$
		\begin{enumerate}[leftmargin=*]

			\item The dispersing client with input $m$ computes $M := [s_1, \ldots, s_n] := \Enc(m, n, \ceil{\frac{n-f+1}{2}})$. The client then sends  $\sig{\Disperse, s_j, w_j, h}$ to party $\node{j}$ $\forall j \in [n]$, where $h$ is the Merkle root over $M$ and $w_j$ is the Merkle proof of $s_j$.
			
			\item \textbf{Echo.} Upon receiving a valid $\sig{\Disperse, s_i, w_i, h}$, party $\node{i}$ sends $\sig{\Echo, s_i, w_i, h}$ to all parties.

            {\color{blue}
			\item \textbf{Vote.} Upon receiving valid $\sig{\Echo, s_*, w_*, h}$ from $\ceil{\frac{n+1}{2}}$ parties, let $M =$ verify\_interpolation($\{s_*\}, h$). If $M \not = \bot$, send $\sig{\Vote, s_i, w_i, h}$ to all parties.
            }

			{\color{blue}
			\item \textbf{Ack.} Party $\node{i}$ sends $\sig{\Ack, h}$ message to all parties (if not already sent), under either of the following conditions:
			\begin{itemize}[leftmargin=*]
				\item[-] Upon receiving valid $\sig{\Echo, s_*, w_*, h}$ from $\ceil{\frac{n+f+1}{2}}$ parties, if $M =$ \Call{verify\_interpolation}{$\{s_*\}$, h} and  $M \not = \bot$,
				
				\item[-] Upon receiving valid $\sig{\Vote, s_*, w_*, h}$ from $\ceil{\frac{n+f+1}{2}}$ parties, if $M =$ \Call{verify\_interpolation}{$\{s_*\}$, h} and  $M \not = \bot$,
			\end{itemize}
			}

			\item \textbf{Ready.}  Party $\node{i}$ sends a $\Ready$ message to all parties (if not already sent), under either of the following conditions:
 
			      \begin{itemize}[leftmargin=*]
			      	\item[-] Upon receiving $\sig{\Ack, h}$ from $\ceil{\frac{n+f+1}{2}}$ parties, send $\sig{\Ready, h}$
                    
			      	\item[-] Upon receiving $f+1$ $\sig{\Ready, h}$, send $\sig{\Ready, h}$
			      \end{itemize}
                  
            {\color{blue}
			\item \textbf{Opt Output.} Upon receiving $\sig{\Echo, s_*, w_*, h}$ from $\ceil{\frac{n+2f+1}{2}}$ distinct parties, party $\node{i}$ decodes $m$ and \textbf{output} $(s_i, w_i, h)$.
            }
            
			\item \textbf{Output.} Upon receiving $2f+1$ $\sig{\Ready, h}$ messages, wait for $\ceil{\frac{n-f+1}{2}}$ $\sig{\Echo, s_*, w_*, h}$ or $\ceil{\frac{n-f+1}{2}}$ $\sig{\Vote, s_*, w_*, h}$, decode $m$ and \textbf{output} $(s_i, w_i, h)$.
		\end{enumerate}
        
        \vspace{0.5em}
        
        // Retrieval phase ($\RETRIEVE$)
        \begin{enumerate}[leftmargin=*]
        \setcounter{enumi}{6}
            \item Upon receiving $\sig{\Retrieve, h}$ from the client, send $\sig{\Symbol, s_i, w_i, h}$ to the client. \Comment{// Server \node{i}}
            \item Upon receiving $\sig{\Symbol, s_*, w_*, h}$ from $ \ceil{\frac{n-f+1}{2}}$ parties, decode and output $m$. \Comment{// Client}
        \end{enumerate}

	\end{boxedminipage}
	\caption{(2, 4, 5) Optimistic AVID}
	\label{fig:opt_avid}
\end{figure}

\mypara{Protocol details.} Our $(2, 4, 5)$ optimistic AVID protocol is presented in~\Cref{fig:opt_avid}. We build upon the AVID protocol of Cachin et al.~\cite{cachin2005asynchronous}, incorporating modifications that enable the dispersal phase to optimistically complete in $2$ steps under optimistic conditions.

In the dispersal phase, the client encodes the input message $m$ using $(n, \ceil{\frac{n-f+1}{2}})$-RS codes to generate the codewords $M := [s_1, \ldots, s_n]$. It then sends $\sig{\Disperse, s_j, w_j, h}$ to each party $\node{j}$ for all $j \in [n]$, where $w_j$ is the Merkle proof for the $j^{th}$ codeword, and $h$ is the Merkle root over $M$. Upon receiving the first $\sig{\Propose, s_i, w_i, h}$ from the client, each party $\node{i}$ sends $\sig{\Echo, s_i, w_i, h}$ to all parties if $w_i$ is a valid Merkle proof for root $h$.

Upon receiving valid $\sig{\Echo, s_*, w_*, h}$ from at least $\ceil{\frac{n+1}{2}}$ parties, party $\node{i}$ verifies whether the client correctly computed the Merkle root using  the encoding of $m$ via verify\_interpolation function (defined in~\Cref{fig:balanced_opt_rbc}). If the verification succeeds, it sends $\sig{\Vote, s_i, w_i, h}$ to all parties. 

An honest party $\node{i}$ decodes $m$, optimistically completes the dispersal phase for $m$ and stores $(s_i, w_i, h)$ upon receiving valid $\sig{\Echo, s_*, w_*, h}$ messages from at least $\ceil{\frac{n+2f+1}{2}}$ parties. Observe that if an honest party optimistically completes the dispersal phase for $m$, all honest parties receive  $\ceil{\frac{n+1}{2}}$ valid $\sig{\Echo, s_*, w_*, h}$ messages. To help with agreement, parties send  $\sig{\Ack, h}$ message to all parties upon receiving at least $\ceil{\frac{n+f+1}{2}}$ $\sig{\Vote, s_*, w_*, h}$ messages.

$\node{i}$ also sends a $\sig{\Ack, h}$ message to all parties upon receiving $\sig{\Echo, s_*, w_*, h}$ from at least $\ceil{\frac{n+f+1}{2}}$ parties if verify\_interpolation succeeds. 

An honest party sends $\sig{\Ready, h}$ to all parties upon receiving
$\sig{\Ack, h}$ from at least $\ceil{\frac{n+f+1}{2}}$ parties. To ensure agreement, an honest party also sends $\sig{\Ready, h}$ upon receiving $f+1$
$\sig{\Ready, h}$ messages from other parties, provided it has not already
done so.

Finally, party $\node{i}$ completes the dispersal phase for $m$ and stores $(s_i, w_i, h)$ upon receiving at least $2f+1$ $\sig{\Ready, h}$ messages.

During the retrieval phase, the client sends $\sig{\Retrieve, h}$ to all parties. In response, each party $\node{i}$ sends $\sig{\Symbol, s_i, w_i, h}$ back to the client. The client successfully decodes $m$ once it has received at least $\ceil{\frac{n-f+1}{2}}$ valid $\Symbol$ messages containing codewords that are consistent with $h$.

\mypara{Dispersal, retrieval and storage cost.} The optimistic AVID protocol (in~\Cref{fig:opt_avid}) builds on the AVID protocol of Cachin et al.~\cite{cachin2005asynchronous}, inheriting the same client dispersal cost of $O(L + \kappa n \log n)$, asymptotically same server-side dispersal cost of $O(nL + \kappa n^2 \log n)$, and the same total storage cost of $O(L + \kappa n \log n)$. However, due to the additional $\Vote$ messages—which include codewords of size $O(L/n)$—our protocol incurs a higher concrete dispersal cost, approximately twice that of Cachin et al.~\cite{cachin2005asynchronous}.

\mypara{Reducing server-side dispersal cost to $O(L + \k n^2 \log{n})$.} In DispersedLedger~\cite{yang2022dispersedledger}, the client sends each party $\node{j}$ its codeword $s_j$, Merkle proof $w_j$, and Merkle root $h$. Upon validating the Merkle proof, each party $\node{i}$ sends $\Echo$, $\Ack$, and $\Ready$ messages containing only the root $h$. Parties neither exchange codewords nor verify the correctness of the encoding. The dispersal phase completes once a party receives $2f+1$ $\sig{\Ready, h}$ messages, reducing the server-side dispersal cost to $O(L + \kappa n^2 \log n)$.

However, this design allows a malicious client to complete dispersal with inconsistent input. To ensure correct retrieval, a client must verify the reconstructed data using \Call{verify\_interpolation}{} function and output $\bot$ if verification fails. As a result, the protocol offers slightly weaker availability guarantees.

Our protocol can also leverage this approach to achieve the same server-side cost of $O(L + \kappa n^2 \log n)$ while additionally providing 2-step optimistic good-case latency. Moreover, it does not significantly increase the concrete dispersal cost, as parties only send a $\Vote$ message on the Merkle root.

\mypara{Related work on AVID.}
We compare existing AVID protocols for an input of $L$ bits. The AVID protocol of Cachin et al.~\cite{cachin2005asynchronous} incurs a client dispersal cost of $O(L + \k n \log n)$, a server-side dispersal cost of $O(nL + \k n^2 \log n)$, and a total storage cost of $O(L + \k n \log n)$, with a good-case latency of 3 steps. DispersedLedger~\cite{yang2022dispersedledger} improves the server-side cost to $O(L + \kappa n^2 \log n)$ by allowing dispersal to complete even with inconsistent inputs, which may result in the client outputting $\bot$ during retrieval—offering slightly weaker availability guarantees. In the same weak availability setting, Alhaddad et al.~\cite{alhaddad2022asynchronous,Alhaddad_2022} propose an AVID protocol with a lower client cost of $O(L + \kappa n)$, server-side cost of $O(L + \kappa n^2)$, and total storage of $O(L + \kappa n)$, but with a 5-step good-case latency. We refer readers to Alhaddad et al.~\cite{alhaddad2022asynchronous,Alhaddad_2022} for a comprehensive comparison of AVID protocols. Compared to these works, our optimistic AVID protocol can match the dispersal and storage efficiency of Cachin et al.~\cite{cachin2005asynchronous} and DispersedLedger~\cite{yang2022dispersedledger}, while offering a 2-step optimistic good-case latency.

\subsection{Security Analysis of Optimistic AVID}
\label{sec:avid-proof}
\begin{claim}\label{clm:avid-no-conflict-vote}
If an honest party optimistically completes the dispersal phase for $m$, then no honest party will send $\Vote$ or $\Ack$ for any other value $h' \not = h$ where $h$ is the Merkle root corresponding to value $m$.
\end{claim}
\begin{proof}
Suppose an honest party $\node{i}$ optimistically completes the dispersal phase for value $m$. This indicates that $\node{i}$ must have received $\sig{\Echo, s_*, w_*, h}$ from at least $\ceil{\frac{n+2f+1}{2}}$ distinct parties. Therefore, at least $\ceil{\frac{n+2f+1}{2}} - f = \ceil{\frac{n+1}{2}}$ honest parties must have sent $\Echo$ messages. Now, for the sake of contradiction, suppose an honest party $\node{j}$ sends a $\Vote$ for another value $h' \neq h$. This would mean that $\node{j}$ must have received $\sig{\Echo, s'_*, w'_*, h'}$ from $\ceil{\frac{n+1}{2}}$ distinct parties, of which at least $\ceil{\frac{n+1}{2}} - f = \ceil{\frac{n-2f+1}{2}}$ must have come from honest parties. Observe that $\ceil{\frac{n+1}{2}} + \ceil{\frac{n-2f+1}{2}} - (n-f) \ge 1$, implying that at least one honest party must have sent $\Echo$ for both $h$ and $h'$. This is a contradiction. Hence, no honest party can send a $\Vote$ for a conflicting value.

Similarly, by a straightforward quorum intersection argument, it is impossible for $\ceil{\frac{n+f+1}{2}}$ $\Echo$ messages for $h'$ to exist. As a result, no honest party will send $\Ack$ for $h'$ upon receiving $\ceil{\frac{n+f+1}{2}}$ $\Echo$ for $h'$. Furthermore, since no honest party will send $\Vote$ for $h'$, no honest party can receive $\ceil{\frac{n+f+1}{2}}$ $\Vote$ for $h'$. Consequently, no honest party will send $\Ack$ for $h'$.
\end{proof}

\begin{claim}\label{clm:avid-no-conflict-ready}
If an honest party sends $\Ready$ for $h$, then no honest party will send $\Ready$ for any value $h' \neq h$.
\end{claim}
\begin{proof}
For an honest party $\node{i}$ to send $\Ready$ for $h$, at least one honest party must have received $\sig{\Ack, h}$ from at least $\ceil{\frac{n+f+1}{2}}$ parties. Suppose, for the sake of contradiction, that an honest party $\node{j}$ sends $\Ready$ for some value $h' \neq h$. Then at least one honest party must have received $\sig{\Ack, h'}$ from $\ceil{\frac{n+f+1}{2}}$ parties.

Therefore, there must exist two sets of at least $\ceil{\frac{n+f+1}{2}}$  parties sending $\Ack$ for $h$ and $h'$ respectively. Observe that $\ceil{\frac{n+f+1}{2}} + \ceil{\frac{n+f+1}{2}} - n \ge f+1$. Therefore, at least one honest party must have sent $\Ack$ for both $h$ and $h'$. A contradiction.    
\end{proof}

\begin{claim}\label{clm:avid-opt-all-complete}
    If an honest party optimistically completes the dispersal phase for value $m$, all honest parties will complete the dispersal phase for $m$.
\end{claim}
\begin{proof}
When an honest party $\node{i}$ optimistically completes the dispersal phase for value $m$, it must have received valid codeword $\sig{\Echo, s_*, w_*, h}$ from at least $\ceil{\frac{n+2f+1}{2}}$ parties. Therefore, at least $\ceil{\frac{n+1}{2}}$ honest parties must have sent $\sig{\Echo, s_*, w_*, h}$ to all parties and all honest parties will receive at least $\ceil{\frac{n+1}{2}}$ $\sig{\Echo, s_*, w_*, h}$ messages. Note that this is sufficient for decoding $m$, as only $\ceil{\frac{n-f+1}{2}}$ valid codewords are required to decode $m$. By~\Cref{clm:avid-no-conflict-vote}, no honest party sends a $\Vote$ for $h' \not = h$. Therefore, all honest parties will send $\Vote$ for $h$ to all parties and subsequently receive $\sig{\Vote, s_*, w_*, h}$ from at least $\ceil{\frac{n+f+1}{2}}$  parties. Hence, all honest parties will send $\sig{\Ack, h}$ and eventually receive at least $\ceil{\frac{n-f+1}{2}}$ $\sig{\Ack, h}$ messages. 

Finally, for an honest party to send $\Ready$ for some value $h' \neq h$, at least one honest party must receive at least $\ceil{\frac{n+f+1}{2}}$ $\sig{\Ack, h'}$ messages. However, since no honest party sends an $\Ack$ for $h'$, no honest party sends $\Ready$ for $h'$. Moreover, since all honest parties eventually receive at least $\ceil{\frac{n+f-1}{2}}$ $\sig{\Ack, h}$ messages, they will send $\Ready$ for $h$ and eventually receive at least $2f+1$ $\sig{\Ready, m}$ messages, decode $m$, thereby completing the dispersal phase for $m$.
\end{proof}

\begin{theorem}
The protocol in~\Cref{fig:opt_avid} is an AVID protocol with an optimistic good-case latency of $2$ steps, a good-case latency of $3$ steps and a bad-case latency of $4$ steps.
\end{theorem}
\begin{proof}
\mypara{Termination and good-case latency.} When the client is honest, it sends valid and consistent codewords $s_j$ to every party $\node{j}$, for all $j \in [n]$. Consequently, each honest party $\node{i}$ sends $\sig{\Echo, s_i, w_i, h}$ to all parties, ensuring that every honest party eventually receives at least $\ceil{\frac{n+f+1}{2}}$ valid $\sig{\Echo}$ messages corresponding
to $h$.

This triggers all honest parties to send $\sig{\Ack, h}$, causing every
honest party $\node{i}$ to eventually receive at least
$\ceil{\frac{n+f+1}{2}}$ $\sig{\Ack, h}$ messages. As a result, all
honest parties send $\sig{\Ready, h}$ and eventually receive at least
$2f+1$ $\sig{\Ready, h}$ messages, enabling them to decode $m$,
output $(s_i, w_i, h)$, and terminate. This process achieves a good-case latency of $4$ communication steps.

Furthermore, if at least $\ceil{\frac{f}{2}}$ other parties send $\Echo$ for $h$ to at least one honest party, that party can complete the disperal phase for $m$ within $2$ steps, achieving an optimistic good-case latency of 2 steps.

\mypara{Agreement and bad-case latency.}
If an honest party optimistically completes the dispersal phase on message $m$, then by~\Cref{clm:avid-opt-all-complete}, all honest parties will eventually finish the dispersal phase for message $m$. 

Next, consider the scenario when an honest party completes the dispersal phase for $m$ upon receiving $2f+1$ $\sig{\Ready, h}$ messages. This indicates that at least $f+1$ honest parties must have sent $\Ready$ for $h$, and these messages will eventually be received by all honest parties. Additionally,by~\Cref{clm:avid-no-conflict-ready}, honest parties cannot send $\Ready$ for a conflicting value $h' \neq h$.  As a result, all honest parties will send $\sig{\Ready, h}$ and eventually receive $2f+1$ $\Ready$ messages for $h$. Furthermore, at least one honest party must have send $\Ack$ upon receiving $\ceil{\frac{n+f+1}{2}}$ $\sig{\Echo, s_*, w_*, h}$ messages or $\ceil{\frac{n+f+1}{2}}$ $\sig{\Vote, s_*, w_*, h}$ messages. Therefore, all honest parties will receive at least $\ceil{\frac{n-f+1}{2}}$ $\sig{\Echo, s_*, w_*, h}$ messages or   $\ceil{\frac{n-f+1}{2}}$ $\sig{\Vote, s_*, w_*, h}$ messages which is sufficient to decode $m$. Therefore, all honest parties will complete the dispersal phase for $m$. This process achieves a bad-case latency of $4$ steps.

\mypara{Availability.}
If an honest party has completed the dispersal phase, then by the agreement property, all honest parties will eventually complete the dispersal phase. This guarantees that each honest party $\node{i}$ has stored $(s_i, w_i, h)$ for some message $m'$. Consequently, when a client issues a $\sig{\Retrieve, h}$ request, all honest parties will respond with their stored values. As a result, the client will receive at least $\ceil{\frac{n-f+1}{2}}$ codewords sufficient to decode $m'$.

\mypara{Correctness.}
First, we argue that no two honest parties can complete the dispersal phase for different values. This ensures that honest clients will always reconstruct the same message $m'$ when invoking $\RETRIEVE$.

If two honest parties optimistically complete the dispersal phase for different values, then at least $\ceil{\frac{n+1}{2}} + \ceil{\frac{n+1}{2}} - (n-f) \ge f+1$ honest parties must have sent $\Echo$ messages for both values. However, this is impossible since honest parties send an $\Echo$ message for only a single value. By a similar argument, no two honest parties can complete the dispersal phase for different values upon receiving $2f+1$ $\Ready$ messages.

Next, we argue that if an honest party $\node{i}$ optimistically completes the dispersal phase for $m$ and another honest party $\node{j}$ completes the dispersal phase for $m'$ upon receiving $2f+1$ $\Ready$ messages for $h'$, then it must be that $m = m'$. By~\Cref{clm:avid-no-conflict-vote}, no honest party sends a $\Vote$ or $\Ack$ message for a conflicting value $m'$. Consequently, no honest party will send $\Ready$ for $h'$ and commit to $m'$ upon receiving $2f+1$ $\Ready$ messages for $h'$.

Finally, we argue that if an honest client invokes $\DISPERSE(m)$, then it must be that $m' = m$. This follows from the termination property, which ensures that all honest parties complete the dispersal phase for $m$. Consequently, honest clients will always reconstruct the same message $m$ when invoking $\RETRIEVE$.
\end{proof}

\section{Security Analysis of \name}
\label{sec:dag_bft_proofs}
We say that a \textit{leader vertex $v_i$ is committed directly} by party $\node{i}$ if $\node{i}$ invokes \Call{commit\_leader}{$v_i$}. Similarly, we say that a \textit{leader vertex $v_j$ is committed indirectly} if it is added to $leaderStack$ in Line~\ref{step:indirect-commit}. In addition, we say party $\node{i}$ consecutively directly commit leader vertices $v_k$ and $v_{k'}$ if $\node{i}$ directly commits $v_k$ and $v_{k'}$ in rounds $r$ and $r'$ respectively and does not directly commit any leader vertex between $r$ and $r'$.

The following fact is immediate from using reliable broadcast to propagate a vertex $v$ and waiting for the entire causal history of $v$ to be added to the DAG before adding $v$.

\begin{fact}\label{fact}
For every two honest parties \( P_i \) and \( P_j \) (i) for every round \( r \), \( \bigcup_{r' \leq r} DAG_i[r'] \) is eventually equal to \( \bigcup_{r' \leq r} DAG_j[r'] \), (ii) At any given time \( t \) and round \( r \), if \( v \in DAG_i[r] \) and  \( v' \in DAG_j[r] \) such that \( v.source = v'.source \), then \( v = v' \). Moreover, for every round \( r' < r \), if \( v'' \in DAG_i[r'] \) and there is a path from \( v \) to \( v'' \), then \( v'' \in DAG_j[r'] \) and there is a path from \( v' \) to \( v'' \).
\end{fact}

\begin{claim} \label{claim: strong path} 
    If an honest party \( P_i \) directly commits a round $r$ leader vertex $v_k$, 
    then for every round \( r' \) leader vertex \( v_\ell \)  such that \( r' > r \), 
    there exists a strong path from \( v_\ell \) to $v_k$.
\end{claim}
\begin{proof}
    Since $\node{i}$ directly committed $v_k$ in round $r$, it must have either received the first messages for $2f+1$ round $r+1$ vertices that have a strong path to $v_k$ or delivered $f+1$ round $r+1$ vertices that have a strong path to $v_k$. In the former case, at least $f+1$ of those vertices originate from honest parties and will eventually be delivered due to the validity property of reliable broadcast. In either case, there exists a set $\HN$ of at least $f+1$ round $r+1$ vertices (with a strong path to $v_k$) that have been or will eventually be delivered by $\node{i}$. We complete the proof by showing that the statement holds for any $r' > r$.

    \textbf{Case} $r' = r + 1$: 
    If $v_{\ell} \in \HN$, we are trivially done. Otherwise, the vertices in $\HN$ are from round $r+1$ honest non-leader parties. Since $|\HN| \ge f+1$, by standard quorum intersection argument, there does not exist a set of $2f+1$ round $r+1$ vertices that do not have a strong path to $v_k$. Consequently, $L_{r+1}$ cannot propose a valid leader vertex by including reference to at least $2f+1$ round $r+1$ vertices without strong path to $v_k$. Moreover, honest parties that proposed vertices in $\HN$ have delivered $v_k$. By the agreement property of reliable broadcast, $L_{r+1}$ will eventually deliver $v_k$. Thus, if $v_{\ell}$ exists, there must exist a strong path from $v_{\ell}$ to $v_k$.

    \textbf{Case} $r' > r + 1$: Note that all vertices in $\HN$ will eventually be delivered by all honest parties and included in $DAG[r+1]$. Additionally, a round $r+2$ vertex has a strong path to $2f+1$ round $r+1$ vertices. By standard quorum intersection, this includes at least $1$ vertex in $\HN$ which has a strong path to $v_k$. Thus, all round $r+2$ vertices (including round $r+2$ leader vertex) have a strong path to $v_k$. Moreover, each round $r''> r+2$ vertex has strong paths to 
    at least $2f+1$ vertices in round $r''-1$. By transitivity, each vertex at round $r''$ has strong paths to at least $2f+1$ vertices in round $r+2$. This implies $v_{\ell}$ must have a strong path to $v_k$.
\end{proof}

\begin{claim}\label{claim: commit}
    If an honest party \( P_i \) directly commits a leader vertex $v_k$ in round \( r \) and
    an honest party $P_j$ directly commits a leader vertex \( v_\ell \) in round \( r'\ge r \),
    then $P_j$ (directly or indirectly) commits $v_k$ in round \( r \).
\end{claim}

\begin{proof}
If $r' = r$, by Fact~\ref{fact}, $v_k = v_\ell$ and $P_j$ directly commits $v_k$. When $r' > r$, by Claim~\ref{claim: strong path}, $v_\ell$ has a strong path to $v_k$. By the code of commit\_leader, after directly committing a leader vertex in round $r'$,  $P_j$ tries to commit a leader vertex if there exists a strong path between the two leader vertices from a smaller round until it meets the last round $r'' < r'$ that directly committed a leader vertex. If $r'' < r < r'$, $P_j$ must indirectly commit $v_k$ in round $r$. If $r < r''$, by inductive argument and~\Cref{claim: strong path}, $P_j$ must indirectly commit $v_k$ after directly committing round $r''$ leader vertex.
\end{proof}

\begin{claim} \label{claim: same order} 
    Let $v_k$ and \( v'_k \) be two leader vertices consecutively directly committed by a party
     \( P_i \) in rounds \( r_i \) and \( r'_i > r_i \) respectively. 
     Let \( v_\ell \) and \( v'_\ell \) be two leader vertices consecutively directly committed by 
     party \( P_j \) in rounds \( r_j \) and \( r'_j > r_j \) respectively. 
     Then, \( P_i \) and \( P_j \) commit the same leader vertices between rounds \( \max(r_i, r_j) \) 
     and \( \min(r'_i, r'_j) \), and in the same order.
\end{claim}

\begin{proof}
If $r'_i < r_j$ or $r'_j < r_i$, then there are no rounds between $\max(r_i, r_j)$ and $\min(r'_i, r'_j)$ and we are trivially done. Otherwise, assume without loss of generality that \( r_i \leq r_j < r_i' \leq r_j' \). By Claim~\ref{claim: commit}, both \( P_i \) and \( P_j \) will (directly or indirectly) commit the same leader vertices in round \( \min(r_i', r_j') \). Assume without loss of generality that \( \min(r'_i, r'_j) = r'_i \). By Fact~\ref{fact}, both \( {DAG}_i \) and \( {DAG}_j \) will contain \( v'_k \) and all vertices have a path from \( v'_k \) in \( {DAG}_i \). According to the commit\_leader procedure, after directly or indirectly committing the leader vertex \( v'_k \), each party will attempt to indirectly commit leader vertices from earlier rounds until reaching a round in which they have previously directly committed a leader vertex. Consequently, both \( P_i \) and \( P_j \) will indirectly commit all leader vertices from rounds \( \min(r'_i, r'_j) \) down to \( \max(r_i, r_j) \). Furthermore, due to the fixed logic of the commit\_leader code, both parties will commit the same leader vertices between rounds \( \min(r'_i, r'_j) \) and \( \max(r_i, r_j) \) in the same order.
\end{proof}

By inductively applying~\Cref{claim: same order} for every pair of honest parties, we get the following:
\begin{corollary} \label{corollary: same order} 
    Honest parties commit the same leader vertices in the same order.
\end{corollary}

\begin{lemma}[Total order] \label{lemma: total order}
    If an honest party $P_i$ outputs $a\_deliver_i$ $(b, r, P_k)$ before 
    $a\_deliver_i(b', r', P_\ell)$, then no honest party $P_j$ outputs 
    $a\_deliver_j(b', r', P_\ell)$ before $a\_deliver_j(b, r, P_k)$. 
\end{lemma}

\begin{proof}
    By~\Cref{corollary: same order}, all honest parties commit the same leader vertices in the same order. According to the logic of order\_vertices, parties process committed leader vertices in that order and a\_deliver all vertices in their causal history based on a predefined rule. By~\Cref{fact}, every honest party has an identical causal history in their DAG for each committed leader. This establishes the lemma.
\end{proof}

\begin{lemma}[Agreement] \label{lemma: agreement}
    If an honest party $\node{i}$ outputs $a\_deliver_i$ $(b,r,P_k)$, then every other honest party $P_j$ eventually outputs $a\_deliver_j$ $(b,r,P_k)$.
\end{lemma}

\begin{proof}
   Since $\node{i}$ outputs a\_deliver$_i(v_i.block, v_i.round,$ $v_i.source)$, it follows from the order\_vertices logic there exists a leader vertex $v_k$ committed either directly or indirectly by $\node{i}$ such that $v_i$ is in the causal history of $v_k$. If $v_k$ was indirectly committed, let $v_\ell$ be the vertex directly committed by $\node{i}$ that led to the commitment of $v_k$; otherwise, $v_\ell = v_k$. Let $r$ be the round of $v_\ell$. Since $v_\ell$ has been directly committed by $\node{i}$, there exists a set of at least $f+1$ round $r+1$ vertices with strong paths to $v_\ell$. By~\Cref{fact}, $\node{j}$ will eventually add those $f+1$ vertices to $DAG_j$, including all the causal histories of these vertices. Consequently, $\node{j}$ will commit $v_\ell$. By the logic of order\_vertices, $\node{j}$ will order all vertices in the causal history of $v_\ell$, which includes $v_i$. Thus, $\node{j}$ eventually outputs $a\_deliver_j(v_i.block, v_i.round,$ $v_i.source)$.
\end{proof}

\begin{lemma}[Integrity]
    For every round $r\in \mathbb{N}$ and party $P_k \in \mathcal{P}$, 
    an honest party $P_i$ outputs $a\_deliver_i(b,r,P_k)$ at most once regardless of $b$.
\end{lemma}

\begin{proof}
    An honest party $P_i$ outputs a\_deliver$_i(v.block, v.round, v.$ $source)$ 
    only when vertex $v$ is in $DAG_i$.
    Note that $v$ is added with to $P_i$'s DAG upon the reliable broadcast 
    r\_deliver$_i(v.block, v.round, v.$ $source)$ event.
    Therefore, the proof follows from the Integrity property of reliable broadcast.
\end{proof}

\mypara{Validity.} We rely on GST to prove validity. Additionally, we use RBC protocol of Bracha~\cite{bracha1987asynchronous} to establish validity. Bracha's RBC protocol operates in $3$ communication steps and ensures the RBC properties at all times. After GST, it offers the following stronger guarantees:

\begin{property} \label{property: 3Delta}
    Let $t$ be a time after GST. If an honest party reliably broadcasts a message $M$ at time $t$, then all honest parties will deliver $M$ by time $t+3\Delta$.
\end{property}

\begin{property} \label{property: 2Delta} 
    Let $t_g$ denote the GST. If an honest party delivers message M at time $t$, then all honest parties deliver M by time $\max(t_g, t)+2\Delta$.
\end{property}

\begin{claim} \label{claim: timeout} 
    Let $t_g$ denote the GST. If an honest party receives $2f+1$ $\tuple{\Timeout, r}$ messages at time $t$, then all honest parties will have received $2f+1$ $\tuple{\Timeout, r}$ messages by time \( \max(t_g, t)+2\Delta \).
\end{claim}

\begin{proof}
If an honest party $\node{i}$ receives $2f+1$ $\tuple{\Timeout, r}$ messages at time $t$, at least \( f+1 \) of these messages must have been sent by honest parties who sent the $\Timeout$ messages to all other parties. Consequently, all honest parties must have received at least \( f+1 \) $\tuple{\Timeout, r}$ messages by time \( \max(t_g, t) + \Delta \). Since, honest parties send $\tuple{\Timeout, r}$ upon receiving  \( f+1 \) $\tuple{\Timeout, r}$ messages (if they have not already done so), all honest parties will have sent their $\tuple{\Timeout, r}$ by time \( \max(t_g, t) + \Delta \). Therefore, all honest parties will receive  $2f+1$ $\tuple{\Timeout, r}$ by time \( \max(t_g, t) + 2\Delta \).
\end{proof}

\begin{claim} \label{claim: optimistic}
    Let $t_g$ denote the GST and $\node{i}$ be the first honest party to enter round $r$. If $\node{i}$ enters round $r$ at time $t$ via delivering round $r-1$ leader vertex, then all honest parties will enter round $r$ or higher by $\max(t_g, t) + 2\Delta$.
\end{claim}
\begin{proof}
    Observe that $\node{i}$ must have delivered $2f+1$ round $r-1$ vertices along with round $r-1$ leader vertex by time $t$. By~\Cref{property: 2Delta}, all honest parties must have delivered $2f+1$ round $r-1$ vertices along with round $r-1$ leader vertex by $\max(t_g, t) + 2\Delta$. Thus, all honest parties will enter round $r$ by $\max(t_g, t) + 2\Delta$ if they have not already entered a higher round.
\end{proof}

\begin{claim} \label{claim: pessimistic}
    Let $t_g$ denote the GST and $\node{i}$ be the first honest party to enter round $r$. If $\node{i}$ enters round $r$ at time $t$ by receiving $2f+1$ $\tuple{\Timeout, r-1}$, then (i) all honest parties (except $L_r$ when $\node{i} \neq L_r$) enter round $r$ or higher by $\max(t_g, t) + 2\Delta$, and (ii) $L_r$ (if honest and $P_i \neq L_r$) enters round $r$ or higher by $\max(t_g, t) + 5\Delta$.
\end{claim}
\begin{proof}
    Observe that $\node{i}$ must have delivered $2f+1$ round $r-1$ vertices and received $2f+1$ $\tuple{\Timeout, r-1}$ messages by time $t$. By~\Cref{property: 2Delta}, all honest parties must have delivered $2f+1$ round $r-1$ vertices by $\max(t_g, t) + 2\Delta$. Furthermore, by~\Cref{claim: timeout}, all honest parties must have received $2f+1$ $\tuple{\Timeout, r-1}$ by $\max(t_g, t) + 2\Delta$. Thus, all honest parties (except $L_r$) will enter round $r$ by $\max(t_g, t) + 2\Delta$ if they have not already entered a higher round. This proves part (i) of the claim.

    Observe that if no honest party delivered round $r-1$ leader vertex $v_\ell$ by $\max(t_g, t) + 2\Delta$, each honest party $\node{i}$ (excluding $L_r$) will reliably broadcast their round $r$ vertices without strong path to $v_\ell$. Thus, $L_r$ will deliver these $2f$ vertices by time $\max(t_g, t) + 5\Delta$. On the other hand, if at least one honest party delivered round $r-1$ leader vertex by $\max(t_g, t) + 2\Delta$, by~\Cref{property: 2Delta}, $L_r$ will deliver round $r-1$ leader vertex by $\max(t_g, t) + 4\Delta$. Thus, $L_r$ enters round $r$ by $\max(t_g, t) + 5\Delta$ if it has not already entered a higher round. This proves part (ii) of the claim.
\end{proof}

\begin{claim} \label{claim: keep entering}
    All honest parties keep entering higher rounds.
\end{claim}
\begin{proof}
    Suppose all honest parties are in round $r$ or higher. Let party $\node{i}$ be in round \( r \). If there exists an honest party \( P_j \) in a round \( r' > r \) at any time, then by Claim~\ref{claim: optimistic} and Claim~\ref{claim: pessimistic},  all honest parties will eventually enter round \( r' \) or higher. Otherwise, all honest parties remain in round \( r \). Note that upon entering round \( r \), all honest parties will r\_bcast a round \( r \) vertex, so each party will deliver \( 2f+1 \) round $r$ vertices.
    
    If no honest party has delivered the round \( r \) leader vertex, 
    then by the timeout rule, all honest parties will multicast \( \langle \Timeout, r \rangle \) and receive \( 2f+1 \) \( \langle \Timeout, r \rangle \) messages. Thus, all honest parties (except $L_{r+1}$ if honest) will enter round $r+1$ and reliably broadcast their round $r+1$ vertex without strong path to round $r$ leader vertex. Consequently, $L_{r+1}$ will deliver at least $2f$ round $r+1$ vertices without strong path to the round $r$ leader vertex and $L_{r+1}$ will enter round $r+1$ and propose a round $r+1$ leader vertex.
     
    On the other hand, if at least one honest party has delivered the round \( r \) leader vertex, then by Fact~\ref{fact}, all honest parties will eventually deliver this round \( r \) leader vertex as well. With \( 2f+1 \) round \( r \) vertices and the round \( r \) leader vertex delivered, all honest parties will proceed to round \( r+1 \).
\end{proof}

\begin{claim} \label{claim: honest entry}
    If an honest party enters round $r$ then at least $f + 1$ honest parties must have already entered $r - 1$.
\end{claim}

\begin{proof}
    If an honest party enters round $r$, it must have delivered $2f+1$ round $r - 1$ vertices. At least $f+1$ of these vertices are sent by honest parties while they are in round $r-1$. Thus, $f+1$ honest parties must have already entered $r - 1$.
\end{proof}

\begin{claim} \label{claim: GST strong path}
    If the first honest party to enter round $r$ does so after GST and $L_r$ is honest, 
    then there will be at least $f+1$ round $r+1$ vertices that have strong paths to the round $r$ leader vertex.
\end{claim}
\begin{proof}
    We prove the claim by considering the case where all honest parties set their 
    timeout parameter $\tau = 8\Delta$. Let $t$ be the time when the first honest party (say $\node{i}$) entered round $r$. Observe that no honest party sends $\sig{\Timeout, r}$ before $t+8\Delta$ due to its round timer expiring. Accordingly, no honest party sends $\sig{\Timeout, r}$ due to receiving $f+1$ $\sig{\Timeout, r}$ before $t+8\Delta$. Additionally, by Claim~\ref{claim: honest entry},  no honest party can enter a round greater than $r$ until at least $f + 1$ honest parties have entered \( r \). Thus, no honest party sends a $\Timeout$ message for a round greater than \( r \) and no honest party enters a higher round by receiving \( 2f+1 \) $\Timeout$ before $t + 8\Delta$.

    Since $\node{i}$ entered round $r$ at time \( t \), by Claim~\ref{claim: pessimistic}, all honest parties (except $L_r$) will enter round \( r \) or higher by \( t + 2\Delta \), and $L_r$ will do so by \( t + 5\Delta \). If any honest party enters a round higher than $r+1$ before \( t + 8\Delta \), there exist at least $2f + 1$ round $r + 1$ vertices. This is because for an honest party to enter round $r'$, it must have delivered $2f+1$ round $r'-1$ vertices. By transitive argument, there must exist $2f+1$ round $r+1$ vertices. Moreover, no honest party sends a $\sig{\Timeout, r}$ message before $t+ 8\Delta$ and does not send a round $r+1$ vertex without strong path to round $r$ leader vertex (say $v_k$) without receiving $2f+1$ $\sig{\Timeout, r}$ messages. This implies there must be at least $f+1$ round $r+1$ vertices with strong path to $v_k$.

    Also, note that if an honest party enters round $r + 1$ before $t+8\Delta$, it must have delivered $2f +1$ round $r$ vertices and vertex $v_k$ (since \( 2f+1 \) $\Timeout$ do not exist before $t+8\Delta$). Thus, its round $r + 1$ vertex must have a strong path to $v_k$.

    Now consider the case where no honest party has entered a round higher than $r$ before time $t + 8\Delta$. By Claim~\ref{claim: pessimistic}, all honest parties (except $L_r$) enter round $r$ by $t + 2\Delta$, and $L_r$ enters by $t + 5\Delta$. Upon entering round $r$, each honest party invokes r\_bcast on its round $r$ vertex. By Property~\ref{property: 3Delta}, round $r$ vertices from all honest parties (except $L_r$) are delivered by $t + 5\Delta$, and $v_k$ (from $L_r$) by $t + 8\Delta$. Thus, honest parties still in round $r$ will send round $r+1$ vertices with strong paths to $v_k$. Since advancing to round $r+1$ requires at least $f+1$ such vertices from honest parties, it follows that at least $f+1$ honest parties must send round $r+1$ vertices referencing $v_k$.
\end{proof}

When an honest party enters round $r$ via delivering round $r-1$ leader vertex, by using Claim~\ref{claim: optimistic} instead of Claim~\ref{claim: pessimistic}, we can see the claim holds with $\tau=5\Delta$.

By the commit rule and Claim~\ref{claim: GST strong path}, we get the following:

\begin{corollary} \label{corollary: validity}
    If the first honest party to enter round $r$ does so after GST and $L_r$ is honest, 
    all honest parties will directly commit round $r$ leader vertex.
\end{corollary}

\begin{lemma}[Validity.]
    If an honest party $\node{i}$ calls $a\_bcast_i(b, r)$ then every honest party eventually outputs 
    $a\_deliver_i(b, r, P_i)$.
\end{lemma}

\begin{proof}
    When an honest party \( P_i \) calls a\_bcast$_i(b, r)$, it pushes \( b \) into the \( blocksToPropose \) queue.
    By Claim~\ref{claim: keep entering}, \( P_i \) continuously progresses to higher rounds, creating new vertices in each of these rounds.
    Consequently, \( P_i \) will eventually create a vertex \( v_i \) with block \( b \) in some round \( r \) and reliably broadcast it.
    By the Validity property of reliable broadcast, all honest parties will eventually deliver \( v_i \).
    According to Fact~\ref{fact}, all honest parties will add \( v_i \) to their DAGs.
    Following the create\_new\_vertex procedure, once \( v_i \) is added to \( DAG_j[r] \),
    every subsequent vertex created by \( P_j \) will contain a path to \( v_i \).
    By Corollary~\ref{corollary: validity}, the leader vertex proposed by an honest leader is directly committed by all honest parties after GST.
    By the code of order\_vertices, each party \( P_j \) will eventually invoke a\_deliver$_j(b, r, P_i)$.
    By Lemma~\ref{lemma: agreement}, all honest parties will eventually invoke a\_deliver$(b, r, P_i)$.
\end{proof}

\begin{figure}[H]
    \centering
    \includegraphics[width=\linewidth]{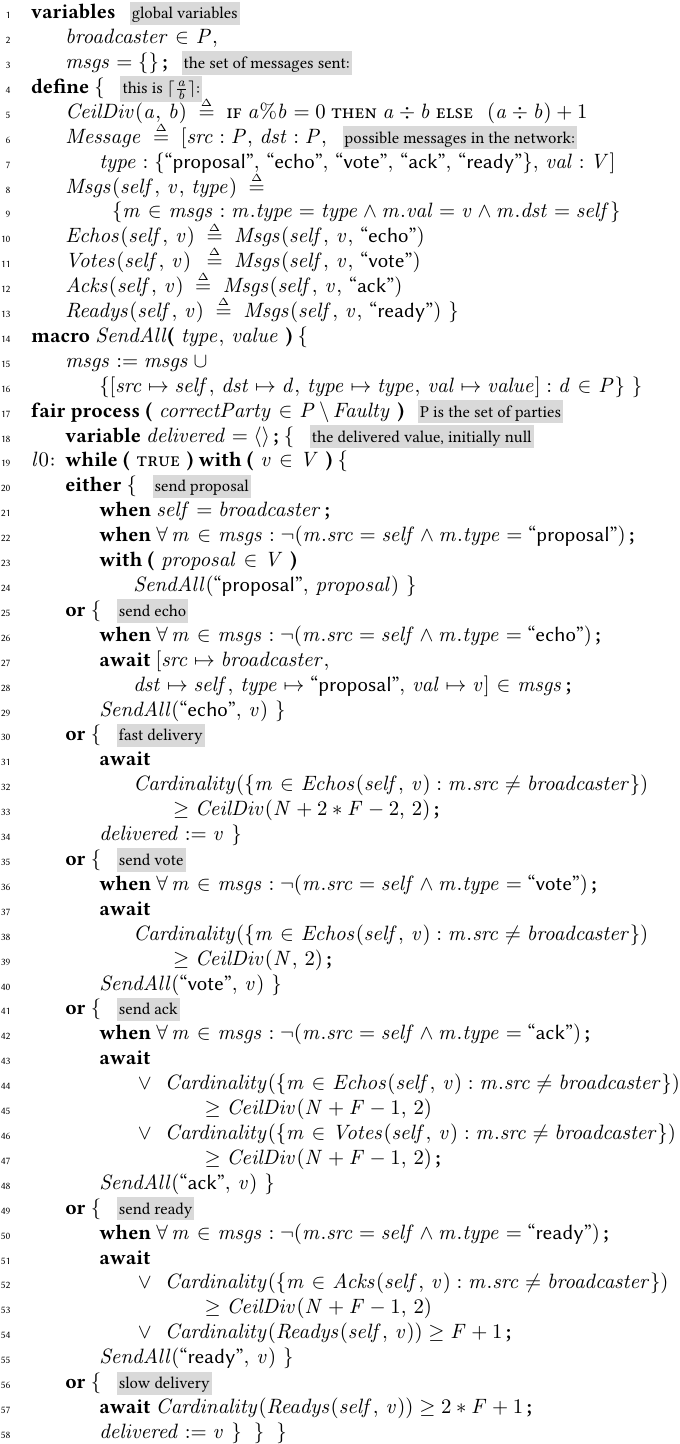}
    \caption{PlusCal/TLA+ model of the optimistic reliable-broadcast algorithm}%
    \label{fig:rbc-algo-spec}
\end{figure}

\begin{figure}[H]
    \centering
    \includegraphics[width=\linewidth]{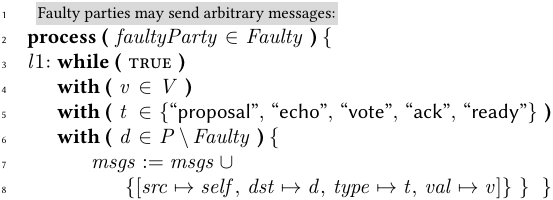}
    \caption{PlusCal/TLA+ model of Byzantine parties in the optimistic reliable-broadcast algorithm}%
    \label{fig:rbc-byzantine-spec}
\end{figure}

\begin{figure}[H]
    \centering
    \includegraphics[width=\linewidth]{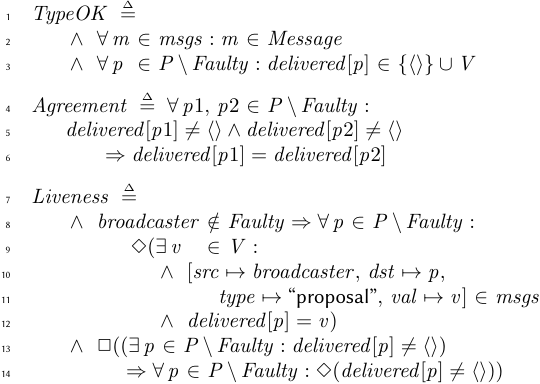}
    \caption{TLA+ type invariant and correctness properties of the optimistic reliable-broadcast algorithm}%
    \label{fig:rbc-properties-spec}
\end{figure}

\section{Formal Specifications}%
\label{sec:formal-specs}

\subsection{Optimistic Reliable Broadcast}%
\label{sec:bcast-spec}

\subsubsection{PlusCal/TLA+ Specification}

A PlusCal/TLA+ specification of the optimistic reliable-broadcast algorithm of~\Cref{sec:opt_rbc} appears in~\Cref{fig:rbc-algo-spec}, the Byzantine-party model in~\Cref{fig:rbc-byzantine-spec}, and the type invariant and correctness predicates in~\Cref{fig:rbc-properties-spec}.
The specification is written in the PlusCal/TLA+ languages~\cite{lamportPlusCalAlgorithmLanguage2009,lamportSpecifyingSystemsTLA2002}.
It models sent protocol messages as a set of records and gives each non-faulty party a fair PlusCal process.
The correct-party process follows the proposal, echo, vote, ack, ready, optimistic-delivery, and slow-delivery rules; the Byzantine process can inject arbitrary protocol messages to non-faulty parties.
The safety predicates state that delivered values agree among non-faulty parties; the liveness predicate states validity for a non-faulty broadcaster and totality once any non-faulty party delivers.

The artifact includes a TLC configuration for this module that instantiates a small finite system and checks the basic type invariant with the TLC model-checker~\cite{yu_model_1999}.
The agreement, ready-agreement, and liveness predicates are included as named TLA+ formulas, but the unbounded correctness argument for optimistic reliable broadcast does not rely on bounded TLC exploration.
Instead, the safety and liveness proof described next is mechanically checked using Ivy and Isabelle/HOL.

\subsubsection{Ivy and Isabelle/HOL Proofs}\label{sec:ivy-proof}
To make sure that the optimistic RBC protocol satisfies the safety and liveness properties of Reliable Broadcast, we developed a mechanically-checked proof of correctness.
This proof consists of two artifacts.

First, we developed a formal specification and correctness proof of the protocol in the Ivy language.
The full Ivy model appears in~\Cref{sec:ivy-specification}.
It declares abstract optimistic, majority, classic, amplification, and commit quorum types and assumes the quorum-intersection properties needed by the protocol.
The Ivy transition system mirrors the protocol actions for proposing, echoing, voting, acknowledging, becoming ready, optimistic committing, normal committing, and Byzantine behavior.
The proof consists of inductive invariants that imply agreement, together with liveness checks for validity and totality.
For validity, we check that, if the broadcaster is honest and no honest party can take a step anymore, then every honest party has delivered. 
For totality, we check that, if an honest party has delivered and no honest party can take a step anymore, then every honest party has delivered. 
Thanks to the axiomatic model, the verification conditions fall in the EPR fragment of first-order logic and are discharged automatically by Ivy.

Second, we developed an Isabelle/HOL proof showing that the axiomatic model is a sound abstraction of the concrete threshold definitions used in the protocol.
The Isabelle development defines the optimistic, majority, classic, amplification, and commit quorum types using explicit cardinality thresholds under the assumptions \(3f<n\) and \(n\geq 2\), and proves that these concrete quorum types satisfy all the intersection properties assumed in the Ivy model.

Together, the Ivy and Isabelle/HOL artifacts cover arbitrary finite party sets satisfying the resilience bound, and the Ivy invariants cover all executions of the abstract transition system regardless of their length.
Methodologically, this follows the same high-level separation of concerns as the threshold-based verification approach of Berkovits et al.~\cite{berkovitsVerificationThresholdBasedDistributed2019}: the protocol proof is carried out in EPR over abstract quorum-intersection assumptions, while the arithmetic content of the thresholds is justified separately.
The difference is that their flow expresses and checks the required threshold-intersection properties using the TIP/BAPA decomposition, whereas our second stage proves the corresponding concrete quorum facts in Isabelle/HOL for the quorum families used by optimistic RBC.

\clearpage
\onecolumn
\subsubsection{Ivy Specification}\label{sec:ivy-specification}

The following listing gives the full Ivy specification and proof script used for the optimistic reliable-broadcast proof.

\begin{lstlisting}[style=ivyspec]
 
#lang ivy1.7

################################################################################
#
# Optimistic Reliable Broadcast (OptiRBC)
#
# Liveness Strategy: the state-space is finite and non-faulty actions are
# monotonic, so eventually no more non-faulty actions are enabled.  At that point
# we check all liveness conditions hold.
#
################################################################################

################################################################################
# Types
################################################################################

type party
type val

# Quorum types (abstract witnesses for threshold sets of parties)
type opt_quorum
type maj_quorum
type classic_quorum
type amplification_quorum
type commit_quorum

################################################################################
# Immutable state (theory parameters)
################################################################################

individual broadcaster : party
relation faulty(P:party)

relation opt_quorum_member(Q:opt_quorum, P:party)
relation maj_quorum_member(Q:maj_quorum, P:party)
relation classic_quorum_member(Q:classic_quorum, P:party)
relation amplification_quorum_member(Q:amplification_quorum, P:party)
relation commit_quorum_member(Q:commit_quorum, P:party)

################################################################################
# Mutable state
################################################################################

relation proposal(V:val)
relation echo(P:party, V:val)
relation vote(P:party, V:val)
relation ack(P:party, V:val)
relation ready(P:party, V:val)
relation committed(P:party, V:val)

individual pv : val  # proposal of the non-faulty broadcaster (if any)

################################################################################
# Abstract domain model assumptions (quorum intersection properties)
################################################################################

# Classic, majority, and optimistic quorums do not contain the broadcaster.
axiom [broadcaster_not_in_classic_quorum]
    forall Q. ~classic_quorum_member(Q, broadcaster)

axiom [broadcaster_not_in_maj_quorum]
    forall Q. ~maj_quorum_member(Q, broadcaster)

axiom [broadcaster_not_in_opt_quorum]
    forall Q. ~opt_quorum_member(Q, broadcaster)

# The set of non-faulty parties satisfies all the quorum thresholds
# except the optimistic quorum threshold.

axiom [non_faulty_commit_quorum]
    exists Q. forall P. commit_quorum_member(Q, P) -> ~faulty(P)

axiom [non_faulty_amplification_quorum]
    exists Q. forall P. amplification_quorum_member(Q, P) -> ~faulty(P)

axiom [non_faulty_maj_quorum]
    exists Q. forall P. maj_quorum_member(Q, P) -> ~faulty(P)

axiom [non_faulty_classic_quorum]
    exists Q. forall P. classic_quorum_member(Q, P) -> ~faulty(P)

# An optimistic quorum contains a non-faulty member.
axiom [opt_quorum_has_nonfaulty_member]
    forall Q. exists P. ~faulty(P) & opt_quorum_member(Q, P)

# When the broadcaster is faulty, any two optimistic quorums share a
# non-faulty member.
axiom [opt_quorum_intersection]
    forall Q1, Q2.
        faulty(broadcaster) ->
            exists P. ~faulty(P) & opt_quorum_member(Q1, P) & opt_quorum_member(Q2, P)

# When the broadcaster is faulty, an optimistic quorum and a majority
# quorum share a non-faulty member.
axiom [opt_maj_quorum_intersection]
    forall Q1, Q2.
        faulty(broadcaster) ->
            exists P. ~faulty(P) & opt_quorum_member(Q1, P) & maj_quorum_member(Q2, P)

# When the broadcaster is faulty, an optimistic quorum and a classic quorum
# share a non-faulty member.
axiom [opt_quorum_intersects_classic_quorum]
    forall Q1, Q2.
        faulty(broadcaster) ->
            exists P. ~faulty(P) & opt_quorum_member(Q1, P) & classic_quorum_member(Q2, P)

# When the broadcaster is faulty, any two classic quorums share a non-faulty member.
axiom [classic_quorum_intersection]
    forall Q1, Q2.
        faulty(broadcaster) ->
            exists P. ~faulty(P) & classic_quorum_member(Q1, P) & classic_quorum_member(Q2, P)

# Every classic quorum contains a non-faulty member.
axiom [quorum_not_faulty]
    forall Q. exists P. ~faulty(P) & classic_quorum_member(Q, P)

# Every majority quorum contains a non-faulty member.
axiom [maj_quorum_has_nonfaulty_member]
    forall Q. exists P. ~faulty(P) & maj_quorum_member(Q, P)

# Every amplification quorum contains a non-faulty member.
axiom [amplification_quorum_has_nonfaulty_member]
    forall Q. exists P. ~faulty(P) & amplification_quorum_member(Q, P)

# When the broadcaster is faulty, for any optimistic quorum there exists a
# majority quorum whose members are all non-faulty and in the optimistic quorum.
axiom [opt_quorum_contains_maj_quorum]
    forall Q.
        faulty(broadcaster) ->
            exists Q2. forall P.
                maj_quorum_member(Q2, P) -> ~faulty(P) & opt_quorum_member(Q, P)

# For any commit quorum there exists an amplification quorum whose members
# are all non-faulty and in the commit quorum.
axiom [commit_quorum_contains_non_faulty_amplification_quorum]
    forall Q.
        exists Q2. forall P.
            amplification_quorum_member(Q2, P) -> ~faulty(P) & commit_quorum_member(Q, P)

################################################################################
# Initialization
################################################################################

after init {
    proposal(V) := false;
    echo(P, V) := false;
    vote(P, V) := false;
    ack(P, V) := false;
    ready(P, V) := false;
    committed(P, V) := false;
}

################################################################################
# Protocol actions
################################################################################

# The broadcaster proposes.
action propose = {
    require forall V. ~proposal(V);
    proposal(pv) := true;
}

# A non-broadcaster party echoes a proposed value (at most one echo per party).
action echo_step(p:party, v:val) = {
    require p ~= broadcaster;
    require forall V. ~echo(p, V);
    require proposal(v);
    echo(p, v) := true;
}

# A non-broadcaster party votes for a value after seeing a majority-quorum
# of echoes (unless it has already voted for some other value).
action vote_step(p:party, v:val) = {
    require p ~= broadcaster;
    require forall V. ~vote(p, V);
    require exists Q. forall P. maj_quorum_member(Q, P) -> echo(P, v);
    vote(p, v) := true;
}

# A non-broadcaster party sends an ack after seeing either a quorum of votes
# or a quorum of echoes (unless it has already acked some other value).
action ack_step(p:party, v:val) = {
    require p ~= broadcaster;
    require forall V. ~ack(p, V);
    require (exists Q. forall P. classic_quorum_member(Q, P) -> vote(P, v))
          | (exists Q. forall P. classic_quorum_member(Q, P) -> echo(P, v));
    ack(p, v) := true;
}

# A party becomes ready after seeing either a quorum of acks
# or an amplification-quorum of readys (unless it has already become ready
# for some other value).
action ready_step(p:party, v:val) = {
    require forall V. ~ready(p, V);
    require (exists Q. forall P. classic_quorum_member(Q, P) -> ack(P, v))
          | (exists Q. forall P. amplification_quorum_member(Q, P) -> ready(P, v));
    ready(p, v) := true;
}

# Optimistic commit: a non-broadcaster party commits after seeing an
# optimistic-quorum of echoes (unless it has already committed some other value).
action opt_commit_step(p:party, v:val) = {
    require forall V. ~committed(p, V);
    require exists Q. forall P. opt_quorum_member(Q, P) -> echo(P, v);
    committed(p, v) := true;
}

# Normal commit: a party commits after seeing a commit-quorum of readys
# (unless it has already committed some other value).
action commit_step(p:party, v:val) = {
    require forall V. ~committed(p, V);
    require exists Q. forall P. commit_quorum_member(Q, P) -> ready(P, v);
    committed(p, v) := true;
}

# Byzantine step: a faulty broadcaster can set proposal arbitrarily.
action byz_proposal = {
    require faulty(broadcaster);
    proposal(V) := *;
}

# Byzantine step: a faulty party can set its echo/vote/ack/ready arbitrarily.
action byz_party(p:party) = {
    require faulty(p);
    echo(p, V) := *;
    vote(p, V) := *;
    ack(p, V) := *;
    ready(p, V) := *;
}

################################################################################
# Ghost enabledness relations
################################################################################

relation propose_enabled
definition propose_enabled =
    ~faulty(broadcaster) & forall V. ~proposal(V)

relation echo_enabled(P:party, V:val)
definition echo_enabled(P, V) =
    ~faulty(P) & P ~= broadcaster & ~faulty(broadcaster) & (forall V2. ~echo(P, V2)) & proposal(V)

relation vote_enabled(P:party, V:val)
definition vote_enabled(P, V) =
    ~faulty(P) & P ~= broadcaster & (forall V2. ~vote(P, V2))
    & (exists Q. forall P2. maj_quorum_member(Q, P2) -> ~faulty(P2) & echo(P2, V))

relation ack_enabled(P:party, V:val)
definition ack_enabled(P, V) =
    ~faulty(P) & P ~= broadcaster & (forall V2. ~ack(P, V2))
    & ((exists Q. forall P2. classic_quorum_member(Q, P2) -> ~faulty(P2) & vote(P2, V))
       | (exists Q. forall P2. classic_quorum_member(Q, P2) -> ~faulty(P2) & echo(P2, V)))

relation ready_enabled(P:party, V:val)
definition ready_enabled(P, V) =
    ~faulty(P) & (forall V2. ~ready(P, V2))
    & ((exists Q. forall P2. classic_quorum_member(Q, P2) -> ~faulty(P2) & ack(P2, V))
       | (exists Q. forall P2. amplification_quorum_member(Q, P2) -> ~faulty(P2) & ready(P2, V)))

relation opt_commit_enabled(P:party, V:val)
definition opt_commit_enabled(P, V) =
    ~faulty(P) & P ~= broadcaster & (forall V2. ~committed(P, V2))
    & (exists Q. forall P2. opt_quorum_member(Q, P2) -> ~faulty(P2) & echo(P2, V))

relation commit_enabled(P:party, V:val)
definition commit_enabled(P, V) =
    ~faulty(P) & (forall V2. ~committed(P, V2))
    & (exists Q. forall P2. commit_quorum_member(Q, P2) -> ~faulty(P2) & ready(P2, V))

################################################################################
# Liveness properties
################################################################################

# Validity: if the broadcaster is non-faulty and no non-faulty actions are
# enabled, then every non-faulty party has committed pv.
action check_validity(p:party) = {
    require ~faulty(broadcaster);
    require ~faulty(p);
    # no action is enabled:
    require ~propose_enabled;
    require forall P, V. ~echo_enabled(P, V);
    require forall P, V. ~vote_enabled(P, V);
    require forall P, V. ~ack_enabled(P, V);
    require forall P, V. ~ready_enabled(P, V);
    require forall P, V. ~commit_enabled(P, V);
    # p must have committed the broadcaster's value:
    assert committed(p, pv);
}

# Totality: if two non-faulty parties exist,
# p1 has committed pv, and no non-faulty actions are enabled, then p2 has
# also committed pv.
action check_totality(p1:party, p2:party, v:val) = {
    require ~faulty(p1) & ~faulty(p2);
    require committed(p1, v);
    # no action is enabled:
    require ~propose_enabled;
    require forall P, V. ~echo_enabled(P, V);
    require forall P, V. ~vote_enabled(P, V);
    require forall P, V. ~ack_enabled(P, V);
    require forall P, V. ~ready_enabled(P, V);
    require forall P, V. ~commit_enabled(P, V);
    assert committed(p2, v);
}

export propose
export echo_step
export vote_step
export ack_step
export ready_step
export opt_commit_step
export commit_step
export byz_proposal
export byz_party
export check_validity
export check_totality


################################################################################
# Invariants
################################################################################

# inv0: non-faulty parties never echo, vote, ack, ready, or commit two different values
invariant [inv0]
    forall P, V1, V2. ~faulty(P) & V1 ~= V2 ->
        ~(echo(P, V1) & echo(P, V2))
        & ~(vote(P, V1) & vote(P, V2))
        & ~(ack(P, V1) & ack(P, V2))
        & ~(ready(P, V1) & ready(P, V2))
        & ~(committed(P, V1) & committed(P, V2))

# inv1: If a non-faulty party commits value v, then either
#   (a) an amplification quorum of non-faulty parties are ready for v, or
#   (b) the non-faulty members of some optimistic quorum have echoed v.
invariant [inv1]
    forall P, V. ~faulty(P) & committed(P, V) ->
        (exists Q. forall P2. amplification_quorum_member(Q, P2) -> ~faulty(P2) & ready(P2, V))
        | (exists Q. forall P2. opt_quorum_member(Q, P2) & ~faulty(P2) -> echo(P2, V))

# inv2: If a non-faulty party is ready for v, then the non-faulty members of
# some classic quorum have acked v.
invariant [inv2]
    forall P, V. ~faulty(P) & ready(P, V) ->
        exists Q. forall P2. classic_quorum_member(Q, P2) & ~faulty(P2) -> ack(P2, V)

# inv3: If a non-faulty party has acked v, then either
#   (a) the non-faulty members of some classic quorum have all voted for v, or
#   (b) the non-faulty members of some classic quorum have all echoed v.
invariant [inv3]
    forall P, V. ~faulty(P) & ack(P, V) ->
        (exists Q. forall P2. classic_quorum_member(Q, P2) & ~faulty(P2) -> vote(P2, V))
        | (exists Q. forall P2. classic_quorum_member(Q, P2) & ~faulty(P2) -> echo(P2, V))

# inv5: if a non-faulty party votes for v, then the non-faulty members of
# some majority quorum have echoed v.
invariant [inv5]
    forall P, V. ~faulty(P) & vote(P, V) ->
        exists Q. forall P2. maj_quorum_member(Q, P2) & ~faulty(P2) -> echo(P2, V)

# Agreement: no two non-faulty parties commit different values.
invariant [agreement]
    forall P1, P2, V1, V2.
        ~faulty(P1) & ~faulty(P2) & committed(P1, V1) & committed(P2, V2) -> V1 = V2

# inv7: if the broadcaster is not faulty, then it makes at most one proposal
# and non-faulty parties only echo, vote, ack, ready, or commit the value
# proposed by the broadcaster.
invariant [inv7] ~faulty(broadcaster) ->
    (forall V. proposal(V) -> V = pv)
    & (forall P, V. ~faulty(P) ->
        (echo(P, V) -> proposal(V))
        & (vote(P, V) -> proposal(V))
        & (ack(P, V) -> proposal(V))
        & (ready(P, V) -> proposal(V))
        & (committed(P, V) -> proposal(V)))

# inv4: no two non-faulty parties are ready for different values.
invariant [inv4]
    forall P1, P2, V1, V2.
        ~faulty(P1) & ~faulty(P2) & ready(P1, V1) & ready(P2, V2) -> V1 = V2

# inv6: if an optimistic quorum has echoed v, then no non-faulty party votes
# for a different value.
invariant [inv6]
    forall V.
        (exists Q. forall P. opt_quorum_member(Q, P) & ~faulty(P) -> echo(P, V)) ->
            (forall P, V2. ~faulty(P) & vote(P, V2) -> V2 = V)

# if an optimistic quorum echoes some value, then no non-faulty party acks a different value
invariant [inv8] (forall P2. opt_quorum_member(Q, P2) & ~faulty(P2)
    -> echo(P2, V2)) -> (forall P1, V1 . ~faulty(P1) & V1 ~= V2 -> ~ack(P1,V1))

# if the broadcaster is not faulty and a non-faulty party commits, then the broadcaster have proposed
invariant [inv9] ~faulty(broadcaster) & ~faulty(P) & committed(P, V) -> proposal(pv)
\end{lstlisting}
\clearpage

\subsubsection{Isabelle/HOL Theories}\label{sec:isabelle-theories}

The following pages give the Isabelle-generated, typeset version of the Isabelle/HOL theories used to prove the soundness of the axiomatic domain model.

\includepdf[pages=-,pagecommand={}]{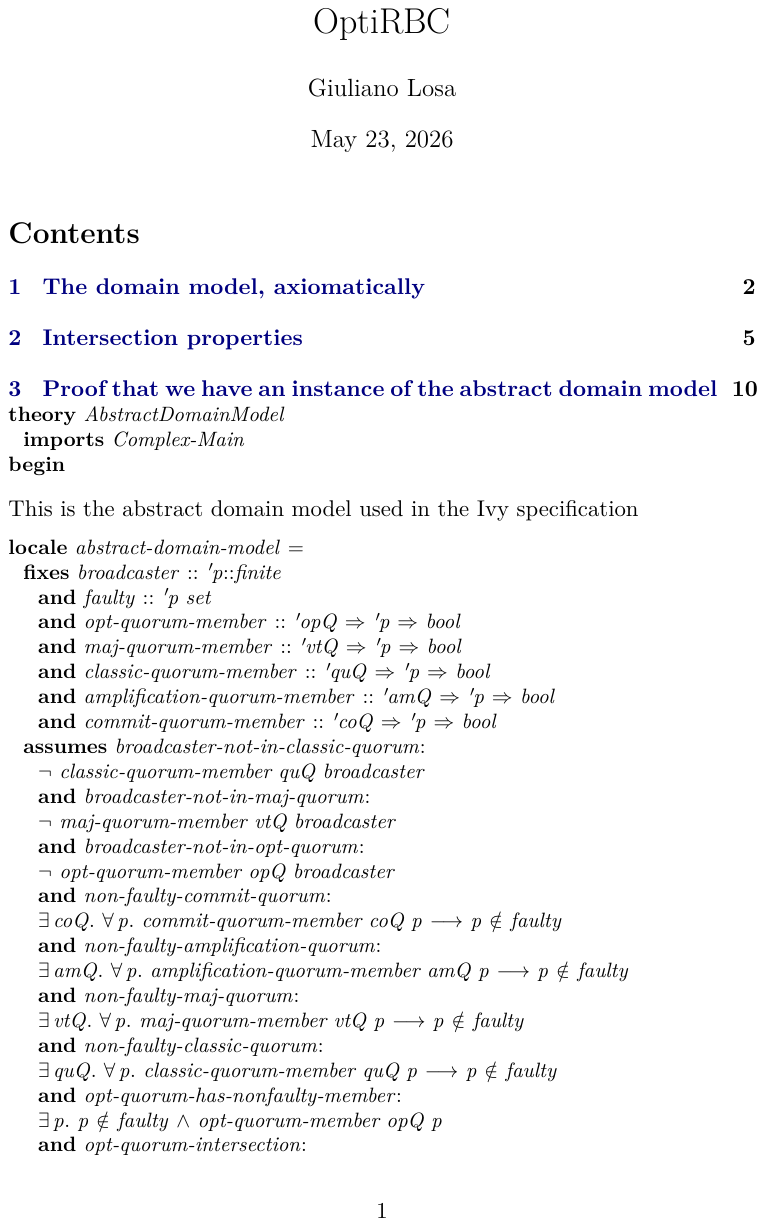}
\clearpage
\twocolumn

\subsection{\name}%
\label{sec:sailfish-spec}

A formal specification of \name in the PlusCal/TLA+ language appears in~\Cref{fig:sailfish-spec}, with the Byzantine-node model in~\Cref{fig:sailfish-byzantine-spec}, supporting definitions in~\Cref{fig:digraph-spec,fig:blockdag-spec}, and the type invariant and correctness properties in~\Cref{fig:correctness-spec}.
At this abstraction level, the same module models Sailfish and \name: the differences between the signature-based and signature-free protocols are hidden behind the assumptions imposed on admissible DAG vertices.

We make a number of abstractions and simplifying assumptions in order to expose the high-level principles of the algorithm succinctly and to keep model-checking with TLC tractable:
\begin{enumerate}[noitemsep,leftmargin=*]
    \item We do not model timers and timeout messages.
        Instead, eventual synchrony is modeled by requiring that, after GST, a correct node entering round \(r\) has all correct vertices from round \(r-1\).
        The liveness predicate then checks that every correct leader at or after GST is included in the log of each correct node that advances at least two more rounds.
    \item We restrict the model to one leader per round, represented by the operator \(Leader\).
    \item DAG vertices are represented only by their creator and round.
        The model abstracts away payloads, signatures, and certificates, which, at this level of abstraction, are irrelevant to the agreement and leader-commit properties being checked.
    \item Nodes do not explicitly use RBC to disseminate their DAG vertices; instead, accepted vertices are added to a global DAG data structure.
        This does not give correct nodes a single global view: each correct node advances to a new round using an arbitrary quorum of vertices from the previous round.
    \item The Byzantine-node process cannot equivocate because RBC would prevent two vertices from the same Byzantine node in the same round.
        After round \(1\), a Byzantine vertex is accepted only if it references a quorum of vertices from the previous round.
    \item We do not model weak edges or the optimization that commits based on the first \(n-f\) RBC deliveries.
        The \(Linearize\) operator uses a deterministic arbitrary ordering of the relevant causal past, which is sufficient for checking agreement.
\end{enumerate}

The artifact contains two bounded TLC harnesses for the Sailfish specification.
The first uses three nodes, one Byzantine node, rounds \(1\ldots 5\), and \(GST=3\); it checks the type invariant, agreement, and the liveness property in~\Cref{fig:correctness-spec}.
The second uses four nodes, one Byzantine node, rounds \(1\ldots 5\), quorums of size at least three, blocking sets of size at least two, and \(GST=6\); with GST beyond the checked round bound, this configuration is a larger safety-oriented check of the same model.

\newpage
\begin{figure}[H]
    \centering
    \includegraphics[width=\linewidth]{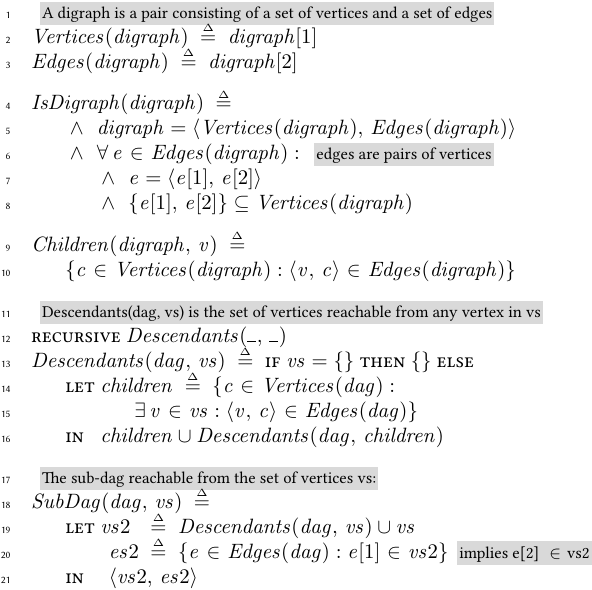}
    \caption{TLA+ formalization of directed-graph notions needed to formally specify \name}
    \label{fig:digraph-spec}
\end{figure}
\begin{figure}[H]
    \centering
    \includegraphics[width=\linewidth]{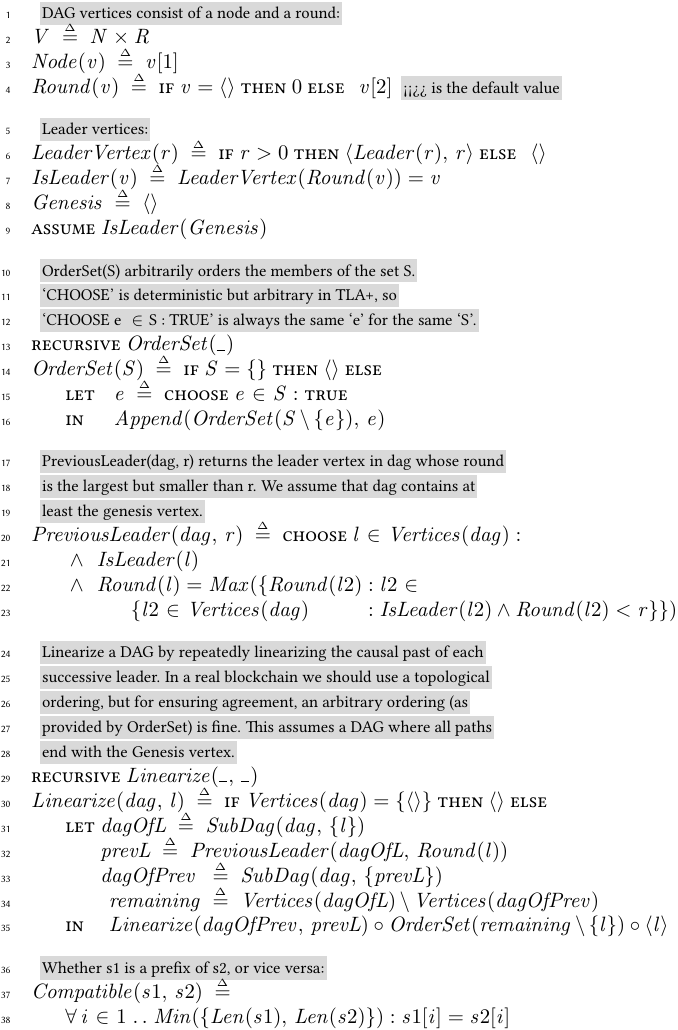}
    \caption{TLA+ formalization of block-DAG notions, including the DAG-ordering function (\textsf{Linearize}) used in \name}
    \label{fig:blockdag-spec}
\end{figure}

\begin{figure}[H]
    \centering
    \includegraphics[width=\linewidth]{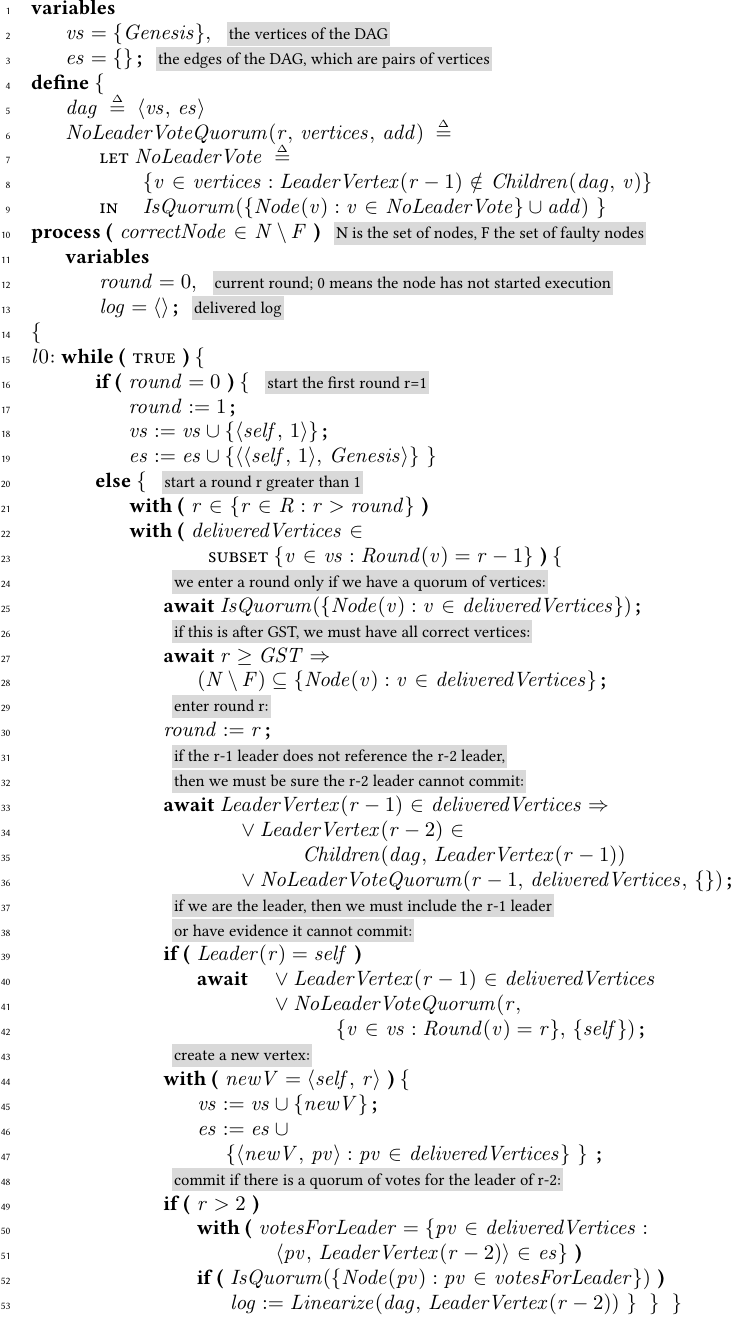}
    \caption{High-level specification of the \name algorithm, formalized in PlusCal/TLA+}
    \label{fig:sailfish-spec}
\end{figure}

\begin{figure}[H]
    \centering
    \includegraphics[width=\linewidth]{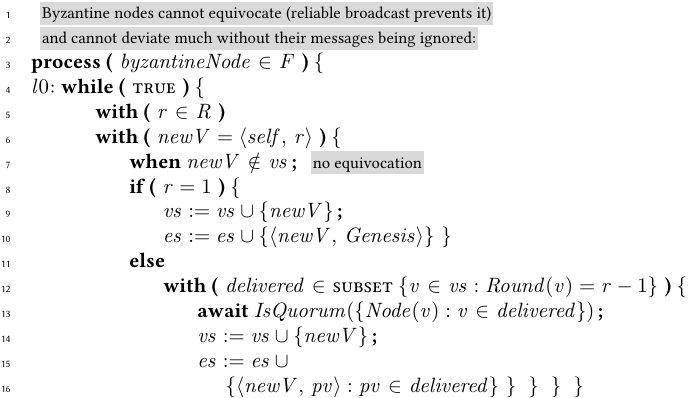}
    \caption{PlusCal/TLA+ model of Byzantine nodes in \name}
    \label{fig:sailfish-byzantine-spec}
\end{figure}

\begin{figure}[H]
    \centering
    \includegraphics[width=\linewidth]{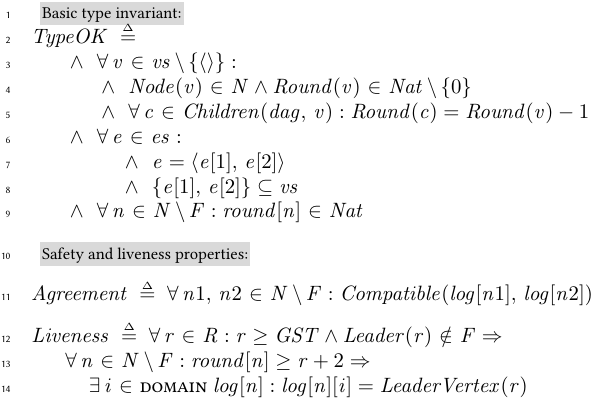}
    \caption{TLA+ type invariant and correctness properties of \name}
    \label{fig:correctness-spec}
\end{figure}

\end{document}